\documentclass[a4paper,onecolumn, allowfontchangeintitle]{quantumarticle}
\pdfoutput=1
\usepackage{graphicx} % Required for inserting images
\usepackage{mathtools} 
\usepackage{ragged2e}
\usepackage{qcircuit}
\usepackage{mathrsfs}
\usepackage{lipsum, babel}
\usepackage{physics}
\usepackage{graphicx}% Include figure files
\usepackage{dcolumn}% Align table columns on decimal point
\usepackage{bm}% bold math
\usepackage{subfigure}
\usepackage[ruled, linesnumbered,lined,boxed,commentsnumbered]{algorithm2e}

\usepackage{amssymb}
\usepackage{xcolor}
\usepackage{dsfont}
\usepackage{amsthm}
\newtheorem{theorem}{Theorem}
\usepackage{graphicx}

\usepackage{amsmath}               
{

}

\usepackage{algpseudocode}
\newtheorem{lemma}[theorem]{Lemma}
\newtheorem{conj}[theorem]{Conjecture}
\newtheorem{corr}[theorem]{Corollary}

\usepackage{amsfonts}

\usepackage[colorlinks=true]{hyperref}

\usepackage[toc,page,header]{appendix}
\usepackage{minitoc}

\usepackage{tcolorbox}
\newtcolorbox{highlight}{colback=yellow!30,colframe=yellow!80!black}

\usepackage[utf8]{inputenc}

\usepackage{color,soul}

\title{\fontsize{20}{15}\selectfont Mitigating photon loss in linear optical quantum circuits}

\author{James Mills}
\affiliation{Quandela, 7 Rue Léonard de Vinci, 91300 Massy, France}
\affiliation{School of Informatics, University of Edinburgh, United Kingdom}
\author{Rawad Mezher}
\affiliation{Quandela, 7 Rue Léonard de Vinci, 91300 Massy, France}
\begin{document}
\doparttoc % Tell to minitoc to generate a toc for the parts
\faketableofcontents % Run a fake tableofcontents command for the partocs

\begin{abstract}
\sloppy
Photon loss rates set an effective upper limit on the size of computations that can be run on current linear optical quantum devices. We present a family of techniques designed to mitigate the effects of photon loss on both output probabilities and expectation values derived from noisy linear optical circuits composed of an input of $n$ photons, an $m$--mode interferometer, and $m$ single photon detectors. 
Central to these techniques is the construction  \emph{recycled probabilities}. Recycled probabilities are constructed from output statistics affected by loss, and are designed to amplify the signal of the ideal (lossless) probabilities. 
Classical postprocessing techniques then take recycled probabilities as input and output a set of loss-mitigated probabilities, or expectation values.
Our postprocessing methods result in biased estimators of the lossless probabilities. Nevertheless, we provide both analytical and numerical evidence that these methods can be applied, up to large sample sizes, to produce output probabilities with lower combined bias and statistical errors than the statistical errors of the output probabilities  obtained from postselection.
Therefore, these methods can outperform postselection - currently the standard method of coping with photon loss when sampling from discrete variable linear optical quantum circuits. 
In contrast, we provide evidence that the popular zero-noise extrapolation technique cannot improve on the performance of postselection for any photon loss rate. 
\end{abstract}
\maketitle

\section{Introduction}

\sloppy

Discrete variable linear optical quantum computing (DVLOQC) is a framework that uses a discrete number $n$ of photons as well as $m$-mode linear optical interferometers  to store and process quantum information. 
Many models of universal and fault-tolerant quantum computation tailored to this framework have been developed over the  years beginning with the work of \cite{knill_scheme_2001}, and followed by various other models and variants (e.g \cite{raussendorf_topological_2007,de_gliniasty_spin-optical_2023,bartolucci_fusion-based_2023,raussendorf_one-way_2001}). 
Furthermore, promising proposals for the near-term demonstration of quantum-over-classical advantage such as boson sampling \cite{aaronson_computational_2011} can naturally be implemented in this framework. 
A quantum device capable of performing DVLOQC usually consists of three components: single photon sources \cite{senellart_high-performance_2017}, multimode interferometers \cite{reck_experimental_1994}, and single photon detectors \cite{hadfield_single-photon_2009}. 
We will refer to the collection of these as a \emph{linear optical circuit}.
One major impediment to scaling up such a device is photon loss \cite{wang_toward_2018}.  Postselection, where all \emph{lossy} output statistics with one or more lost photons are discarded and only those where all photons are detected are kept, may be used to obtain the ideal output distribution of a  linear optical circuit subject to photon loss \cite{brod_photonic_2019}. 
This can be viewed as a form of quantum error mitigation, where the ideal distribution is accessible, but with a sampling cost scaling exponentially with the depth of the  circuit. 
This scaling comes from the fact that the probability of at least one photon being lost approaches one exponentially  quickly with increasing  circuit depth \cite{renema_classical_2019}, which results in most outputs being discarded.

Many techniques have been developed to improve the performance of quantum computations run on currently available quantum hardware. These  are generally referred to as \emph{quantum error mitigation} (QEM) techniques. Some well-known QEM techniques are zero noise extrapolation (ZNE) \cite{temme_error_2017, endo_practical_2018, giurgica-tiron_digital_2020, he_zero-noise_2020}, probabilistic error cancellation \cite{temme_error_2017, strikis_learning-based_2021, mari_extending_2021, van_den_berg_probabilistic_2023}, verification-based mitigation \cite{bonet-monroig_low-cost_2018, sagastizabal_experimental_2019-1, obrien_error_2021, mezher_mitigating_2022}, virtual distillation and exponential error suppression \cite{koczor_exponential_2021, koczor_dominant_2021, huggins_virtual_2021, yamamoto_error-mitigated_2022}, quantum subspace expansion mitigation \cite{mcclean_decoding_2020,yoshioka_generalized_2022, yang_dual-gse_2023, ohkura_leveraging_2023}, and measurement error mitigation \cite{maciejewski_mitigation_2020, arrasmith_development_2023, mills_simplifying_2023, bravyi_mitigating_2021} (see \cite{endo_hybrid_2021} for a review of QEM techniques). These techniques are tailored to the circuit model of quantum computing, and thus adapting them to DVLOQC can in some cases be complicated by the considerable differences in the computational setup and the types of noise.

In this work we address the question of whether the effects of photon loss in the output distributions of linear optical quantum circuits can be mitigated by classical postprocessing of lossy output statistics. 
Crucially, we require that these postprocessing techniques outperform postselection, in the sense that they provide more converged estimates of the ideal probabilities for a comparable sample number.
We present various new techniques to mitigate photon loss in linear optical circuits, and perform rigorous analysis of their performance. 
We refer to these techniques collectively as \emph{recycling mitigation}, as they all involve the use of lossy output statistics that otherwise would be discarded.
At the heart of recycling mitigation is the construction of the so-called \emph{recycled probabilities}. 
The recycled probabilities can be thought of as precursors of the mitigated outputs. Mitigated probabilities are approximations of the ideal probabilities, and can be obtained from the recycled probabilities by appropriate classical postprocessing. 
We introduce several techniques by which this classical postprocessing step can be performed. These give rise to different photon loss mitigation techniques. 
Fig. \ref{fig:PRschematic} shows, at a very high level,  the main steps underlying our mitigation techniques. 
We provide analytical and numerical evidence that there exists a threshold value, lower bounded by a constant, for the loss per mode $\eta$, denoted as $\eta_{th}$,
such that when $\eta \geq \eta_{th}$, recycling mitigation outperforms postselection. 

Furthermore, we provide analytic and numerical evidence showing that mitigation techniques based on artificially increasing noise and Richardson extrapolation, called zero-noise extrapolation (ZNE) \cite{temme_error_2017, li_efficient_2017}, such as those applied to mitigate photon loss in the continuous variable regime \cite{su_error_2021}, do not outperform postselection for the problem of mitigating photon loss in the DVLOQC setting. 
This result is particularly interesting because we also show that, in the case of photon loss, ZNE produces  unbiased estimators of the ideal probability. 
Indeed, we show that the application of ZNE to mitigate photon loss reduces to the problem of inverting a Vandermonde matrix \cite{gautschi_inverses_1978} which, in the absence of statistical error, perfectly computes the ideal probabilities. 
In the presence of statistical error, however, we show that the resulting error on the inversion process is higher than the statistical error of postselection. 
The requirement of outperforming postselection, which is itself an unbiased estimator of the ideal probability, is what motivated us to search for biased loss mitigation techniques with smaller combined bias and statistical error than postselection.

Recycling mitigation is a biased error mitigation technique, meaning that in addition to  statistical error there is also a bias error which, contrary to statistical error, does not decrease with increasing sample size. 
Nevertheless, we provide strong evidence that: $(1)$ For $\eta \geq \eta_{th}$, there is a  sample size up to which the combined statistical and bias errors of recycling mitigation are less than the statistical error of postselection, meaning that recycling mitigation outperforms postselection up to this sample size. 
$(2)$ We present analytic and numerical evidence that the sample size after which postselection starts to outperform recycling mitigation seems very large, i.e. of the order of ${m \choose n}^2$, where ${m \choose n}$ is the size of the Fock space. 
$(3)$ We provide analytic and numerical evidence that, for fixed $m$ and $n$, and generic interferometers, the bias error in some of our introduced techniques seems small enough such that, in the limit of very low statistical error the mitigated outputs, like the ideal outputs, seem hard to approximate efficiently classically (see Section \ref{sec:properties_mitigated}).

Any loss parameter $0< \eta< 1$ results in a binomial distribution over $k$, the number of lost photons at the output. 
The key observation underpinning  the utility of recycling is that when $\eta \geq \eta_{th}$, the  statistical error denoted $\epsilon^{stat}_{k}$ is greatest for the postselection estimators where $k=0$, and furthermore that $\epsilon^{stat}_{k=0} \geq \epsilon^{stat}_{k=1} \geq \ldots \geq \epsilon^{stat}_{k=n_d}$, where $k \in \{1, \dots, n_d\}$ denotes the number of photons lost. $\epsilon^{stat}_{k}$ is the statistical error of recycled probabilities constructed from statistics with $k$ lost photons.
In order to obtain the mitigated outputs these techniques use sampled outputs for which $k>0$ meaning these lossy estimators have lower statistical error than the postselection estimators, but where $k$ is still low enough  such that the estimators contain information on the ideal outputs. 
The fact that statistical errors in recycling mitigation are lower than those of postselection, coupled with the low bias error in our techniques, is what allows the combined bias and statisical errors of the mitigated outputs to be lower than the statisical error of postselection, and thus for our methods to outperform postselection.  

Since recycling mitigation is applied to the statistics relating to relatively few lost photons, we show that the number of samples (runs of the lossy DVLOQC circuit) required for constructing a recycled probability at a given $k$, and consequently computing the mitigated values to a fixed accuracy, scales exponentially with system size as $O\Big(\frac{1}{{n \choose k}(1-\eta)^{n-k} \eta^k}\Big)$. 
On the one hand, this is a significant improvement over  the scaling of postselection when $\eta \geq \eta_{th}$, which is typically $O\big(\frac{1}{(1-\eta)^n}\big)$.
On the other hand, exponential scaling is a bottleneck  common to all error mitigation protocols involving solely classical postprocessing \cite{de_palma_limitations_2023,takagi_universal_2023,quek_exponentially_2023, tsubouchi_universal_2023}, and not one specific to our technique.
Finally, the fact that error mitigation techniques have exponential scaling does not necessarily hinder the achievement of a quantum computational advantage or utility for a fixed system size  \cite{zimboras2025myths}. 
Indeed, an initial experimental attempt to achieve a quantum computational utility for a pre-fault tolerant quantum device coupled with error mitigation was presented in \cite{kim2023evidence}. 
This experiment was later shown to be efficiently simulable classically, albeit using a highly non-trivial tensor network approach \cite{tindall2024efficient}. We are hopeful that other attempts in this direction will soon appear. Notably, in the DVLOQC framework where the methods we present can potentially aid in reaching a quantum computational utility in the pre-fault tolerant era.

The photon loss mitigation techniques we introduce can be used in a variety of applications. 
The work presented in \cite{salavrakos_error-mitigated_2024} introduced a quantum circuit Born machine (QCBM)  \cite{benedetti_parameterized_2019} tailored to the DVLOQC framework. 
Performance during the training of the QCBM, when afflicted by photon loss, considerably improved when the recycling mitigation method was applied.
Use of the mitigation in the QCBM experiments reduced the overall number of samples needed to reach a desired precision, thereby decreasing the overall computational time.
In addition to the application mentioned in \cite{salavrakos_error-mitigated_2024}, other  applications include variational quantum eigensolvers  \cite{maring_general-purpose_2023,lee_error-mitigated_2022}, photonic differential equation solving \cite{heurtel_perceval_2023}, photonic quantum machine learning \cite{gan_fock_2022},  and graph problems with DVLOQC \cite{mezher_solving_2023}.

Recent work \cite{taylor_quantum_2024} has provided a method for mitigating photon loss in the continuous variable (CV) setting, adapting pre-existing methods of probabilistic error cancellation \cite{temme_error_2017, strikis_learning-based_2021, mari_extending_2021, van_den_berg_probabilistic_2023} to the CV setting. 
Although the techniques of \cite{taylor_quantum_2024} are applicable to the DVLOQC setting, these are solely for expectation value mitigation (sometimes called weak mitigation) \cite{quek_exponentially_2023}. 
Our methods, by contrast, can perform both strong (full probability distribution mitigation) as well as weak mitigation, with a classical memory cost scaling with the size of the set of probabilities to be mitigated. 
Furthermore, we present analytical and extensive numerical evidence that our methods provably outperform postselection, whereas, to our knowledge, no comparison between the performance of the mitigation methods developed and postselection is presented in \cite{taylor_quantum_2024}.

This paper is structured as follows. 
Section \ref{sec:results} gives a high-level overview of our main contributions. 
Section \ref{sec:prelim} introduces some notation and basic concepts. 
Sections \ref{sec:construction}-\ref{sec:normalising} detail the construction of the recycled distributions, the classical postprocessing techniques needed to obtain the mitigated distribution. 
Section \ref{sec:numerics2} contains numerical simulations that aid in understanding how to use recycling mitigation in practice, as well as examples of the techniques in action.
Section \ref{sec:properties_mitigated}  discusses the properties of the mitigated probabilities in the limit of zero statistical error. 
Section \ref{sec:ZNE} provides strong evidence that techniques based on ZNE do not in general outperform postselection. Section \ref{sec:discussion} contains a discussion of our results as well as a set of interesting open questions.

\section{Overview of main results}
\label{sec:results}

\begin{figure}[]
\centering
\includegraphics[width=0.96\textwidth]{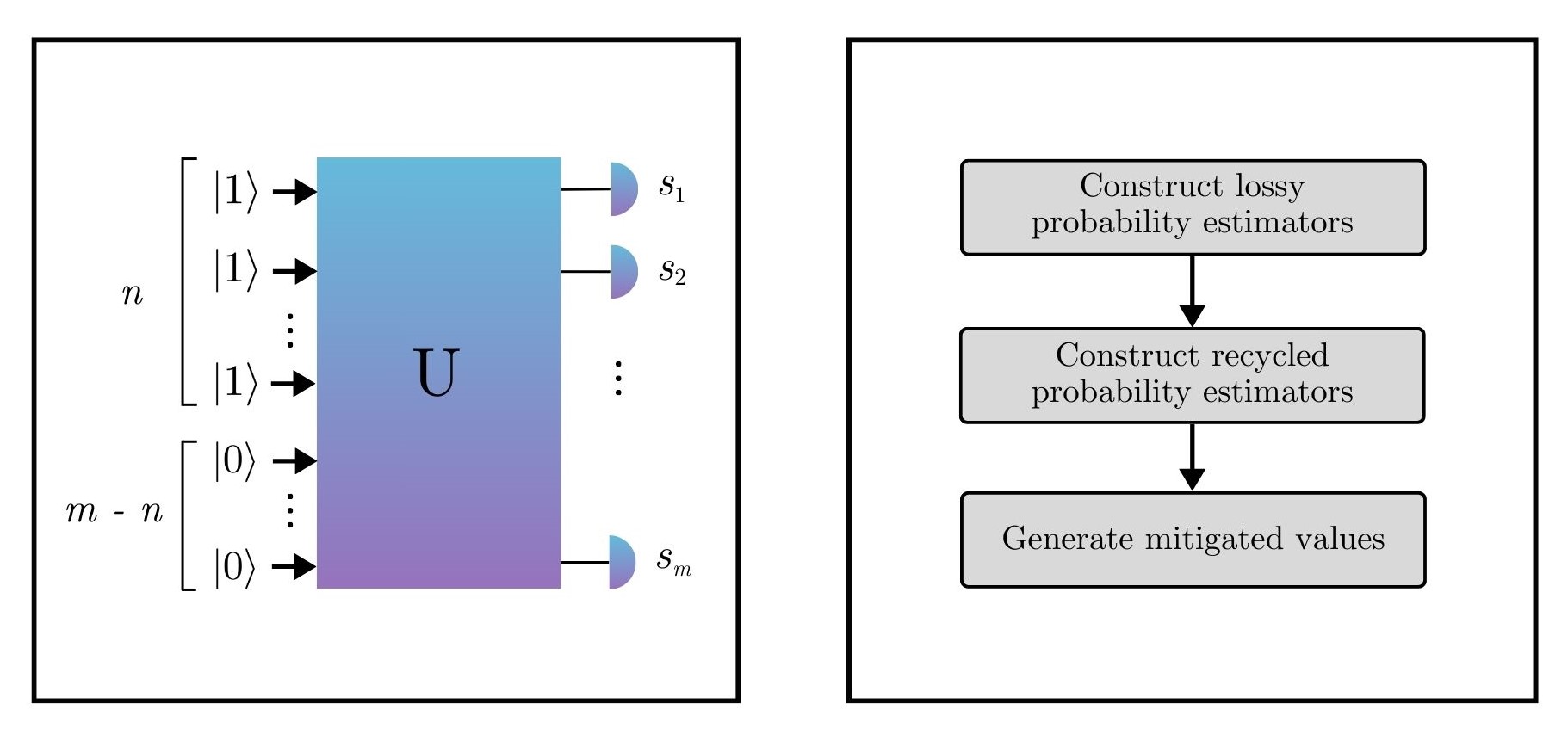}
\text{(a)\vspace{0.1em}\hspace{21.5em}}
\text{(b)\vspace{0.1em}\hspace{-0.5em}}
\caption{A schematic illustrating the main steps of the recycling mitigation protocol. (a) An input state of $n$ photons is introduced to a lossy $m$-mode linear optical interferometer implementing a unitary transformation. The classical measurement outcomes of each of the output modes is denoted by a classical $m-$bit string. (b) The classical data set generated by repeatedly sampling from the quantum circuit is then used as input for classical postprocessing. This classical postprocessing consists of three stages. First, the data set is used to generate lossy probability estimators. These estimators are then used to construct recycled probability estimators, which are in turn then used to generate mitigated values.}
\label{fig:PRschematic}
\end{figure}

This section summarises the main technical contributions of this work.

 First, we  compute the overall error on the output probabilities when applying the recycling mitigation technique. That is, we compute the difference between the mitigated and ideal (lossless) output probabilities. This overall error can mainly be seen as a sum of statistical error, as well as the bias error of the technique. This result is stated precisely in the following theorem.

\newpage

\begin{theorem}
\label{thpost}
 Consider $N_{tot}$ samples generated from a DVLOQC circuit with $n$ photons and $m$ modes. Let  $U \in \mathsf{U}(m)$ be the unitary implemented by the circuit , $\eta \in [0,1]$ the uniform probability a photon is lost in any mode.  
 Assume we are in the no-collision regime, $m \in \Omega(n^2)$,  where at most one photon occupies each mode.
 There exists a classical algorithm (linear solving recycling mitigation) that uses a subset  of  $$N_{rec}= {n \choose k} (1-\eta)^{n-k} \eta^k N_{tot}   \in   O \Big(
\sqrt{\frac{n}{k(n-k)}}
\;\frac{n^n}{k^k (n-k)^{\,n-k}}
\eta^k (1-\eta)^{n-k}N_{tot} \Big), $$ samples  where $k>0$ 
 photons were lost, runs in time $N_{rec} + \mathsf{poly}(m,n,k)$, and that outputs with probability $1-\delta$ an approximation $p_{mit}(s)$ of the exact lossless circuit probability of bit string $s$, $p_{id}(s)$ such that
\begin{equation}
\label{eqrec}
    | p_{mit}(s) - p_{id}(s) | \leq  f(n,m,k,\eta,N_{tot}),
\end{equation}
where $f(n,m,k,\eta,N_{tot})=3\Bigg[\sqrt{\frac{2}{\delta}} \left(\frac{e(m-n+k)}{k}\right)^k (\frac{n}{m})^{n/2} + \sqrt{\ln(4/\delta)} \left(\frac{k}{n}\right)^{k/2} \sqrt{\frac{1}{(1-\eta)^{n-k}\eta^kN_{\mathrm{tot}}}} \Bigg]$.
\end{theorem}

This theorem is proved in Appendix \ref{proofmainres}.
The proof makes use of Theorem  \ref{nonHaarbound} to bound the bias error of the mitigated output, and as the derived bound does not have an explicit dependence on $k$ it seems likely this could be considerably tightened.
As such, Theorem \ref{thpost} provides a loose upper bound on $f$.  For example, if one looks at minimising $f$ by varying $k$, one finds $k_{min}= \mathsf{ceil}(\eta n)$ minimises $f$. It is tempting to say that this is an optimal value of $k$ which should be used in the protocol. However, this is a misleading conclusion directly due to the upper bound for $f$ being loose. Indeed, numerically we find that the error mitigated probabilities for values of $k$ near $k_{min}$ are far from the ideal probabilities.   
Intuitively, this is because the recycled probabilities used to extract the error mitigated probabilities can efficiently be estimated classically at $k=k_{min}$ by collecting samples from a  classical  algorithm simulating lossy boson sampling.  
Indeed, collecting samples from a lossy boson sampler near the average number of lost photons, $\mathsf{ceil}(\eta n)$,  can be done efficiently classically  for all constant $\eta$    \cite{renema_classical_2019}. 
Numerically, we find that values of  $k \in O(1)$ produce the most accurate mitigation results (see Section \ref{sec:numerics2}). 
Note that these are rare sampling events for the  classical sampler of \cite{renema_classical_2019}. 
We believe that any hope of proving the optimality of  $k \in O(1)$ values requires a significant tightening of the upper bound on $f$, which in turn requires a more detailed understanding of the distribution of permanents of random matrices, which to the best of our knowledge remains incomplete \cite{nezami_permanent_2021}.

Despite its limitations, Theorem \ref{thpost} is still useful  to show that a regime of advantage over postselection exists, as we will now show.

%Note that to avoid including difficult-to-parse binomial terms, we make use of the relation that, for integers $a \geq b \geq 1$, $(\frac{a}{b})^b \leq {a \choose b} \leq (\frac{ea}{b})^b.$

% We derive conditions for which recycling mitigation outperforms postselection. 
Postselection uses a subset $N_{post}=(1-\eta)^nN_{tot}$ of samples where $k=0$ (i.e. no loss occurred) and outputs an estimate $p_{post}(s)$ of $p_{id}(s)$  such that
\begin{equation}
\label{eqpost}
|p_{post}(s)-p_{id}(s)| \leq \gamma \sqrt{\frac{1}{(1-\eta)^nN_{tot}}},
\end{equation}
with probability $1-2e^{-\gamma^2}$ and runtime $N_{post}+1$.
This upper bound on the additive error is computed via Hoeffding's inequality \cite{hoeffding_collected_1994}.

% In the worst-case, when the additive errors of recycling mitigation and postselection equal their upper bounds in Equation (\ref{eqrec}) and (\ref{eqpost}), 
% %and when setting $3\alpha=3\beta=\gamma$  
% we derive the following corollary.

The regime of advantage for the mitigation technique over postselection occurs where the combined statistical and bias errors of the mitigated output are lower than the statistical errors of the postselected output.
To see the sampling condition where the recycling mitigation outperforms postselection indicated by the analytical bounds one can, in the worst-case where the additive errors of recycling mitigation and postselection equal their upper bounds in eqn. (\ref{eqrec}) and (\ref{eqpost}), solve the inequality: $|p_{mit}(s)-p_{id}(s)| \leq |p_{post}(s)-p_{id}(s)|$.
% This is done in the worst-case, when the additive errors of recycling mitigation and postselection equal their upper bounds in eqn. (\ref{eqrec}) and (\ref{eqpost}).
% A sampling regime where this can be so is indicated in the following lemma.
This is stated in the following corollary.
\begin{corr}
We first assume that we are operating in the worst-case case error regime where $|p_{mit}(s)-p_{id}(s)|=f(m,n,k,\eta,N_{tot})$ (c.f. Theorem \ref{thpost}) and  $|p_{post}-p_{id}|=\gamma \sqrt{\frac{1}{(1-\eta)^n N_{tot}}}$. 
Furthermore, we assume that the statistical error of the postselected outputs is higher than the statistical error of the $(n-k)$-photon lossy outputs. 
Recycling mitigation outperforms postselection when
%, meaning $|p_{mit}(s)-p_{id}(s)| \leq |p_{post}(s)-p_{id}(s)|$, when

\begin{equation}
N_{\mathrm{tot}} \leq \frac{\delta(\Delta\eta)^2}{18} \Big(\frac{em}{n}\Big)^n
\left(\frac{k}{e(m-n+k)}\right)^{2k},
\end{equation}
where 
% $$\Delta\eta=\sqrt{\frac{1}{(1-\eta)^n}}- \left(\frac{k}{n}\right)^{k/2}\sqrt{\frac{1}{(1-\eta)^{n-k} \eta^k}},$$ 
\begin{equation}
\Delta\eta := \gamma \sqrt{\frac{1}{(1-\eta)^n}} - 3\sqrt{\ln (\!\frac{4}{\delta})}\,
 \left(\frac{k}{n}\right)^{k/2}\sqrt{\frac{1}{(1-\eta)^{n-k} \eta^k}},
\end{equation}
and, from  a union bound, this upper bound on $N_{tot}$ applies with confidence $1 -\delta'$, where $\delta' = \delta + 2e^{-\gamma^2}$.

\end{corr}
Note that the assumption in the above corollary that the statistical error from postselection is higher than the statistical error of the mitigated output may be concisely stated as $\Delta \eta >0$.
% \begin{equation}
% N_{tot} \leq \frac{(\Delta \eta)^2 \big(\frac{em}{n}\big)^n }{\big((\frac{m-n+k}{k})^k-1\big)^2}.
% \end{equation}
% This may be seen by rearranging the terms in the inequality  $f(m,n,k,\eta,N_{tot}) \leq \gamma \frac{1}{\sqrt{(1-\eta)^nN_{tot}}}$ and setting $3\alpha=3\beta=\gamma$.
This corollary means that, in worst-case, there exists a non-trivial sampling regime up to which recycling mitigation outperforms postselection. 
Furthermore, the condition $\Delta \eta >0$ indicates that there is a photon loss rate threshold that needs to be satisfied in order for recycling mitigation to outperform postselection. 
Indeed, $\Delta \eta >0$ implies the existence of a lower bound on the photon loss rate for which it is possible for mitigation to outperform postselection: $\eta > \eta_{th}(n,k)$
with 
$$\eta_{th}(n,k) \geq \frac{1}{1+\frac{en}{k}}.$$

For sufficiently large $n$ and for a fixed $k$ independent of $n$, we have that $\left(\frac{e(m-n+k)}{k}\right)^{2k} \in O\big(\mathsf{poly}(m,n)\big)$ and $\Delta \eta \in O\bigg(\frac{1}{\sqrt{(1-\eta)^n}}\bigg)$.
This indicates a regime of advantage for recycling mitigation over postselection up to a sample number
\begin{equation}
N_{tot} \in O\bigg(\frac{\big(\frac{em}{n}\big)^n}{(1-\eta)^n \mathsf{poly}(m,n)}\bigg) \in e^{n \text{log}(\mathsf{poly}(n))+O(n) - \text{log}(\mathsf{poly(n))} +O(1)}.
\end{equation}
Consequently, recycling mitigation 
outperforms postselection for approximating the ideal probability  up to an additive error 
\begin{equation}
\epsilon \in O\bigg(\mathsf{poly}(m,n) \Big(\frac{n}{m}\Big)^{n/2}\bigg) \in e^{-\frac{n}{2} \text{log}(\mathsf{poly}(n)) + \text{log}(\mathsf{poly}(n))+O(1)}.
\end{equation}
Intuitively, because the statistical error for recycling mitigation is lower than that for postselection, the point up  to which recycling mitigation outperforms postselection is typically determined by the bias error of the mitigation. 
Indeed, $O(\mathsf{poly}(m,n)\big(\frac{n}{m}\big)^{n/2})$ is a high-confidence upper bound on the bias error obtained in Theorem \ref{nonHaarbound}.

We present another recycling mitigation technique, exponential extrapolation, which outperforms the linear solving technique in numerical simulation experiments.
Unlike linear solving, which uses a single value of $k$, exponential extrapolation uses recycled probabilities computed at multiple values of $k$, and attempts to extract ideal probabilities from the decay of these probabilities towards the uniform distribution (obtained at $k=n$ when all photons are lost).  
This decay is fitted to an exponential curve. 
The choice of the decay function is primarily heuristic and accurately captures the form of the decay function observed numerically.

We numerically observe that exponential extrapolation has lower bias error than linear solving. 
Although we do not show this analytically, we make the following conjecture about the performance of exponential extrapolation based on numerical evidence presented in Section \ref{sec:properties_mitigated}.
\begin{conj}
\label{conjbias}
When $m \in \Omega(n^2)$, and for values of $k \in O(1)$, exponential extrapolation recycling mitigation gives the following performance guarantee with high probability over the choice of $U$
\begin{equation*}
E_{s}(\epsilon_{{mitbias},s}) \leq \delta(n) \Big(\frac{n}{m}\Big)^n,% \frac{1}{{m \choose n}},
\end{equation*}
where $E_{s}(\epsilon_{{mitbias},s})$ is the average, for fixed $U$, over all bit strings $s$ of the bias error $\epsilon_{{mitbias},s}$, and $\forall n$ $\delta(n) <1$  is some function of $n$ .
\end{conj}
By using a Markov inequality, Conjecture \ref{conjbias} being true implies that, with high probability over the choice of $U$, and with probability $1-\frac{1}{w}$ over the choice of $s$ for a fixed $U$ and $w>1$, we have $\epsilon_{{mitbias},s} \leq w\delta(n) \big(\frac{n}{m}\big)^n$. 
In particular, one can choose a $w$ such that $w\delta(n)<1$, and relabel $\kappa(n)=w\delta(n) <1$, Conjecture  \ref{conjbias} being true therefore implies that with high probability over $U$ and with probability $1-\frac{1}{w}$ over the choice of $s$ for a fixed $U$, $\epsilon_{{mitbias},s} \leq \kappa(n) \big(\frac{n}{m}\big)^n$ with $\kappa(n)<1$.

Furthermore, when $\eta$ is above some threshold value, we numerically find that the statistical error of extrapolation  is lower  than that of postselection. 
Let $\epsilon_{stat}(m,n,k,\eta,N_{tot})$  be the upper bound on the absolute value of the statistical error of extrapolation mitigation, we conjecture based on numerical evidence the following.

\begin{conj}
\label{conjstatextr}
For $\eta > \eta_{th}$, where $\eta_{th} \in [0,1]$ is some threshold loss value, we have that  $$\epsilon_{stat}(m,n,k,\eta,N_{tot}) \leq \frac{1}{\sqrt{(1-\eta)^nN_{tot}}}.$$ Furthermore, for large enough $n$, $\epsilon_{stat}(m,n,k,\eta,N_{tot}) \in o\bigg(\frac{1}{\sqrt{(1-\eta)^nN_{tot}}}\bigg)$.
\end{conj}
 
If our conjectures on the bias and statistical errors of exponential extrapolation are true, we can show the following theorem in worst-case, when the errors on postselection and exponential extrapolation equal their upper bounds.

\begin{theorem}
\label{thexpextrap}
Assume that Conjectures \ref{conjbias} and \ref{conjstatextr} are true. Furthermore, assume that $n$  is sufficiently large and we are in the case where $|p_{mit}(s)-p_{id}(s)| = \epsilon_{stat}(m,n,k,\eta,N_{tot})+ \kappa(n) \big(\frac{n}{m}\big)^n$, with $\kappa(n)<1$, and $ |p_{post}(s)-p_{id}(s)|= \sqrt{\frac{1}{(1-\eta)^nN_{tot}}}$, where $p_{mit}(s)$ is the output probability of exponential extrapolation recycling mitigation. 
Then exponential extrapolation outperforms postselection for up to sample number $N_{tot} \ \in O\Big(\frac{(\frac{em}{n})^{2n}}{\kappa(n)^2 (1-\eta)^n}\Big) \in e^{2n \log(\mathsf{poly}(n))+O(n)-\log(\kappa(n))+O(1)}$ and up to additive error $\epsilon \in O\Big(\kappa(n)\big(\frac{n}{m}\big)^n\Big) \in e^{-n \log(\mathsf{poly}(n))+\log(\kappa(n)) +O(1)}$.
\end{theorem}
A proof of this is provided in Appendix \ref{extrapolation_perf_g_app}.

Theorem \ref{thexpextrap} indicates that in worst-case, exponential extrapolation has better performance  than linear solving, as the additive error up to which exponential extrapolation outperforms postselection is significantly smaller than that of linear solving.

% \begin{proof}
% Conjecture \ref{conjstatextr} being true implies that  $$\sqrt{\frac{1}{(1-\eta)^nN_{tot}}}-\epsilon_{stat}(m,n,k,\eta,N_{tot}) \geq c\sqrt{\frac{1}{(1-\eta)^nN_{tot}}},$$ for some $0<c<1$ for sufficiently large $n$. 
% Thus, the condition that exponential extrapolation outperforms postselection, given by the following inequality holding $|p_{mit}-p_{id}(s)| \leq |p_{post}(s)-p_{id}(s)|$, is satisfied if $\kappa(n) \frac{1}{{m \choose n}} \leq c\sqrt{\frac{1}{(1-\eta)^nN_{tot}}}.$ 
% The  bound on $N_{tot}$ in the Theorem immediately follows. Furthermore, the  bound on $\epsilon$ is directly obtained by replacing  $(1-\eta)^nN_{tot} \in O\Big(\frac{{m \choose n}^2}{\kappa(n)^2}\Big)$ in $\frac{1}{\sqrt{(1-\eta)^nN_{tot}}}.$
%     \end{proof}

In Section \ref{sec:numerics2}, we present the results of numerical simulations comparing the performance of the mitigation techniques with postselection.
The recycling mitigation methods (linear solving and exponential extrapolation) consistently outperform postselection for a range of loss values and sample numbers.
We also present results that indicates that in the average-case the errors of both recycling mitigation methods and postselection are below their analytically computed upper bounds. 
% These can be found in Section \ref{sec:numerics2}.
These experiments were performed in loss regimes relevant to near-term quantum photonics hardware.

In Section \ref{sec:properties_mitigated}, we analyse the mitigated outputs in the regime $m \in \Omega(n^2)$ for exponential extrapolation error mitigation, and conjecture that, in the absence of statistical errors, the distribution of $p_{mit}(s)$ is likely classically hard to compute.
We support this conjecture with numerical calculations. 
Specifically, the conjecture is that computing the mitigated outputs is a problem that reduces to solving a restricted version of the $|GPE_{\pm}|^2$ problem \cite{aaronson_computational_2011}. 
This result is interesting because it indicates that, in general and in the absence of statistical error, the bias errors are small enough to render the error mitigated  probabilities hard to compute classically.

Our final  contribution is in providing strong evidence that photon loss error mitigation techniques based on ZNE  \cite{endo_hybrid_2021} cannot outperform postselection in the DVLOQC setting. 
Indeed, in Section \ref{sec:ZNE} we provide an upper bound on the error incurred from  ZNE, $E_{extrap}$,  which is larger than the worst-case statistical error of postselection (see Thm.  \ref{thnogozne}). 
In Appendix \ref{appnogo}, we also provide numerical evidence that for all $n \geq n_0$, for some $n_0 \in \mathbb{N}$, $E_{extrap}$ is always larger than the statistical error of postselection.
Indeed, in DVLOQC there is a natural way to mitigate photon loss errors, postselection, which is not present for other types of errors in other types of hardware \cite{endo_hybrid_2021}. 
This therefore sets a benchmark that a photon loss mitigation technique in DVLOQC must satisfy to be useful, namely that it must outperform postselection.

\section{Preliminaries}
\label{sec:prelim}

Consider the DVLOQC setting where a photonic quantum device is composed of a single-photon source \cite{senellart_high-performance_2017}, a universal $m$-mode linear optical interferometer \cite{reck_experimental_1994}  capable of implementing any unitary transformation $U \in \mathsf{U}(m)$, with $\mathsf{U}(m)$ the group of unitary $m \times m$ matrices, and single photon detectors \cite{hadfield_single-photon_2009}. Generating  a sample using this device proceeds as follows. First, $n$ single photons emitted from the source pass through the linear optical inteferometer in an input configuration $\mathbf{T}:=(t_1, \dots, t_m)$, where $t_i$ is the number of photons in mode $i$. Let $|\psi_{in}\rangle:=|t_1,\dots,t_m\rangle$  be the  input Fock state of single-photons corresponding to the configuration $\mathbf{T}$.  The linear optical interferometer implements a unitary transformation $\phi(U)$ on the input state $|\psi_{in}\rangle$ \cite{aaronson_computational_2011}, resulting in an output state $|\psi_{out}\rangle:=\phi(U)|\psi_{in}\rangle$, where $\phi(U)$ represents the action of the unitary $U$, implemented by the interferometer, on $|\psi_{in}\rangle$. Note that $\phi(U)$ and $U$ are related by a homomorphism, detailed in \cite{aaronson_computational_2011}. A sample $\mathbf{S}:=(s_1, \dots, s_m)$, with $s_i$ the number of photons in output mode $i$, is then obtained by measuring the number of photons in each output mode using  single-photon detectors. 
This corresponds to projecting $|\psi_{out}\rangle$ onto the state  $|s_1,\dots,s_m\rangle$.  Many computational tasks on photonic quantum devices can be implemented by collecting samples according to the previous procedure, and then performing classical postprocessing \cite{maring_general-purpose_2023,mezher_solving_2023,singh_proof--work_2023,lee_error-mitigated_2022,knill_scheme_2001}.

In the absence of any errors affecting the device, $\sum_{i=1, \dots , m}t_i=\sum_{i=1, \dots, m}s_i=n$,  furthermore, the probability of obtaining the sample $\mathbf{S}$ is proportional to the modulus squared of the permanent of a submatrix $U_{\mathbf{T},\mathbf{S}}$  of $U$, whose rows and columns are determined by the input and output occupancies $\mathbf{T}$ and $\mathbf{S}$ \cite{aaronson_computational_2011}. By appropriately choosing the unitary transformation $U$, and performing the above mentioned sampling procedure repeatedly, one can perform both non-universal and universal quantum computing with linear optics. In particular, if $U$ is chosen to be Haar random, one performs boson sampling, a non-universal sampling task which is hard for classical computers to carry out efficiently \cite{aaronson_computational_2011}. Alternatively, choosing specific unitaries $U$, and postselecting on detecting a specific output configuration, one can perform universal quantum computation \cite{knill_scheme_2001}.

We will now describe our error model as well as the assumptions we will make throughout this paper.

\begin{itemize}
\item We consider photon loss as the only source of error affecting our devices. Our error model is the uniform loss model, where a photon is  equally likely  to be lost in any mode $i \in \{1, \dots,m\}$ with probability $\eta \in [0,1]$. Following the commutation rules of photon loss \cite{oszmaniec_classical_2018}, we assume without loss of generality that the photons are lost at the output of the interferometer, just before the single-photon detectors which are assumed to be perfect.

\item We assume that sampling occurs in the no-collision regime, where at most  one photon occupies any output mode. 
This is approximately true for  $m \in \Omega (n^2)$ \cite{aaronson_computational_2011}. 
In this regime the samples are given by $\mathbf{S}=(s_1, \dots, s_m)$, with $s_i \in \{0,1\}$, and are therefore bit strings of length $m$. The total number of possible no-collision outputs is ${ m \choose n}$.
 \end{itemize}

The uniform loss model is a widely used error model for simulating photon loss, and is the standard assumption used when constructing loss-tolerant quantum error correcting codes \cite{stace_error_2010}. 
It is also often assumed when deriving efficient classical algorithms for simulating lossy linear optical setups \cite{garcia-patron_simulating_2019}. 
Working in the no-collision regime is primarily interesting for two reasons. 
Firstly, one can prove statements of quantum advantage in boson sampling in this regime \cite{aaronson_computational_2011}. 
And secondly, the fact that at most one photon occupies each mode allows us to assume the use of the standard and widely available threshold detectors, rather than number resolving ones which are currently challenging to practically implement. 
As a final note, while the no-collision assumption is a useful one, it is not necessary for our techniques to work. 
Indeed, as discussed in later parts of this paper, our techniques can in principle  be generalised to the case where more than one photon can occupy a mode.

\bigskip
The uniform loss model induces a binomial distribution on the samples, in the sense that sampling $N_{tot}$ times from a uniformly lossy linear optical circuit produces approximately $N_{tot,k}:={n \choose k} N_{tot} \eta^k (1-\eta)^{n-k}$ samples corresponding to $k$ lost photons, for $k \in \{0, \dots, n\}.$ Note that $N_{tot}=\sum_{k=0, \dots, n}N_{tot,k}$. 
We will take $s^{n-k}_i$ to mean a bit string of the form $\{s_1, \dots, s_m\}$ where $\sum_i s_i=n-k$, $s_i \in \{0,1\}$ is the number of photons in mode $i$, and $k \in \{0, \dots, n\}$. 
This corresponds to a sample drawn from the probability distribution where $k$ of the initial $n$ input photons have been lost. 
In order to estimate the probability $p(s_i^{n-k})$ from a set $\mathcal{W}:=\{s^{n-k}_j\}$  of samples \footnote{The set can contain repeated identical bit strings, i.e. there can be  $j_1 \neq j_2$, such that  $s^{n-k}_{j_1}=s^{n-k}_{j_2}$. }  where $|\mathcal{W}| \leq N_{tot,k}$, we perform the following procedure. For each $w$ ranging from 1 to $|\mathcal{W}|$, assign a value 1 to a random variable $X_w \in \{0,1\}$ if the sample  $s^{n-k}_w$  is the bit string $s^{n-k}_i$, and assign the value 0 to $X_w$ otherwise. The estimate $\tilde{p}(s^{n-k}_i)$ is then 
\begin{equation} \label{estimator_eqn}
\tilde{p}(s^{n-k}_i):=\frac{\sum_w X_w}{|\mathcal{W}|}.
\end{equation}
This estimation therefore induces a statistical error given by
\begin{equation}
    \epsilon_{stat}(s^{n-k}_i):=|\tilde{p}(s^{n-k}_i)-p(s^{n-k}_i)|.
\end{equation}
In postselection, estimators are constructed using only the outputs for which $k=0$, when photon loss is uniform this corresponds to approximately $ N_{tot} (1-\eta)^{n}$ samples.
As the the size of the system  is increased, the probability of postselecting non-lossy outcomes decays exponentially towards zero.

For  values of loss above a certain threshold, the statistical error of probabilities constructed from lossy statistics is in general lower than those constructed from lossless ones. 
As an example of how to see this, note that for $k=1$, the expected number of samples is $N_{tot}n(1-\eta)^{n-1}\eta$. 
For  the lossless $k=0$ case, we have $N_{tot}(1-\eta)^n$ such samples.
Since the statistical error is typically upper bounded by $1 / \sqrt{N}$, where $N$ is the sample number, we can see that the inequality  $N_{tot}n(1-\eta)^{n-1}\eta \geq N_{tot}(1-\eta)^n$ (which implies lower statistical errors for lossy probabilities) holds whenever $\eta \geq \frac{1}{n+1}$.

Recycling mitigation uses the $n-k$--photon probability estimates $\{\tilde{p}(s_i^{n-k})\}$, potentially for a range of $k$ values, and construct from these a \emph{mitigated} $n$--photon probability distribution  $\{p_{mit}(s_j^{n})\}$. To have any utility, a photon loss mitigation technique needs to outperform computing the $n$--photon probability estimates $\{\tilde{p}(s_j^{n})\}$ from the samples $N_{tot,0}$, which we will henceforth refer to as postselection on $n$--photon outputs, or just postselection. We will therefore use postselection as the benchmark to evaluate the performance of recycling mitigation. Postselection is, to our knowledge, the only technique being used to mitigate the effects of photon loss on current DVLOQC  hardware \cite{maring_general-purpose_2023}. Another factor motivating recycling mitigation is that it does not increase the overall sample cost relative to postselection. This contrasts favourably with many error mitigation results that have an accompanying sampling overhead \cite{endo_hybrid_2021}. 

\section{Method}

\subsection{Recycled probabilities}
\label{sec:construction}

The recycled probabilities are constructed from $n-k$ output photon statistics, where $k\in\{1,\ldots,n-1\}$. To explicitly analyse the signal of the ideal probability within the recycled probability, the recycled probabilities may be decomposed into a combination of an ideal $n$--photon output probability, and an interference term consisting of a mixture of other $n$--photon output probabilities from the distribution. 
We first describe the construction of the recycled probabilities from $n-k$ output photon statistics, which generalises for all $k$. We then describe the analytical decomposition of the recycled probabilities into $n$--photon output probabilities.

\subsubsection{Construction of recycled probabilities from lossy outputs}

We now detail the construction of recycled probabilities from $n-k$--photon output statistics - that is, the output statistics in which exactly $k$ of $n$ photons have been lost. This construction should be applied to obtain the recycled probability distribution in an experiment. The recycled probability for bit string $s^n_l$ computed from $n-k$--photon output statistics is denoted $p_R^k(s^n_l)$, with $k\in\{1, \ldots, n-1\}$. Hence there are $n-1$ recycled probabilities one can construct from lossy output statistics for any $n$--photon output bit string, one for each possible value of $k$. Performing the construction involves summing over the $n-k$ output photon bit string probabilities that relate to a particular ideal probability. Informally, this relation is that these are the probabilities of $n-k$--photon output bit strings that the ideal output bit string can be mapped to through the loss of $k$ photons. We now provide a formal statement of this relation.

We first define a mapping procedure from each $n$--photon output bit string to a set of $n-k$--photon output bit strings. Each mapped set of bit strings represents the set of all possible states that the associated $n$--photon output state could become after losing $k$ photons. Let $\mathcal{M}_{\text{unocc.},k,i}$ be the subset of $\{1, \dots, m\}$ corresponding to the unoccupied modes of the  output bit string $s^{n-k}_i$. That is, the set of indices of the modes $j\in\{1, \dots, m\}$ of the bit string for which $s_j=0$. The number of unoccupied modes for $n-k$--photon outputs is $|\mathcal{M}_{\text{unocc.},k,i}|=m-n+k$. Let $\overline{\mathcal{M}_{\text{unocc.},k,i}}$ be the complement of $\mathcal{M}_{\text{unocc.},k,i}$ in $\{1, \dots, m\}$, so that $\overline{\mathcal{M}_{\text{unocc.},k,i}} \cup \mathcal{M}_{\text{unocc.},k,i} = \{1, \dots, m\}$. The subset $\overline{\mathcal{M}_{\text{unocc.},k,i}}$ denotes the occupied modes of $s^{n-k}_i$, consisting of the set of indices of modes $j\in\{1, \dots, m\}$ for which $s_j=1$. As the bit string $s^{n-k}_i$ represents an $n-k$--photon output the number of occupied modes is $|\overline{\mathcal{M}_{\text{unocc.},k,i}}|=n-k$. We define the set $\mathcal{L}(s^n_i):=\{s^{n-k}_j | s^n_i \Rightarrow s^{n-k}_j\}$. And the set of all size $k$ subsets of $\overline{\mathcal{M}_{\text{unocc.},0,i}}$ is $\overline{\mathcal{S}_{0,i}}:=\{X \subset \overline{\mathcal{M}_{\text{unocc.},0,i}}|\hspace{0.3em} |X|=k\}.$ The symbol `$\Rightarrow$' denotes the operation where, for every size $k$ subset $\{\overline{l_1},\ldots,\overline{l_k}\} \in \overline{\mathcal{S}_{0,i}}$, 
 the $n$--photon output bit string $s_i^{n}$ is mapped to a new $n-k$--photon output bit string $s^{n-k}_j$ by replacing $s_{\overline{l_i}}=1$ with $s_{\overline{l_i}}=0$. 
In this case, the number of $k-$subsets of 
$\overline{\mathcal{M}_{\text{unocc.},0,i}}$ is ${ n \choose k}$, and so $|\mathcal{S}_{0,i}| = {n \choose k}$ and the size of the generated set of bit strings is $|\mathcal{L}(s^n_i)|={ n \choose k}$.

The recycled probability for the bit string $s_l^{n}$ can be defined as the sum of probabilities of the ${n \choose k}$ bit strings $s_i^{n-k} \in \mathcal{L}(s^n_l)$, 
\begin{equation} \label{recycled_probs}
\begin{split}
     p_R^k(s^n_l)&:= \sum_{s_i^{n-k} \in \mathcal{L}(s^n_l)}p(s^{n-k}_i) .\\
\end{split}
\end{equation}
To ensure the normalisation of the recycled distribution the above expression is multiplied by a normalisation factor $\mathbf{N} = \frac{1}{{m-n+k \choose k}}$, so that in practice it is
\begin{equation} \label{normalised_recycled_probs}
\begin{split}
     p_R^k(s^n_l)&= \frac{1}{{m-n+k \choose k}}\sum_{s_i^{n-k} \in \mathcal{L}(s^n_l)}p(s^{n-k}_i) .\\
\end{split}
\end{equation}
A derivation of the normalisation factor is provided in the next section. 

To illustrate how this construction might work in practice we now provide a small example. If we would like to compute the recycled probability for the bit string: 111000, in an experiment in which there are $m=6$ modes, $n=3$ input photons and the construction is being performed for $k=1$ lost photons. The recycled probability may be computed directly from eqn. \ref{normalised_recycled_probs} as being
\begin{equation}
p_R^1(111000) = \big(p(110000) + p(101000) + p(011000)\big){4 \choose 1}^{-1},
\end{equation}
where the normalising parameter is $\mathbf{N} = {4 \choose 1}^{-1}$. 

In practice, rather than using the set of exact probabilities, $\{p(s_i^{n-k})\}$, to compute recycled probabilities in the manner shown in eqn. \ref{normalised_recycled_probs}, instead empirical estimates of the exact probabilities, $\{\tilde{p}(s_i^{n-k})\}$, are used. These are calculated from the set of measured experimental output bit strings, as described for eqn. \ref{estimator_eqn}, and so include statistical errors due to finite samples. This statistical error is an important consideration when making comparisons with postselection, and is later included in the analysis of the protocol performance. We now provide pseudocode detailing how to construct the recycled probability for a specific $n$ output photon bit string and value of $k$ from a set of $N$ output bit strings sampled from a given circuit.

\vspace{1em}

\begin{algorithm*}[H]
\DontPrintSemicolon
\SetAlgoLined
\SetKwInOut{Input}{input}\SetKwInOut{Output}{output}
\SetKwRepeat{Repeat}{repeat}{until}
\Input{The $n$--photon output bit string $s^n_l$ for which the recycled probability estimator is to be constructed, a set of $N$ output sample bit strings from the DVLOQC circuit $\{s_j\}_{j\in \{1,\ldots,N\}}$, and the choice of $k$ value indicating that $n-k$ output photon statistics be used for the construction.}
\BlankLine\
Initialise variable $\tilde{p}_R^k(s^{n}_l)\gets 0$ for the recycled probability estimator to be computed. \;
Create a new set of bit strings by discarding all output bit strings from set $\{s_j\}_{j\in \{1,\ldots,N\}}$ except those for which the number of measured output photons was $n-k$, with the new list denoted $\{s_l\}_{l\in \{1,\ldots,N_{est,k}\}}$ where $N_{est,k}\leq N$.\;
Generate the set of $n-k$--photon lossy output bit strings $\mathcal{L}(s^n_l)$.\;
Initialise variable $X_{s_l}\gets 0$. \;
\For{\textnormal{l = $1$ \textbf{to} $N_{est,k}$}}{
\If{$s_l \in \mathcal{L}(s^n_l) $}{
$X_{s_l}\gets X_{s_l}+1$.
}
}
Update recycled probability estimator variable as $\tilde{p}_R^k(s^{n}_l)\gets \frac{ X_{s_l}}{{m-n+k \choose k}N_{est,k}}$.\;
\Output{Recycled probability estimator $\tilde{p}^k_R(s^n_l)$\;}
\caption{Construction of a recycled probability estimator from output statistics}
\end{algorithm*}

\subsubsection{Decomposition of recycled probabilities into $n$--photon output bit string probabilities}

In later sections, the recycled probabilities are analysed in terms of their decomposition into $n$--photon output probabilities. For a given output bit string $s^n_l$, this representation allows explicit treatment of the signal of the ideal $n$--photon output probability $p(s^n_l)$ within the recycled probability $p^k_R(s^n_l)$. To get the recycled probabilities in this form, the lossy output probabilities within the sum in eqn. \ref{normalised_recycled_probs} are decomposed into $n$--photon output probabilities from the ideal distribution. The details of this decomposition will now be formalised.

For this purpose we now define another mapping procedure, this time from each lossy $n-k$--photon output bit string from the sum in eqn. \ref{normalised_recycled_probs} to a set of $n$--photon output bit strings. 
Where each mapped set of bit strings represents the set of $n$--photon outputs that the associated lossy output bit string could have been had loss not occurred. 
The set of all size $k$ subsets of $\mathcal{M}_{\text{unocc.},k,i}$ is defined $\mathcal{S}_{k,i}:=\{X \subset \mathcal{M}_{\text{unocc.},k,i}|\hspace{0.3em} |X|=k\}$.
We define the set $\mathcal{G}(s^{n-k}_{i}):= \{s_j^{n} | s_i^{n-k} \to s^n_{j} \}$. The symbol `$\to$' denotes the operation where, for every size $k$ subset $\{l_1,\ldots,l_k\} \in \mathcal{S}_{k,i}$, the bit string $s_i^{n-k}$ is mapped to a new bit string by replacing
 each bit $s_{l_i}=0$ with $s_{l_i}=1$. 
The number of $k-$subsets of $\mathcal{M}_{\text{unocc.},k,i}$ is ${m-n+k \choose k}$, and so $|\mathcal{S}_{k,i}| = {m-n+k \choose k}$ and the size of the generated set of bit strings is $|\mathcal{G}(s^{n-k}_{i})| = {m-n+k \choose k}$. 

From the addition rule of probabilities and the uniformity of the loss, the probability $p(s^{n-k}_i)$ of obtaining the output bit string $s^{n-k}_i$ is 
\begin{equation}\label{eqn:n_k_probs}
    p(s^{n-k}_i) = \sum_{ s^n_j \in \mathcal{G}(s_i^{n-k})}p(s^n_j)\frac{1}{{n \choose k}},
\end{equation}
where the sum includes all the bit string outputs $s^n_j \in \mathcal{G}(s_i^{n-k})$ from which the loss of $k$ photons maps to $s_i^{n-k}$. Each element of the sum is composed of the probability $p(s^n_j)$ that the output is the $n$--photon bit string $s^n_j$, multiplied by the uniform probability ${{n \choose k}}^{-1}$ of the loss of $k$ photons from $s^n_j$ resulting in the output $s^{n-k}_i$.

As in eqn. \ref{eqn:n_k_probs}, in the definition of the recycled probabilities in eqn. \ref{normalised_recycled_probs} each of the lossy probabilities can be decomposed into a convex combination of probabilities from the ideal distribution. That is, the $n-k$--photon output probabilities within the sum in eqn. \ref{normalised_recycled_probs} can be replaced with a sum over $n$--photon output probabilities. For some arbitrary labelling of bit strings $s^n_j$, and noting that $s^n_l \in \mathcal{G}(s_i^{n-k})$, we can expand the $n-k$ output photon probabilities in the form
\begin{equation}
p(s^{n-k}_i)=  p(s^n_l)\frac{1}{{n \choose k}}+ \sum_{ s^n_j \in \mathcal{G}(s_i^{n-k}), j \neq l}p(s^n_j)\frac{1}{{n \choose k}}.
\end{equation}
This can then be used to separate the contribution of the ideal bit string probability $p(s^n_l)$ out from the rest of the probabilities which are grouped into a sum we call an \emph{interference term} in the recycled probability. So it becomes
\begin{equation}
\begin{split}
  p_R^k(s^n_l)   &= p(s^n_l) + \sum_{s_i^{n-k} \in \mathcal{L}(s^n_l)}\sum_{ s^n_j \in \mathcal{G}(s_i^{n-k}), j \neq l}p(s^n_j)\frac{1}{{n \choose k}}.\\
\end{split}
\end{equation}
The last step is to normalise the distribution $p^k_R(s^n_l)$, namely to compute $\mathbf{N}$
such that $\mathbf{N}\cdot\sum_lp_R^k(s^n_l)=1$. To do this, note that
\begin{equation} \label{recycling_normalisation_factor}
\begin{split}
   \sum_{l=1}^{{m \choose n}}p_R^k(s^n_l) & =  \sum_{l=1}^{{m \choose n}} \sum_{s_i^{n-k} \in \mathcal{L}(s^n_l)}p(s^{n-k}_i)\\
   & =  \sum_{l=1}^{{m \choose n}} \sum_{s_i^{n-k} \in \mathcal{L}(s^n_l)}\sum_{ s^n_j \in \mathcal{G}(s_i^{n-k})}p(s^n_j)\frac{1}{{n \choose k}}\\
   & ={{m-n+k} \choose k},\\
\end{split}
\end{equation}
This is because $\sum_l p(s^{n}_l)=1$, and, because $|\mathcal{L}(s^n_i)|={ n \choose k}$ and $|\mathcal{G}(s^{n-k}_{i})| = {m-n+k \choose k}$, each distinct bit string $s^{n}_l$ appears exactly  ${n \choose k}{m-n+k \choose k}$ times in the above sum \footnote{Another way to see this is to consider the sum over recycled probabilities as $\sum_lp_R(s^n_l)=  \sum_l \sum_{s_i^{n-k} \in \mathcal{L}(s^n_l)}p(s^{n-k}_i)= {m-n+k \choose k}$. Where the rightmost equality follows from the fact that each $p(s^{n-k}_i)$ appears exactly ${m-n +k \choose k}$ times in this sum, since there are $|\mathcal{G}(s^{n-k}_i)|$ recycled probabilities sharing a single $p(s^{n-k}_i)$. Furthermore note that $\sum_ip(s^{n-k}_i)=1$.  }. Therefore $\mathbf{N} = \frac{1}{{m-n+k \choose k}}$, and the expression for the normalised recycled probability is
\begin{equation} \label{recycled_prob_def}
\begin{split}
  p_R^k(s^n_l)   &= p(s^n_l)\frac{1}{{m-n+k \choose k}}+\sum_{s_i^{n-k} \in \mathcal{L}(s^n_l)}\sum_{ s^n_j \in \mathcal{G}(s_i^{n-k}), j \neq l}p(s^n_j)\frac{1}{{m-n+k \choose k}{n \choose k}}.\\
\end{split}
\end{equation}
Let $N_k:={m-n+k \choose k}{n \choose k}$, $N'_k:=\big({m-n+k \choose k}-1\big){n \choose k}$. The recycled probability $p^k_R(s_l^n)$ is composed of the ideal output probability $p(s^n_l)$ and an interference term, defined as 
\begin{equation}
\label{eqinterference}
I_{s_l^n,k}:=\sum_{s_i^{n-k} \in \mathcal{L}(s^n_l)}\sum_{ s^n_j \in \mathcal{G}(s_i^{n-k}), j \neq l}p(s^n_j)\frac{1}{N'_k}.
\end{equation}
The recycled probability may then be written explicitly in this form
\begin{equation}
    p_R^k(s^n_l)=p(s^n_l)\frac{1}{{m-n+k \choose k}}+\frac{N'_k}{N_k}I_{s_l^n,k}.
\end{equation}
Importantly, the photon recycled probability $p_R^k(s^n_l)$ contains an amplified signal of the ideal probability $p(s^n_l)$, by a factor of ${n \choose k}$, relative to the other probabilities contained within the interference term. 

\subsection{A classical simulation algorithm for recycled probabilities}
\label{sec:properties}

In this section, we show that recycled probabilities  constructed from output statistics where most photons were lost (i.e. when $n-k$
is a constant independent of $n$) are efficiently classically computable, and therefore not useful for obtaining interesting mitigation performance. Furthermore, we provide evidence that recycled probabilities constructed from $n-k$ output statistics where $k$ is a constant independent of $n$ (corresponding to output statistics where a small number of photons have been lost) are hard to compute classically, making these interesting to use for recycling mitigation protocols.

Recall that, up to normalisation, a recycled probability is a sum of the form  $$\sum_{s_i^{n-k} \in \mathcal{L}(s^n_l)}p(s^{n-k}_i),$$
where the number of terms of this sum is ${n \choose k}$. For Haar-random interferometers $U$, in the no-collision regime where $m=\Omega(n^5)$, we can use the results of \cite{aaronson_bosonsampling_2016}, and in turn express each $p(s^{n-k}_i)$  as 
$$p(s^{n-k}_i)=\frac{1}{{n \choose k}}\sum_{i}\frac{|\mathsf{Per}(X_i)|^2}{m^n}, $$
where the $X_i$'s are  $n-k \times n-k$ matrices with independently distributed Gaussian entries \cite{aaronson_computational_2011}, and the number of terms in this sum is also ${n \choose k}$.
We show that in the high loss regime, where $k=n-r$, and $r$ is a constant independent of $n$, the sum $\sum_{s_i^{n-k} \in \mathcal{L}(s^n_l)}p(s^{n-k}_i)$ is computable efficiently classically. Our result is encompassed in the following lemma proven in appendix \ref{applemm1}.
\begin{lemma} \label{ryser_lemma}
Let $k=n-r$, there is a classical algorithm running in time $O(2^{r-1}r\big({n \choose n-r}\big)^2)$  which exactly computes $\sum_{s_i^{n-k} \in \mathcal{L}(s^n_l)}p(s^{n-k}_i)$.
\end{lemma}
Notice that when $r$ is a constant independent of $n$,  the runtime of the above algorithm of the order of ${n \choose n-r} \in O(n^r)$ which is a polynomial in $n$. Thus, recycled probabilities corresponding to a high number of lost photons are efficiently computable classically.  
For low values of loss, in particular when $r$ scales with $n$, the above efficient classical simulability results break down, this however does not necessarily imply that there is no other, possibly efficient, algorithm for computing $\sum_{s_i^{n-k} \in \mathcal{L}(s^n_l)}p(s^{n-k}_i)$.

For the case where $k$ is a constant independent of $n$, there is  evidence that computing  $\sum_{s_i^{n-k} \in \mathcal{L}(s^n_l)}p(s^{n-k}_i)$  is probably hard (inefficient) to do classically (see also the discussion in Section \ref{sec:properties_mitigated}).  Indeed,  it is known that in worst-case the probabilities $p(s^{n-k}_i)$ when $k$ is a constant independent of $n$ are not efficient to compute classically, unless the polynomial hierarchy collapses to its third level \cite{aaronson_bosonsampling_2016}. 
This motivates the application of recycling mitigation in the low $k$ regime, as the $n-k$--photon output statistics are efficiently computable classically when $k$ is high, and therefore any potential quantum advantage is lost. In all our simulations, we apply our error mitigation techniques to statistics where the number of lost photons is a constant independent of system size, and discard all other statistics.

\subsection{Bounding the bias error}
\label{sec:bounding}

In this section, we give bounds on the bias error induced when replacing the interference term with its expectation value. This bias error will be present in all the error mitigation techniques  we derive later on. We provide  probabilistic bounds on this error, using statistical inequalities. Also, we provide a \emph{deterministic} upper bound by using the results of \cite{berkowitz_stability_2018}  which hold for specific families of unitary matrices.

The recycled probabilities have a composite structure, comprising a mixture of the ideal probability and the interference term. While the interference term, 
\begin{equation}
I_{s^n_l,k}=\sum_{s_i^{n-k} \in \mathcal{L}(s^n_l)}\sum_{ s^n_j \in \mathcal{G}(s_i^{n-k}), j \neq l}p(s^n_j)\frac{1}{N'_k},
\end{equation}
is itself a mixture of $n$--photon output probabilities. 
It is possible to account for the errors caused by using approximations of the interference terms when generating the mitigated probabilities by upper bounding the deviation of the interference term away from an expected value. 

Firstly, we consider the case of Haar random matrices. As we are working in the no-collision regime, the output probabilities generated by sampling from a Haar random matrix is linked to permanents of Gaussian random matrices $X \in \mathcal{G}_{n \times n}$ \cite{aaronson_bosonsampling_2016} with $\mathcal{G}_{n \times n}$ the set of all such Gaussian matrices; more precisely the set of complex matrices whose real and imaginary parts are chosen independently from the normal distribution $\mathcal{N}(0, \frac{1}{2})$.  Let $p_{unif}:=\frac{1}{{m \choose n}}$, in appendix \ref{appintbounds} we show the following.

\begin{lemma} \label{expectation_haar}
For all $s^n_l$, we have
\begin{equation}
\begin{split}
\mathbf{E}_{X \in \mathcal{G}_{n \times n}}\big(I_{s^n_l, k}\big) & =  p_{unif}\big(1 + g(m) \big) \approx p_{unif} ,\\
\end{split}
\end{equation}
where $g(m)\in O(m^{-1}) $, with $\mathbf{E}_{X \in \mathcal{G}_{n \times n}}(.)$ the expectation value over the set $\mathcal{G}_{n \times n}$.
\end{lemma}

With Lemma \ref{expectation_haar} in hand, it is possible to bound the deviation of the interference terms around this expected value according to the  theorem shown in Appendix \ref{appintbounds}.

\begin{theorem} \label{Haarbound}
The deviation of interference terms around $p_{unif}$ for Haar random matrices is bounded 
\begin{equation}
Pr\Big(\big|I_{s^n_l,k}-p_{unif}\big| \geq \epsilon_{bias,s^n_l}\Big) \leq \frac{np^2_{unif}}{\epsilon^2_{bias,s^n_l}},
\end{equation}
where $\epsilon_{bias,s^n_l}$ is a positive real number.
\end{theorem}

Secondly, as it is desirable not to be restricted to only mitigating the output of Haar random matrices, we derive an additional bound for arbitrary matrices.  Let $D^k_R$ be the uniform distribution over recycled probabilities $\{p^k_R(s^n_l)\}_{l}$. Since every recycled probability has an associated interference term $I_{s^n_l,k}$, one can equivalently think of $D^k_R$ as a distribution over interference terms. A random variable $Y$ chosen from $D^k_R$ means choosing, with uniform probability $p_{unif}$, a value from the set $\{p^k_R(s^n_l)\}_{l}$  (or, equivalently, choosing a value from $\{I_{s^n_l,k}\}_l$ uniformly randomly). 
In Appendix \ref{appintbounds}, we show the following.

\begin{lemma}\label{expectation_arbitrary}

\begin{equation}
\begin{split}
\mathbf{E}_{D^k_R}\big(I_{s^n_l, k}\big) & = p_{unif},\\
\end{split}
\end{equation}
where $\mathbf{E}_{D^k_R}(.)$ denotes the expectation value over $D^k_R$.
\end{lemma}
The deviation of interference terms around this expected value is upper bounded according to the following inequality.
\begin{theorem} \label{nonHaarbound}
The deviation of interference terms around $p_{unif}$ for an arbitrary matrix is bounded 
\begin{equation}
\begin{split}
\text{Pr}\Big(\big|I_{s^n_l, k}-  p_{unif}\big| \geq \epsilon_{bias,s^n_l} \Big) & \leq \frac{ p_{unif}}{\epsilon_{bias,s^n_l}^2} ,\\
\end{split}
\end{equation}
where $\epsilon_{bias,s^n_l}$ is a positive real number.
\end{theorem}
While both these upper bounds use Chebyshev's inequality \cite{lin_probability_2011}, the Thm. \ref{Haarbound} bound applies only to Haar random matrices, while the Thm. \ref{nonHaarbound} bound applies to arbitrary matrices. 
However, to upper bound the variance for the Haar random case we use exact moments originally derived in \cite{nezami_permanent_2021}. 
While in the case of arbitrary matrices the Bhatia-Davis inequality \cite{bhatia_better_2000} is instead used, which results in a looser bound. 

We note that a tighter, although distribution-dependent, bound may be found than the one provided in Thm. \ref{nonHaarbound}. 
This result follows from Lemmas \ref{Var_recycle_upperbound} and \ref{Var_interfer_upperbound}, proven in Appendix \ref{appintbounds}, which will now be stated, that use the definition of the variance that for a set of real values $\{x_i\}_i$ with mean $\mu$ and cardinality $N$, $\mathsf{Var}(\{x_i\}_i):=N^{-1}\sum_i(x_i - \mu)^2$.

\begin{lemma} \label{Var_recycle_upperbound}
The variance of the set of recycled probabilities is less than or equal to the variance of the set of ideal probabilities, that is
\begin{equation}
\mathsf{Var}\big(\{p(s^{n}_l)\}_l\big) \geq \mathsf{Var}\big(\{p_{R}^k(s^{n}_l)\}_l\big).
\end{equation}
\end{lemma} 
\begin{lemma} \label{Var_interfer_upperbound}
The variance of the set of interference terms is less than or equal to the variance of the set of ideal probabilities, that is
\begin{equation}
\mathsf{Var}\big(\{p(s^{n}_l)\}_l\big) \geq \mathsf{Var}\big(\{I_{s^n_l, k}\}_l\big).
\end{equation}
\end{lemma} 
Using these lemmas, an upper bound on the largest probability of the output distribution may be used to derive a tighter upper bound on the variance using the Bhatia-Davis inequality \cite{bhatia_better_2000}. Let $p_{\text{upper}}$ be an experimentally derived upper bound on the largest probability in the ideal $n$--photon output distribution $p_{\text{max}}$, such that $p_{\text{upper}}\geq p_{\text{max}}$. 
Following similar steps as for the proof for Thm. \ref{nonHaarbound} results in an upper bound on the confidence of $1 - \frac{ p_{unif} p_{{upper}}} {\epsilon_{bias,s^n_l}^2} - \delta\Big(1 - \frac{ p_{unif} p_{{upper}}}{\epsilon_{bias,s^n_l}^2}\Big) $, where $\delta$ is the confidence parameter for the $p_{upper}$ estimator. This result is stated formally in the following.

\begin{theorem} \label{p_upper_BD_bound}
\textit{The deviation of interference terms around $p_{unif}$ for an arbitrary matrix is bounded }
\begin{equation*}
\begin{split}
\text{Pr}\Big(\big|I_{s^n_l, k}-  p_{unif}\big| \geq \epsilon_{bias,s^n_l} \Big) & \leq \frac{ p_{unif} p_{{upper}}} {\epsilon_{bias,s^n_l}^2} + \delta\Big(1 - \frac{ p_{unif} p_{{upper}}}{\epsilon_{bias,s^n_l}^2}\Big)  ,\\
\end{split}
\end{equation*}
where \textit{$\epsilon_{bias,s^n_l}$ is a positive real number, and $p_{\text{upper}}$ is an empirically computed upper bound on the largest probability of the ideal $n$ output photon probability distribution with confidence $1-\delta$.}
\end{theorem}

A proof of this result is given in appendix \ref{appintbounds}. We note also that as $p_{\text{max}}\leq 1$ with confidence $1$, so that $\delta=0$ in Thm. \ref{p_upper_BD_bound}, Thm. \ref{nonHaarbound} follows as a corollary.

With assumptions about the structure of the unitaries, like in \cite{berkowitz_stability_2018}, it is possible to deterministically and exponentially upper bound the interference terms.

For a large class of unitary matrices, we show that the bias error scales as an inverse exponential in $n$. As will be seen later, this shows that our mitigation techniques outperform postselection in estimating output probabilities and expectation values of linear optical circuits for up to inverse exponential precisions. Here we use the the operator 2-norm and the infinity norm. For an $n\times n$ matrix $A$ the operator 2-norm is defined $\|A\|_2 := \sup_{\|\vec{x}\|_2\leq 1, \vec{x} \in \mathbb{C}^n}\|A\vec{x}\|_2$, where $\|\vec{v}\|_p$ is the $l_p$ norm, that is $\|\vec{v}\|_p = \big(\sum_p|v_i|^p\big)^{1/p}$. And the infinity norm which is defined as $\|\vec{v}\|_{\infty} := \max_i |v_i|$. Let $h^A_{\infty}:=\frac{1}{n} \sum_{i=1, \dots, n} \|\mathbf{A}_i\|_{\infty}$, where $\mathbf{A}_i$ is the $i$th row of $A$.  In Appendix \ref{appbreaking}, we use a result of \cite{berkowitz_stability_2018} to show the following.
\begin{theorem} \label{exp_barrier_thm}
For the  class of unitary matrices $U$ with submatrices $A$ such that $p_{max}:=\mathsf{max}_{s^n_l}(p(s^n_l))=|\mathsf{Per}(A)|^2$, and where these matrices $A$ satisfy $\frac{h^A_{\infty}}{\|A\|_2}\ll 1$, the bias error $\epsilon_{bias,s^n_l}$ is bounded
\begin{equation}
\epsilon_{bias,s^n_l} \in  O(e^{-2\times 10^{-5} n}).
\end{equation}
\end{theorem}

Interestingly, the probabilistic bounds of Theorems \ref{Haarbound} and \ref{nonHaarbound} also give, with high-confidence, exponentially low bounds on the bias error. 
Indeed, by setting for example $\epsilon_{bias,s^n_l} \in \Omega(\sqrt{p_{unif}})$ in Theorem \ref{nonHaarbound}, we observe that $|I_{s^n,k}-p_{unif}|$ is upper bounded by $\Omega(\sqrt{p_{unif}})$ with confidence $1-o(1)$. 
Intuitively, and as will be seen later in more details, both the probabilistic and deterministic upper bounds on the bias error indicate that recycling mitigation outperforms post-selection for up to inverse-exponential in $n$ precision. 

Appendix \ref{sec:prospects} discusses possible directions to improve our performance guarantees, and contains a technical result about sums of permanents of i.i.d. Gaussian matrices \cite{aaronson_computational_2011} that might find use beyond this work.

\subsection{Generating the loss-mitigated outputs}
\label{sec:generating}

We now present two methods for constructing loss mitigated outputs from the recycled probabilities, these we refer to as \emph{linear solving} and \emph{extrapolation}.
In linear solving, 
the interference term within each recycled probability is substituted for its expected value, and the resulting expressions are then solved to find estimators of the ideal probabilities. While in extrapolation, the decay of the ideal signal in the set of recycled distributions with $k$ is used to compute estimators of the ideal probabilities.

While the postprocessing required to generate each mitigated value may be performed efficiently, the size of the output distribution is exponential in $m$ and $n$. Meaning the classical postprocessing  cost (i.e. the memory cost) is proportional to the number of mitigated probabilities that are to be generated. So that to generate a full mitigated output distribution this scales as $m^n$, while for a subset of size $N_s$ the postprocessing cost is then proportional to $N_s$.

\subsubsection{Linear solving}

The linear solving method involves substituting the interference term within each recycled probability for an approximate value, and then solving the resulting expressions for the ideal probabilities. As the expectation of interference terms over the Haar measure and over the recycled distribution are $p_{unif} + O(m^{-n+1}) $ and $p_{unif}$, respectively, see Lemmas \ref{expectation_haar} and \ref{expectation_arbitrary}, the term $\frac{N'_k}{N_k}p_{unif}$ is used for this substitution. This allows the upper bounding of the error introduced by the substitution using Thms. \ref{Haarbound} and \ref{nonHaarbound}. 
Each recycled probability constructed from the $n-k$ output photon statistics may then be written 
\begin{equation}
    \tilde{p}_R^k(s^n_l)=\frac{p(s^n_l)}{{m-n+k \choose k}}+\frac{N'_k}{N_k}\Big(p_{unif}+\epsilon_{bias,s^n_l}\Big)+\epsilon_{stat.,s^n_l},
\end{equation}
where $\frac{N'_k}{N_k}\epsilon_{bias,s^n_l}$ is the bias error introduced by replacing the interference term in recycled probability $p_R^k(s^n_l)$ with $p_{unif}$. And $\epsilon_{stat.,s^n_l}$ is the statistical error from estimating the recycled probabilities from a finite number of samples. These new expressions can then be solved to generate the mitigated outputs
\begin{equation}
\label{eqmitlinsolv}
    p_{\text{miti}}(s^n_l) = {{m-n+k \choose k}}\bigg|\tilde{p}_R^k(s^n_l) - \frac{N'_k}{N_k}p_{unif} \bigg| .
\end{equation}
We now provide pseudocode with the steps required to use the linear solving method to generate the mitigated output.  

\vspace{1em}

\begin{algorithm*}[H]
\label{algolinsolv}
\DontPrintSemicolon
\SetAlgoLined
\SetKwInOut{Input}{input}\SetKwInOut{Output}{output}
\SetKwRepeat{Repeat}{repeat}{until}
\Input{Recycled probability estimator $\tilde{p}^k_R(s^n_l)$, and uniform probability $p_{unif}$.}
\BlankLine\
Initialise variable $p_{\text{miti}}(s^n_l)\gets 0$ for the mitigated output to be computed. \;
Update mitigated output variable as $p_{\text{miti}}(s^n_l) \gets {{m-n+k \choose k}}\big|\tilde{p}_R^k(s^n_l) - \frac{N'_k}{N_k}p_{unif} \big|$.\;
\Output{Mitigated output $p_{mit}(s^n_l)$\;}
\caption{Linear solving method}
\end{algorithm*}

\vspace{1em}

There exists a regime of recycling mitigation usefulness, where the combined bias and the statistical errors present in the mitigated probabilities are lower than the statistical errors of the postselected distribution. Indeed, from Equation \ref{eqmitlinsolv}, it can be seen that the bias error of linear solving is upper bounded by $$M_{bias} \in O\bigg({m-n+k \choose k}\bigg) \epsilon_{bias,s^n_l},$$ we formally prove this in Appendix \ref{applinearsolv}. 
Similarly, and as stated  in Section \ref{sec:prelim}, we show in Appendix \ref{applinearsolv} that statistical errors for a recycled probability constructed from $n-k$ photon statistics in linear solving is upper bounded by $M_{stat, k} \in O\bigg(\sqrt{\frac{{m \choose n}}{{n \choose k}(1-\eta)^{n-k}\eta^kN_{tot}}}\bigg)$, the ${m \choose n}$ can be replaced with $N_s$, if one is interested in mitigating a subset $N_s$ of probabilities. With these upper bounds in hand, one can find a condition on the number of samples up to  which recycling mitigation  outperforms postselection by solving for $N_{tot}$ in the following inequality
\begin{equation*}
    M_{bias}+M_{stat,k} \leq M_{stat,k=0}.
\end{equation*}
More details on this can be found in Appendix \ref{applinearsolv}.

\bigskip
We now introduce the notion of dependency. This quantifies the correlation of the interference terms with the ideal probability within the recycled probabilities. A positive correlation means that the signal for the ideal probability is greater than ${{m-n+k} \choose k}^{-1}$, which can be used to improve the performance of linear solving. Each recycled probability can be rewritten to include a dependency term $d_k(s_l^n)$ quantifying this correlation, by reformulating the interference term as a linear function of $p(s_l^n)$ and $p_{unif}$. The expression
\begin{equation}
I_{s^n_l, k} =  \big(1-d_k(s^n_l)\big)p_{unif} + d_k(s^n_l)p(s^n_l) 
\end{equation} 
defines the dependency term $d_k(s^n_l)$ of each recycled probability, and $k\leq n-1$. An average dependency term over the distribution, denoted $d_k$, may be  calculated from the recycled probabilities (see eqn. \ref{eqexpdk} and appendix \ref{applinearsolv}). Each recycled probability constructed from the $n-k$ output photon statistics may then be expressed in the form
\begin{equation}
\begin{split}
\tilde{p}_R^k(s^n_l) & = \frac{p(s^n_l)}{{m-n+k \choose k}}+\frac{N'_k}{N_k}\big(\big((1-d_k)p_{unif}+d_kp(s^n_l) + \epsilon_{bias,s^n_l} \big) + \epsilon_{stat.,s^n_l},\\
\end{split}
\end{equation}
where $\epsilon_{bias,s^n_l}$ is the bias error introduced by replacing the interference term in recycled probability $p_R^k(s^n_l)$ with $\big((1-d_k)p_{unif}+d_kp(s^n_l)\big)$. 
Note that if the computed estimator for $d_k$ is negative or greater than $1$ 
then the dependency approach should be aborted and the original version of linear solving used. We conjecture it is always the case that $1 \geq d_k \geq 0$. The following pseudocode details the steps required to perform the linear solving with dependency method and generate the mitigated output. 

We note that although the upper-bounds derived in Appendix \ref{applinearsolv} on the statistical and bias errors of linear solving with dependency are similar to those for linear solving without dependency, nevertheless  numerical simulations show that linear solving with dependency reliably outperforms linear solving without dependency (see Fig. \ref{fig:performancecomps} (a) and (b)). 

\vspace{1em}

\begin{algorithm*}[H]
\label{algolinsolvdep}
\DontPrintSemicolon
\SetAlgoLined
\SetKwInOut{Input}{input}\SetKwInOut{Output}{output}
\SetKwRepeat{Repeat}{repeat}{until}
\Input{Recycled probability estimator $\tilde{p}^k_R(s^n_l)$, uniform probability $p_{unif}$, an estimator for the absolute average deviation for the $n$ output photon distribution $\tilde{D}_0$, and an estimator for the absolute average deviation for the $n-k$--photon recycled distribution $\tilde{D}_k$ (eqn.  (\ref{eqabsavdk})).}
\BlankLine\
Initialise variables for the mitigated output $p_{mit}(s^n_l) \gets 0$ and the average dependency term $d_k \gets 0$.\;
Update the average dependency term variable $d_k \gets \frac{1}{{{m-n+k} \choose k} -1}\bigg(\frac{{{m-n+k} \choose k}\tilde{D}_k}{\tilde{D}_0} - \frac{1}{{{m-n+k} \choose k}}\bigg)$.\\
Update the mitigated output variable $p_{mit}(s^n_l) \gets \Bigg|\frac{\tilde{p}_R^k(s^n_l) +  \frac{N'_k}{N_k}\big((-1 + d_k )p_{unif} \big) }{{{m-n+k \choose k}}^{-1} + \frac{N'_k}{N_k}d_k }\Bigg|$\;
\Output{Mitigated output $p_{mit}(s^n_l)$\;}
\caption{Linear solving with dependency method}
\end{algorithm*}

\vspace{1em}

\subsubsection{Extrapolation}

We now present methods by which extrapolation may be used to generate loss-mitigated outputs. 
From the definition of the recycled probability,
\begin{equation} \label{decay_equation}
    p_R^k(s^n_l)=\frac{p(s^n_l)}{{m-n+k \choose k}}+\frac{N'_k}{N_k}I_{s^n_l, k},
\end{equation}
the magnitude of the ideal probability signal within 
the recycled probabilities 
is proportional to ${m-n+k \choose k}^{-1}$. 
As $k$ increases, the ideal signal magnitudes decrease as the recycled distributions converge towards uniform. 
The rate of decay of the ideal signal with $k$ can be computed and then used to extrapolate mitigated outputs. 
We present two variations of extrapolation in which different types of dependence of the ideal probability signal on number of lost photons $k$ are considered. A linear dependence is used for a linear extrapolation method, and an exponential dependence for an exponential extrapolation method. 
These two functions were considered primarily for heuristic reasons as seeming to reflect the decay behaviour observed in the recycled distributions.
Both the linear and exponential extrapolation methods involve two iterations of optimisation. The first iteration computes an average decay parameter using the set of average absolute deviations of the recycled distributions. Where, for the set of recycled probabilities $\{p^k_R(s^n_l)\}_{l}$ constructed from $k$ photon statistics, the average absolute deviation is defined as 
\begin{equation}
\label{eqabsavdk}
\begin{split}
D_k & := {m \choose n}^{-1}\sum_{l }\abs{{p}_{R}^k(s^n_l) - p_{unif}}.\\
\end{split}
\end{equation}
The second iteration then uses the decay parameter to compute the mitigated values. 

Linear extrapolation applies a linear model function to compute mitigated outputs. Here the least squares method is used to identify optimal parameters to fit a linear model function to data. 
Parameter optimisation is performed by minimising the sum of the squared residuals, where a residual is the difference between a data point and the model. 
A data set of $N$ points is denoted $\{x_i,y_i\}_{i=1}^N$, where $\{x_i\}_{i=1}^N$ are the independent variables and $\{y_i\}_{i=1}^N$ are the dependent variables. The model function $f(x,\boldsymbol{\alpha})$ is optimised by varying the $\boldsymbol{\alpha}$ parameters to approximate the relation between independent and dependent variables found in the data set. 
The residual for each data point is defined $r_i := y_i - f(x_i,\boldsymbol{\alpha})$.
The sum of the squared residuals is minimised to generate the optimal parameters 
\begin{equation}
\boldsymbol{\alpha}_{\text{min}}=\mathop{\arg \min}\limits_{\boldsymbol{\alpha}}\sum_{i=1}^N r_i^2,
\end{equation}
which are used to generate the optimised model function $f(x,\boldsymbol{\alpha}_{\text{min}})$. This can then be used to make predictions about data outside the range of the data set used for optimisation.

In linear extrapolation, the linear model function used for the first iteration of linear least squares is 
\begin{equation}
f(x,g_{\text{avg}}) = -g_{\text{avg}}x + \tilde{D}_0,
\end{equation}
where $\tilde{D}_0$ is the average absolute deviation of the $n$-photon distribution from uniform computed from the  statistics corresponding to postselecting on detecting all $n$ photons, and $g_{\text{avg}}$ is the optimal global linear decay parameter.  
The set $\{k,\tilde{D}_k\}_{k=1}^K$ is used as the data set to compute $g_{\text{avg}}$. Where $\{\tilde{D}_k\}_{k=1}^K$ is the set of average absolute deviations from uniform for the different distributions, for $K\leq n$, and $\tilde{D}_k$ is the estimated (from statistics where $k$ photons were lost) absolute average deviation for the $n-k$--photon recycled distribution. 
After the optimal decay parameter is identified, another iteration of least squares is performed with updated model functions this time to generate the mitigated values. 
The data set used in this step is $\{k, \tilde{p}^k_R(s^n_l)- p_{unif}\}^K_{k=1}$. For the second iteration of linear least squares, each output bit string is assigned a linear model function of the form 
\begin{equation}
f_{s_n}(x,\alpha_{s^n_l})= \text{sgn}(p_{unif}-\tilde{p}^{k=1}_R(s^n_l))g_{\text{avg}}x + \alpha_{s^n_l}.
\end{equation}
For each output bit string $s_n$ an optimal $\alpha_{s^n_l}$ is computed, generating the set $\{\alpha_{s^n_l}\}_l$, and the set of mitigated outputs is then $\{\alpha_{s^n_l}+p_{unif}\}_{l}$. 
In Appendix \ref{appextrap}, error bounds are derived for the use of linear extrapolation to perform recycling mitigation, these indicate the existence of a regime where the mitigation outperforms postselection.
As with the analytical linear solving error bounds, the linear extrapolation error bounds allow the computation of an estimate of the number of samples up to which recycling mitigation with linear extrapolation outperforms postselection by solving for $N_{tot}$ in the following inequality
\begin{equation*}
    M_{bias}+M_{stat,k} \leq M_{stat,k=0}.
\end{equation*}
We now provide pseudocode for applying the linear extrapolation method to generate mitigated outputs.

\vspace{1em}

\begin{algorithm*}[H]
\label{algolslin}
\DontPrintSemicolon
\SetAlgoLined
\SetKwInOut{Input}{input}\SetKwInOut{Output}{output}
\SetKwRepeat{Repeat}{repeat}{until}
\Input{The number of data points $n_d \in \{n_d \in \mathbb{Z}^+ | n_d<n\}$ to be used in both iterations of least squares, the data set $\{k,\tilde{D}_k\}_{k=1}^{n_d}$ used to compute the gradient parameter $g_{\text{avg}}$ in the first iteration of least squares, $\tilde{D}_0$, and  the data set $\{k,\tilde{p}^k_R(s^n_l)-p_{unif}\}_{k=1}^{n_d}$ used to compute the mitigated output in the second iteration of least squares.}
\BlankLine\
Initialise an average decay parameter variable $\tilde{g}_{\text{avg}} \gets 0$, a prefactor variable $\alpha_{s^n_l} \gets 0$, and a mitigated output variable $p_{mit}(s^n_l) \gets 0$.  \;
Use least squares method with model function $f(x_i,g_{\text{avg}}) = -\tilde{g}_{\text{avg}}x_i + \tilde{D}_0$ and data set $\{x_i,y_i\}_{i=1}^{n_d}:=\{k,\tilde{D}_k\}_{k=1}^{n_d}$ to compute the value of the average decay parameter (slope), and assign this to variable $\tilde{g}_{\text{avg}}$.\;
Use least squares method with model function $f_{s^n_l}(x_i,\alpha_{s^n_l})= \text{sgn}(p_{unif}-\tilde{p}_R^{k=1}(s^n_l))\tilde{g}_{\text{avg}}x_i + \alpha_{s^n_l}$ and data set $\{x_i,y_i\}_{i=1}^{n_d}:=\{k,\tilde{p}^k_R(s^n_l)-p_{unif}\}_{k=1}^{n_d}$ to compute the value of the y-axis intercept, and assign this to variable $\alpha_{s^n_l}$.\;
Update mitigated output variable as $p_{mit}(s^n_l) \gets p_{unif} + \alpha_{s^n_l}$.\;
\Output{Mitigated output $p_{mit}(s^n_l)$\;}
\caption{Linear extrapolation method}
\end{algorithm*}

\vspace{1em}

The method for exponential extrapolation is broadly similar. However, non-linear least squares or a non-linear numerical optimisation method (e.g. the Levenberg–Marquardt algorithm \cite{levenberg_method_1944, marquardt_algorithm_1963}) is instead used to compute the decay factor and the mitigated probabilities. The exponential model function used for the first step is
\begin{equation}
f(x,\alpha_{\text{avg}}) = \tilde{D}_0 e^{-\alpha_{\text{avg}}x},
\end{equation}
where $\alpha_{\text{avg}}$ is the optimal global exponential decay parameter \footnote{Note that, in order to satisfy $f(n,\alpha_{\mathsf{avg}})=0$, the fitting function should be $f(x,\alpha_{\mathsf{avg}})=\tilde{D}_0 e^{-\alpha_{\text{avg}}x}+ \beta$ with $\beta<0$. However, as we are intersted in the decay of $f(.)$ with increasing $k$, and we look at (far from $n$) values of $k \in O(1)$ in the fitting, we can exclude $\beta$ in the above expression. }. Where again the set $\{k,\tilde{D}_k\}_{k=1}^K$ is used as the data, this time to compute the optimal value of $\alpha_{\text{avg}}$. 
And then each output bit string for which a mitigated probability is to be generated is assigned a model function of the form 
\begin{equation}
f_{s_n}(x,\Lambda_{s_n})= \Lambda_{s^n_l} e^{-\alpha_{\text{avg}}x}+p_{unif}.
\end{equation}
An optimal prefactor, denoted $\Lambda_{s^n_l}^{\text{opt}}$, is computed for each bit string $s^n_l$ using numerical optimisation. This generates the set of prefactors $\{\Lambda_{s^n_l}^{\text{opt}}\}_l$, and the set of mitigated outputs is then $\{\Lambda_{s^n_l}^{\text{opt}}+p_{unif}\}_l$. 
The following pseudocode gives the steps required to use the exponential extrapolation method to generate the mitigated outputs. 

\begin{algorithm*}[H]
\label{algolsexp}
\DontPrintSemicolon
\SetAlgoLined
\SetKwInOut{Input}{input}\SetKwInOut{Output}{output}
\SetKwRepeat{Repeat}{repeat}{until}
\Input{The number of data points $n_d \in \{n_d \in \mathbb{Z}^+ | n_d<n\}$ to be used in both iterations of least squares, the set $\{k,\tilde{D}_k\}_{k=1}^{n_d}$ used to compute the gradient parameter $g_{\text{avg}}$ in the first iteration of least squares, $\tilde{D}_0$, and the data set $\{k,\tilde{p}^k_R(s^n_l)-p_{unif}\}_{k=1}^{n_d}$ used to compute the prefactor value in the second iteration of least squares.}
\BlankLine\
Initialise an average decay parameter variable $\alpha_{avg} \gets 0$, a prefactor variable $\Lambda_{s^n_l} \gets 0$, and a mitigated output variable $p_{mit}(s^n_l) \gets 0$.  \;
Use least squares method with model function $f(x,\alpha_{\text{avg}}) = \tilde{D}_0 e^{-\alpha_{\text{avg}}x}$ and data set $\{x_i,y_i\}_{i=1}^{n_d}:=\{k,\tilde{D}_k\}_{k=1}^{n_d}$ to compute the value of the average decay parameter and assign this to variable ${\alpha}_{\text{avg}}$.\;
Use least squares method with model function $f_{s_n}(x,\Lambda_{s^n_l})= \Lambda_{s^n_l} e^{-\alpha_{\text{avg}}x}+p_{unif}$ and data set $\{x_i,y_i\}_{i=1}^{n_d}:=\{k,\tilde{p}^k_R(s^n_l)-p_{unif}\}_{k=1}^{n_d}$ to compute the value of the prefactor and assign this to variable $\Lambda_{s^n_l}$.\;
Update mitigated output variable as $p_{mit}(s^n_l) \gets p_{unif} + \Lambda_{s^n_l}$.\;
\Output{Mitigated output $p_{mit}(s^n_l)$\;}
\caption{Exponential extrapolation method}
\end{algorithm*}

Note that the ideal signal magnitude decays proportionally with ${{m-n+k} \choose k}^{-1} \sim m^{-k}$ in eqn. \ref{decay_equation}, which intuitively motivates the choice of an exponential model function to reflect this decay behaviour. 
In the next sections, we provide numerical evidence indicating that there exists a non-trivial sampling regime where exponential extrapolation outperforms postselection. We also note that in the numerical simulations, extrapolation using an exponential model function consistently outperforms linear extrapolation (see Fig. \ref{fig:performancecomps} (c) and (d)). This may be a consequence of the exponential model function better reflecting the ideal signal decay behaviour.

As a final remark we comment on how one can estimate $D_k$ in practice. 
In \cite{aaronson_bosonsampling_2014}, it was shown that $D_0$ is lower bounded by a constant. 
The results of numerical simulations plotted in Fig. \ref{fig:D_k_plots} highlight that the same appears to hold for $D_k$ when $k \in O(1)$. 
Intuitively, this makes sense as $D_k$, similar to $D_0$, are expected to contain probabilities that are hard to simulate classically (see Section \ref{sec:properties_mitigated}) , and are thus far from the uniform distribution.  
In practice, this means that in order to estimate $D_k$ to a good ($\epsilon \in O(1)$) precision, one only needs to compute an $O(1)$ fraction of probabilities $p_R(s^n_l)$, with $s^n_l$ picked uniformly randomly, and compute the mean of $|p_R(s^n_l)-p_{unif}|$. 
This gives a  high-confidence estimate of $D_k$ by Hoeffding's inequality. 
Consequently, this means that estimating the $D_k$ to be used in extrapolation recycling mitigation does not require significantly more statistics than what is needed in linear solving.

\begin{figure}[]
\centering
\subfigure[]{
\includegraphics[width=.75\textwidth]{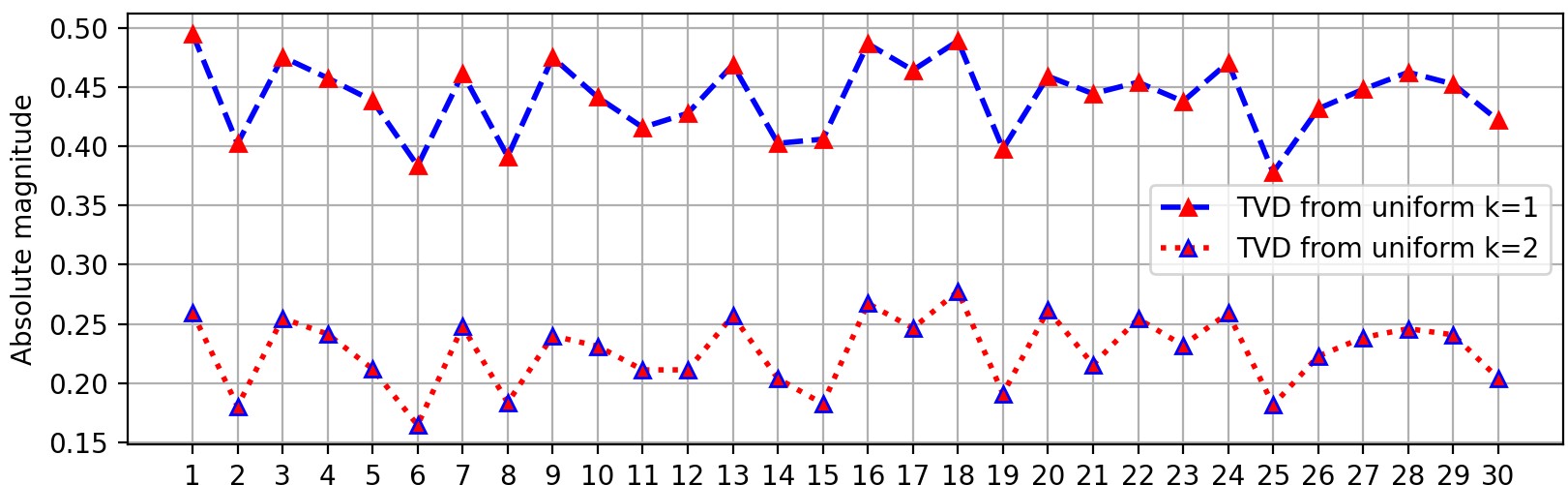}
}
\subfigure[]{
\includegraphics[width=.75\textwidth]{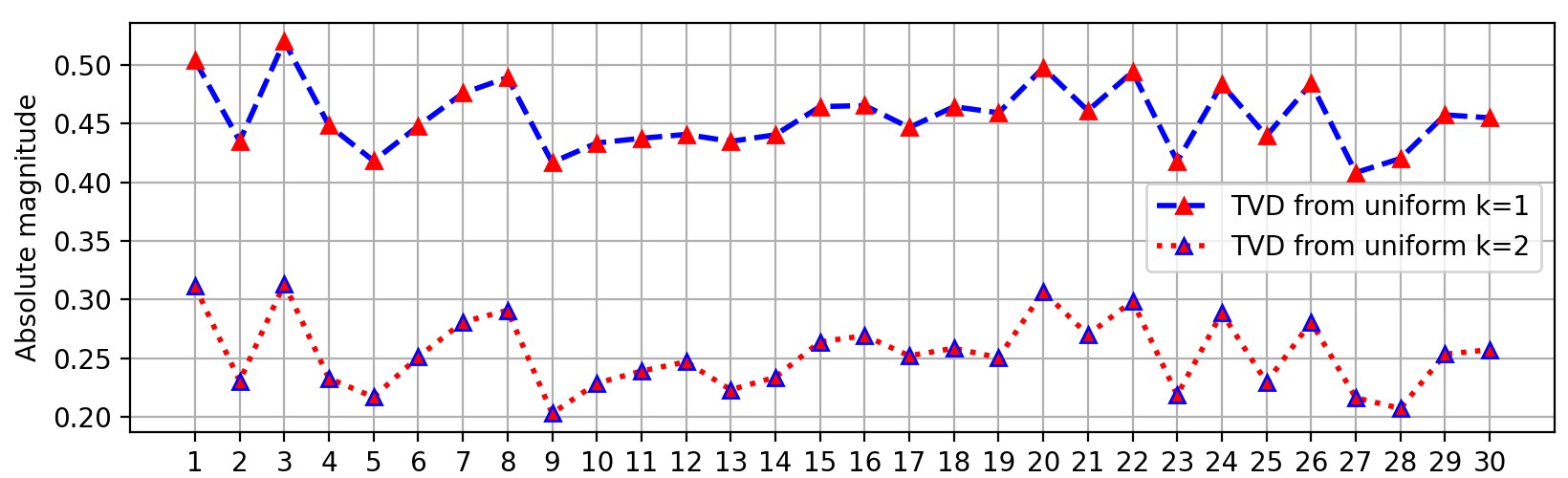}
}
\subfigure[]{
\includegraphics[width=.75\textwidth]{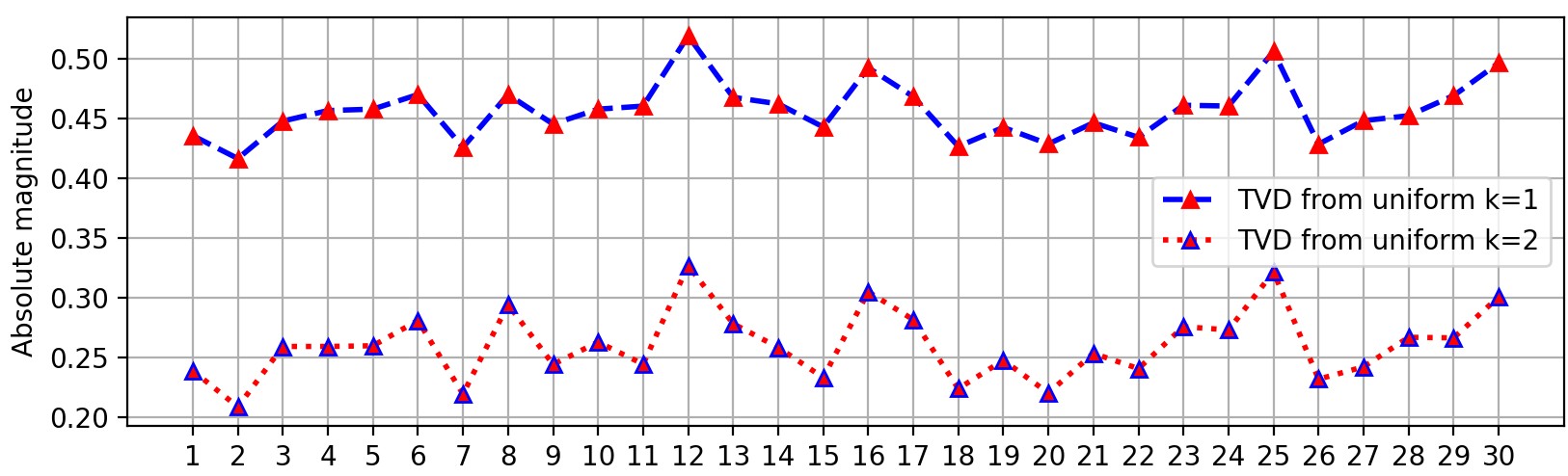}
}
\subfigure[]{
\includegraphics[width=.75\textwidth]{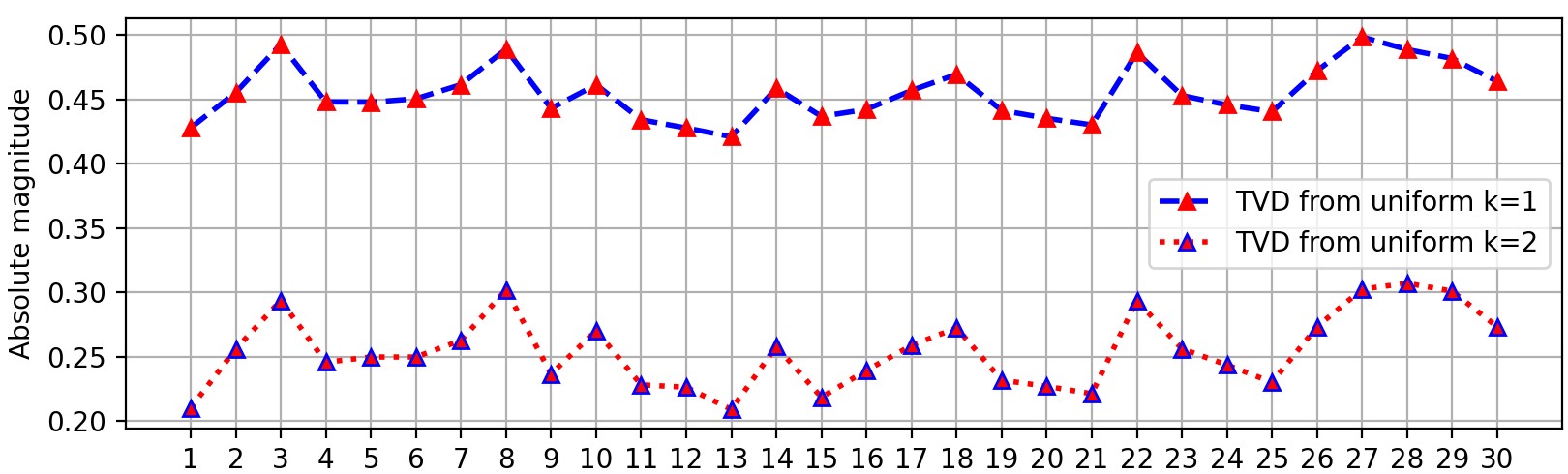}
}
\subfigure[]{
\includegraphics[width=.75\textwidth]{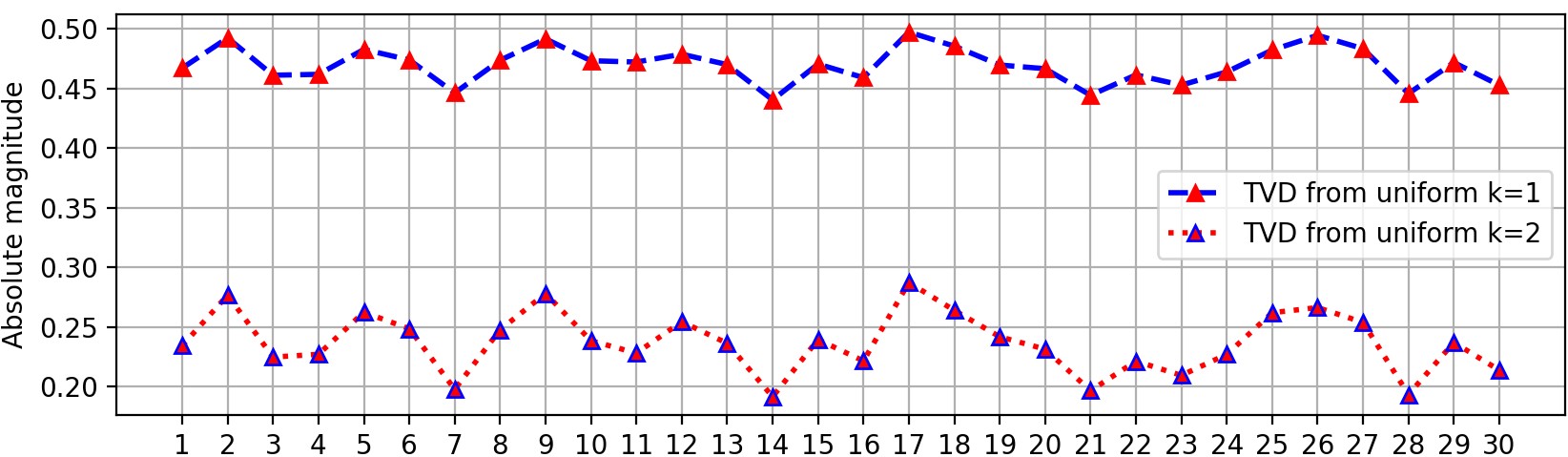}
}
\caption{The results of numerical simulations where $D_k$ for $k\in \{1,2\}$ is plotted for $30$ random unitaries with $m=20$ for (a) $n=3$, (b) $n=4$, (c) $n=5$, (d) $n=6$, and (e) $n=7$. From these results, as with \cite{aaronson_bosonsampling_2014} where it was shown that $D_0$ is lower bounded by a constant, it appears this also holds for $D_k$ when $k \in O(1)$. }
\label{fig:D_k_plots}
\end{figure}

\subsection{Normalising the mitigated distribution}
\label{sec:normalising}
The postprocessing techniques we have presented  take as input a set of recycled probabilities $\{p_R(s^n_l)\}_l$ and output a set of mitigated values $\{p_{mit}(s^n_l)\}_l$, or a subset of these values. We have made the distinction between \emph{values} and \emph{probabilities}, as the outputs  $\{p_{mit}(s^n_l)\}_l$ satisfy $p_{mit}(s^n_l)>0$ but are not normalised in general. That is, $\sum_lp_{mit}(s^n_l)=N$ with no guarantee that $N = 1$. 
Analytical and numerical calculations performed in previous sections have shown that, with high probability, $$||\vec{p}_{mit}-\vec{p}_{id}||_1 \leq ||\vec{p}_{post}-\vec{p}_{id}||_1,$$  and $|p_{mit}(s^n_l)-p(s^n_l)| \leq |p_{post}(s^n_l)-p(s^n_l)|$ up to a certain number of samples. $\vec{p}_{mit}$, $\vec{p}_{id}$, and $\vec{p}_{post}$, are vectors containing respectively the mitigated values, the ideal $n$--photon probabilities, and the probabilities obtained by postselection, and $||.||_1$ is the usual $l_1$ norm, $||\vec{v}||_1:=\sum_i|v_i|$.

In practice, one is usually interested in computing the expectation value of some observable $\mathsf{O}$, defined as $\langle \mathsf{O} \rangle=\sum_ip_iw_i$, where $w_i$ are some weights, and $\{p_i\}$ a subset of probabilities of the quantum circuit. In this case we can replace the $p_i$'s by the unnormalised mitigated values and obtain guarantees similar to those stated above. 
However, if we are interested in mitigating the entire distribution, then we would need to normalise $\vec{p}_{mit}$. The easiest way to do this, and what we do for our numerical simulations,  is to define, $\vec{p}_{mit,nor}:=\frac{1}{N}\vec{p}_{mit}$. One can easily check that $||\vec{p}_{mit,nor}||_1=1$, and therefore that it is a vector of normalised probabilities.

We will show that such a normalisation provably produces a distribution $\vec{p}_{mit,nor}$ which is closest possible to $\vec{p}_{mit}$
in $l_1$-norm.

Let $\vec{q}_m$ be a vector of positive values which is  a \emph{closest} (in $l_1$-norm) normalised vector to $\vec{p}_{mit}.$ More precisely,  $\vec{q}_m$ satisfies $$||\vec{q}_m-\vec{p}_{mit}||_1=\mathsf{min}_{\vec{q}\big|||\vec{q}||_1=1}||\vec{q}-\vec{p}_{mit}||_1.$$ 
An immediate observation, by definition, is that
$$||\vec{q}_m-\vec{p}_{mit}||_1 \leq ||\vec{p}_{id}-\vec{p}_{mit}||_1.$$ 
We now prove
\begin{lemma} \label{normalise_miti}
A valid choice of $\vec{q}_m$ is $\vec{q}_m=\vec{p}_{mit,nor}$.
\end{lemma}
\begin{proof}
    $||\vec{p}_{mit,nor}-\vec{p}_{mit}||_1=||\vec{p}_{mit,nor}-\vec{p}_{mit,nor}N||_1=|1-N|=\big|1-||\vec{p}_{mit}||_1\big|$.

    For any normalized vector of positive values $\vec{q}$, we can use a reverse triangle inequality to show
    $||\vec{q}-\vec{p}_{mit}||_1 \geq \big | ||\vec{q}||_1-||\vec{p}_{mit}||_1 \big | \geq \big|1-||\vec{p}_{mit}||_1\big| \geq ||\vec{p}_{mit,nor}-\vec{p}_{mit}||_1$. Thus $\vec{p}_{mit,nor}$ is a valid choice of $\vec{q}_m$, by definition of $\vec{q}_m$.
\end{proof}

\section{Numerical simulations}

\label{sec:numerics2}

\begin{figure}[]
\centering
\subfigure[]{
\includegraphics[width=.48\textwidth]{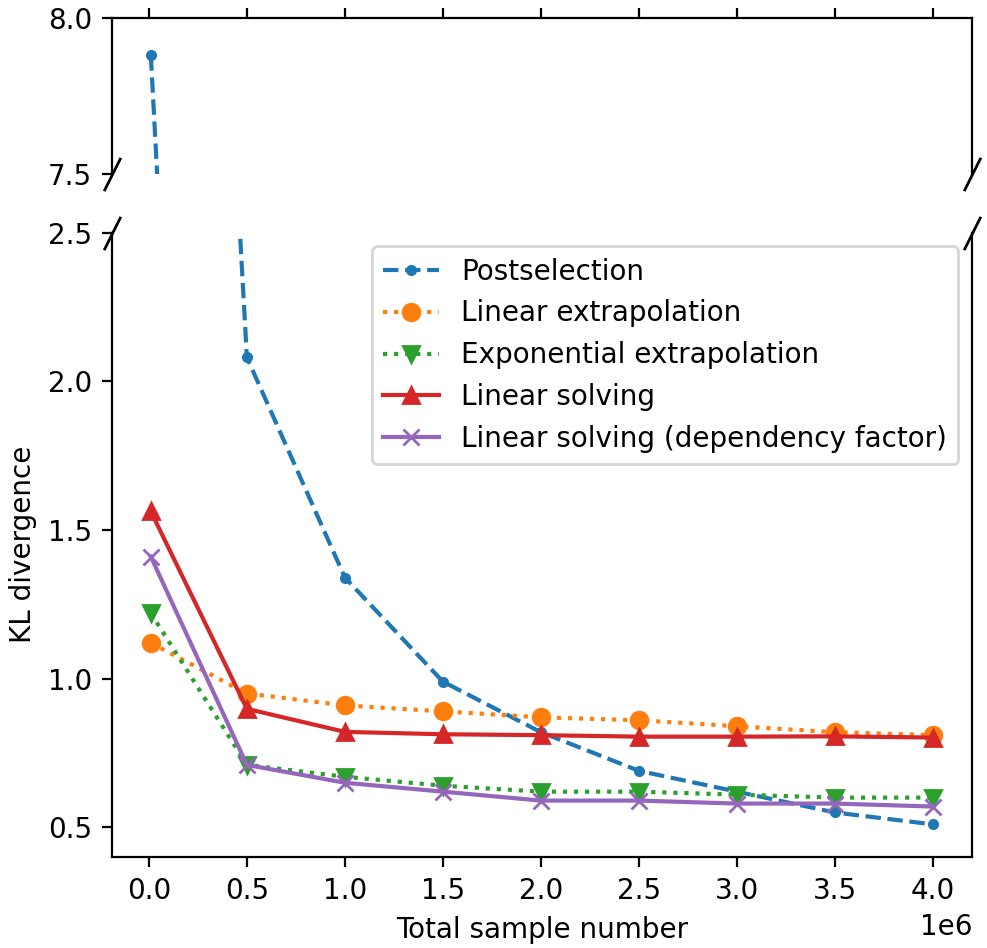}
}
\subfigure[]{
\includegraphics[width=.48\textwidth]{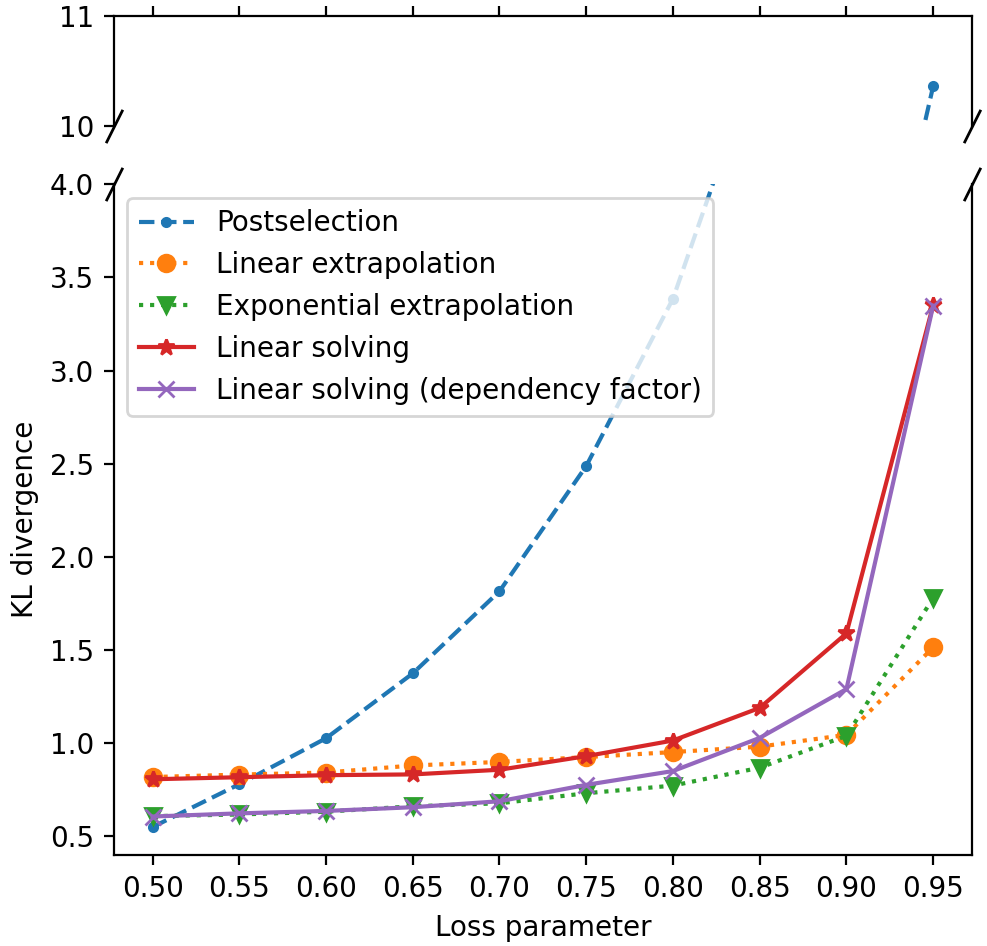}
}
\caption{A numerical performance comparison of the different methods of recycling mitigation and postselection for random unitary circuits with $m=20$ modes and $n=4$ photons.  (a) For a uniform loss parameter of $\eta=0.8$, the KL divergence from the ideal output distribution for linear extrapolation, exponential extrapolation, linear solving and linear solving with dependency versions of recycling mitigation and postselection is plotted against total sample number. (b) For a total number of samples of $N_{tot}=1 \times 10^5$ and a uniform loss parameter in the range $\eta \in [0.5,0.9]$, the KL divergence from the ideal output distribution of linear extrapolation, exponential extrapolation, linear solving and linear solving with dependency versions of recycling mitigation and postselection is plotted against total sample number.}

\label{fig:performancecomps}
\end{figure}

In the numerical simulations recycling mitigation is used to mitigate the effects of uniform photon loss error on the output of an otherwise ideal simulation of a linear optical quantum circuit. In the simulations, random matrices are chosen for each experiment which are decomposed and implemented with a linear interferometer, and a uniform photon loss model is applied with a defined loss parameter. The simulations were performed using Perceval \cite{heurtel_perceval_2023}, a pythonic framework for the simulation of photonic quantum circuits. Uniform photon loss channels commute with the interferometer, and so all loss channels, including loss due to imperfect sources and measurement, can be propagated to the end of the circuit. And so the effects of loss may be modelled by an ideal photon source, an ideal interferometer $M$, and a combined photon loss channel acting immediately before an ideal measurement operation. This means that rather than an output photon from the circuit incident on a detector being in the state $\ket{1}\bra{1}$, it is instead $$\ket{1}\bra{1} \to (1-p)\ket{1}\bra{1}+p\ket{0}\bra{0}.$$ With the probability of photon loss, $p$, the same for all output modes. 
Noise of this form may be considered analogous to the types of measurement noise commonly considered in circuit model quantum computing, as this is often modelled as an error channel followed by an ideal measurement.

For all the experiments in Fig. \ref{fig:performancecomps}, parameter settings of 20 modes and 4 photons were used. Fig. \ref{fig:performancecomps} (a) plots the performance of linear solving recycling mitigation, both with and without using the dependency factor, and postselection against the total number of samples used for a fixed uniform loss parameter of $\eta=0.8$. Linear solving outperformed postselection up to $\approx 1.1 \times 10^6$ samples, while linear solving with dependency outperformed postselection up to $\approx 3.1 \times 10^6$ samples. Note that ${20 \choose 4}^2 \approx 23 \times 10^6$, and thus recycling mitigation is outperforming postselection for sample sizes of the order of ${m \choose n}^2={20 \choose 4}^2$. Fig. \ref{fig:performancecomps} (b) plots the performance of the linear solving methods and postselection against changing loss parameter for a fixed number of samples of $1 \times 10^5$. Linear solving outperformed postselection for loss above $\approx 0.6$, and linear solving with the dependency factor above loss of $\approx 0.5$. Fig. \ref{fig:performancecomps} (c) plots the performance of linear extrapolation, exponential extrapolation and postselection against the total number of samples used for a fixed uniform loss parameter of $\eta=0.8$. Linear extrapolation outperformed postselection up to $\approx 1.7 \times 10^6$ samples, while exponential extrapolation does so up to $\approx 3.0 \times 10^6$ samples. Fig. \ref{fig:performancecomps} (d) plots the performance of the extrapolation methods and postselection against changing loss parameter for a fixed number of samples of $1 \times 10^5$. Linear extrapolation outperformed postselection for loss above $\approx 0.5$, and linear solving with the dependency factor above loss of $\approx 0.5$. These results indicate the exponential function is a better model for the physical behaviour of the decay than the linear function, we believe this may be explained in terms of the prefactor of the ideal output probability in eqn. \ref{decay_equation} decreasing exponentially (specifically $\propto (m-n+k)^{-k}\big)$ with increasing $k$. 
In current quantum linear optical devices losses including source, interferometer and measurement are commonly above $50$\% \cite{maring_general-purpose_2023}. Therefore, the evidence from these simulations indicates that recycling mitigation may be usefully applied to the current generation of linear optical devices. Also, in the previous section an upper bound on the number of samples was computed for a particular experiment.

An intuitive explanation for the results in Fig. \ref{fig:performancecomps} is that the mitigated outputs have a lower statistical error and so converge more quickly than the postselected outputs for increasing sample number and decreasing loss parameter. Furthermore, the bias errors in the mitigated outputs mean that they do not converge to the ideal outputs, whereas the postselected outputs are unbiased estimators. 
And so for sufficiently high sample number or low loss parameter postselection will eventually outperform recycling mitigation. There are, however, large sampling and loss regimes for which recycling mitigation reliably outperforms postselection.

Theorems \ref{Haarbound} and \ref{nonHaarbound} in section \ref{sec:bounding} provide statistical upper bounds on the deviation of the interference terms from $p_{unif}$. An inspection of these bounds shows that this deviation is with high confidence $\Omega(\sqrt{p_{unif}})$.

We conjecture that this bound could be considerably tightened. In support of this we ran numerical experiments for computations of size $m=16$ modes and $n=4$ photons, computing the average of $|I_{s^n_l,k}-p_{unif}|$  over all bit strings $\{s^n_l\}_l$. This computation was repeated for 20 randomly selected linear optical intereferometers $U$.
This data is plotted in Fig. \ref{fig:mean_abs_int_dev_unif}. As can be seen from the figure it seems that the analytically computed bounds on the bias error are loose. Indeed, it seems that $|I_{s^n_l, k}-p_{unif}| \in      o(p_{unif})$, in line with conjecture \ref{conjbias}.

Finally, note that $|I_{s^n_l, k}-p_{unif}|$ seems to be generally decreasing with increasing $k$. However, ${m -n +k \choose k} |I_{s^n_l, k}-p_{unif}|$ actually increases with increasing $k$, as can be verified by a direct calculation using data from Figure \ref{fig:mean_abs_int_dev_unif}. Ultimately, since ${m -n +k \choose k} |I_{s^n_l, k}-p_{unif}|$ is directly related to the bias error incurred when computing the mitigated probabilities from the recycled probabilities (see appendix \ref{applinearsolv}), increasing values of $k$ in the recycling mitigation will lead to worse performances, as predicted by Lemma \ref{ryser_lemma}.

\begin{figure}[]
     \centering
         \includegraphics[width=0.65\textwidth]{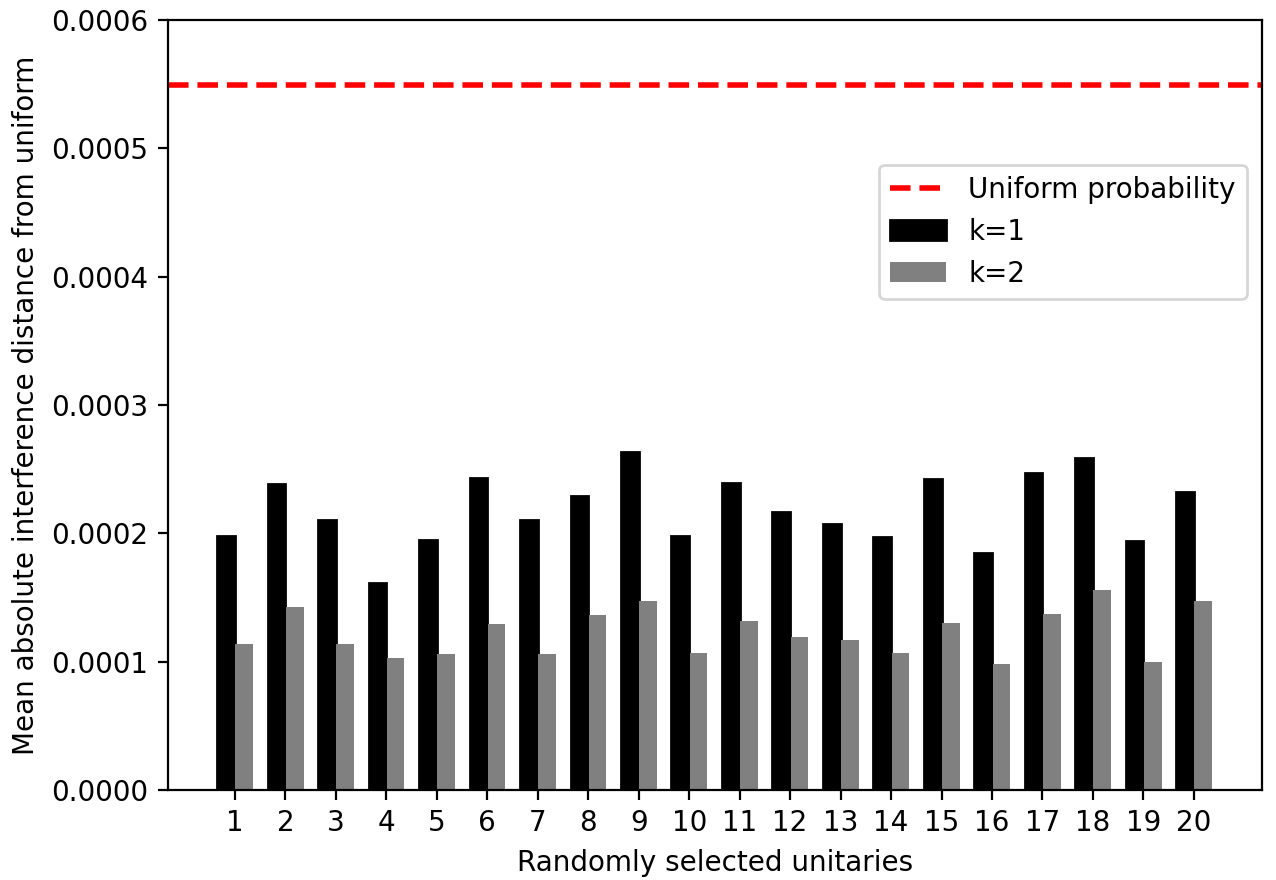}
         
        \caption{The mean absolute value of the  distance of the interference terms from the uniform probability were computed for 20  randomly selected unitaries. For all unitaries, and for $k=1$ and $k=2$, the magnitude of the computed values were observed to be exponentially small in terms of $m$ and $n$ (being of the order of ${m \choose n}^{-1}$). This indicates that it may be possible to derive tighter analytical bounds than those stated in thm. \ref{Haarbound} and thm. \ref{nonHaarbound}.}
    \label{fig:mean_abs_int_dev_unif}
\end{figure}

\section{Properties of the mitigated probabilities in the absence of statistical error}
\label{sec:properties_mitigated}

\begin{figure}[]
\centering
\subfigure[]{
\includegraphics[width=.48\textwidth]{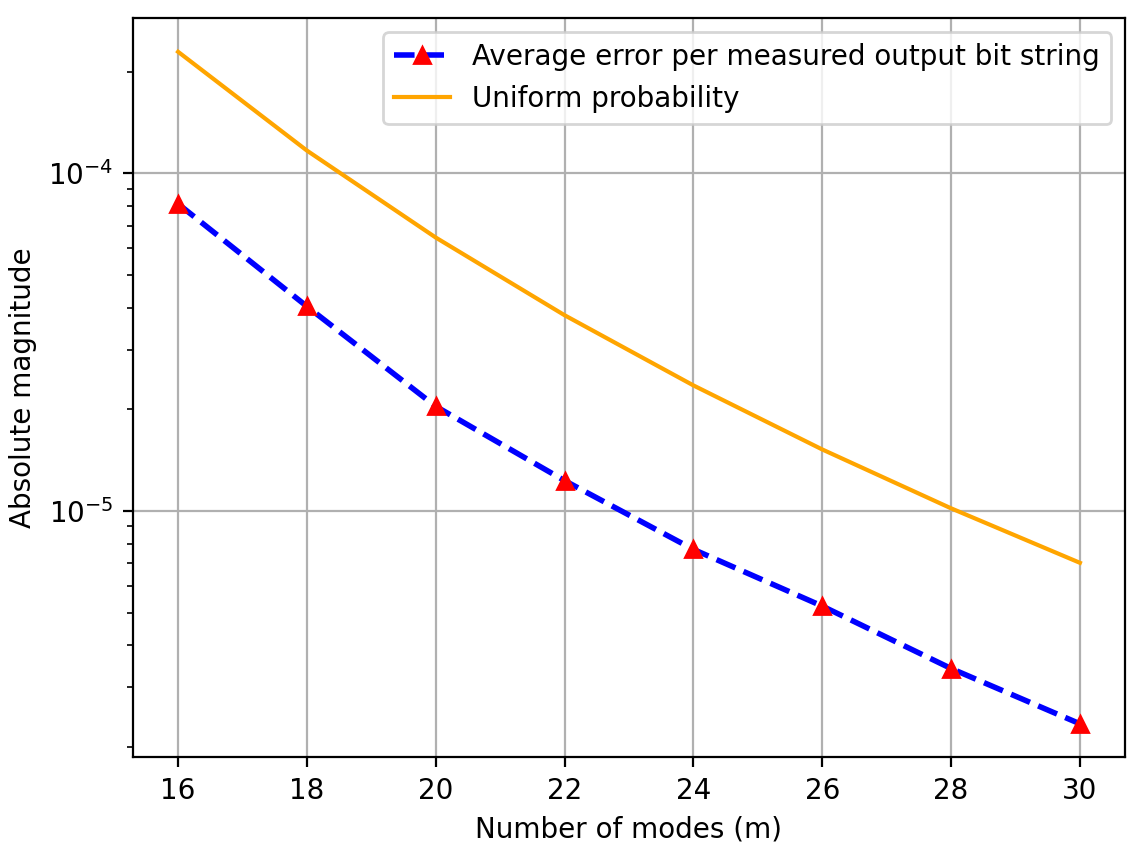}
}
\subfigure[]{
\includegraphics[width=.48\textwidth]{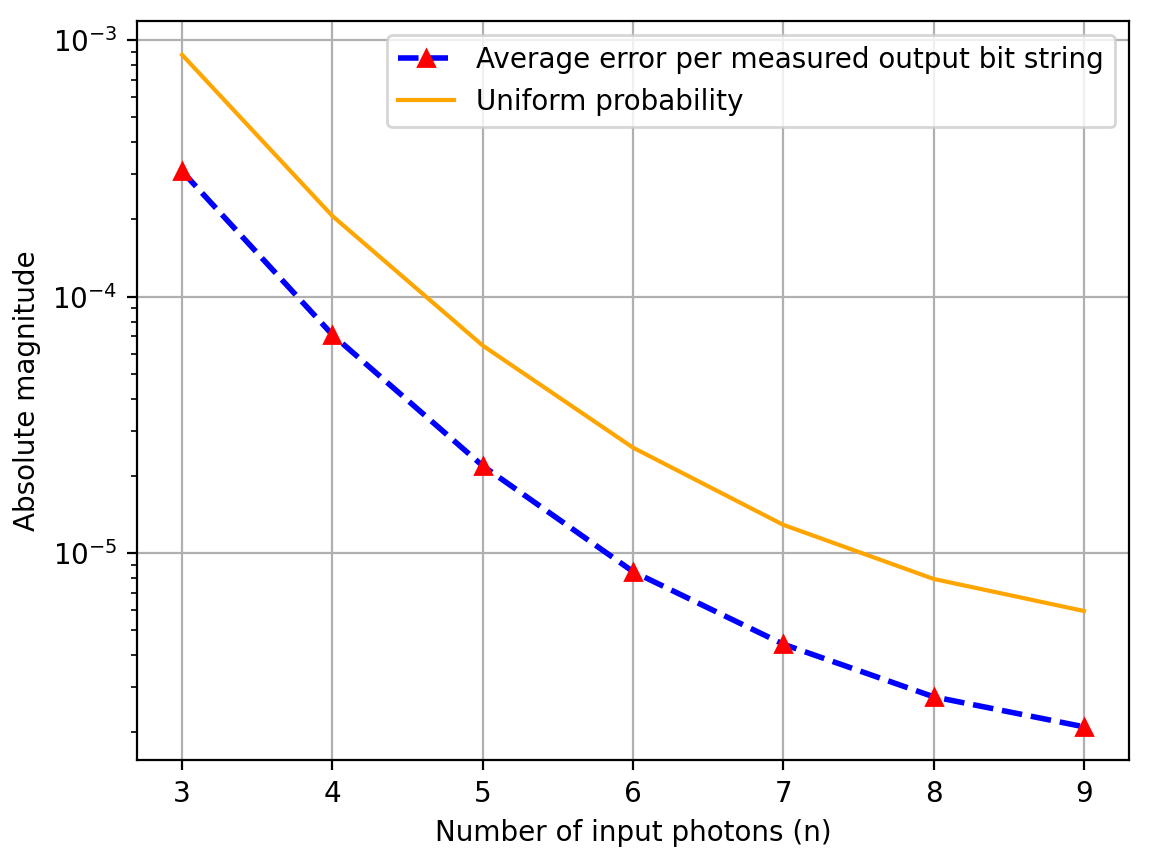}
}
\caption{The numerical simulations that produced these plots involved exact computation of the mitigated probabilities and so only the bias error is present. The magnitude of the average error per output bit string is compared with that of the uniform probability. For all chosen parameters the average error per bit string is below the uniform probability.  (a) The average error per bit string plotted for a range of modes $[16,30]$ for a fixed number of input photons $n=5$. (b) The average error per bit string plotted for a range of input photons $[3,9]$ for a fixed number of modes photons $m=20$.}

\label{fig:avg_error_comparison_with_uniform_prob}
\end{figure}
As seen previously, recycling mitigation is a biased error mitigation. Given a $U$ and a bit string $s$, this means that in the limit of infinite samples the mitigated probability $p(s)$ is related to the ideal probability $p_{id}(s)$ by 
\begin{equation*}
    |p(s)-p_{id}(s)|=\epsilon_{{mitbias},s}.
\end{equation*}
Numerical simulations in Fig. \ref{fig:avg_error_comparison_with_uniform_prob} using the exponential extrapolation  recycling mitigation technique with values of  $k$  independent of $m,n$ show that
$E_{s}(\epsilon_{{mitbias},s}) < p_{unif}$ both in the case of fixed $m$ and varying $n$, as well as fixed $n$ and varying $m$. 
Although the plots in Fig. \ref{fig:avg_error_comparison_with_uniform_prob} are for a fixed $U$, we observe similar curves when testing with different values of $U$. 
Our numerical simulations suggest that Conjecture \ref{conjbias} is true. 
We will now assume that this conjecture is true and derive from it a condition on the mitigated probabilities.

Let  $m=\Omega(n^5)$, then $p_{id}(s) \approx \frac{|\mathsf{Per}(X_s)|^2}{m^n}$ with $X_s$ an i.i.d Gaussian matrix \cite{aaronson_computational_2011},  also $p_{unif} \approx \frac{n!}{m^n}$. 
Thus, Conjecture \ref{conjbias} implies, by using a Markov inequality with $\alpha>1$,

\begin{equation*}
Pr_{s}(|m^np(s)-|\mathsf{Per}(X_s)|^2| \leq \delta(n)\alpha n! ) \geq 1-\frac{1}{\alpha}.
\end{equation*}
Choosing $\alpha$ such that $\alpha \delta (n):=\gamma(n) <1 $, we have that

\begin{equation*}
Pr_{s}(|m^np(s)-|\mathsf{Per}(X_s)|^2| \leq \gamma(n) n! ) \geq 1-\frac{1}{\alpha}.
\end{equation*}
Since this equation holds with high probability $p$ over the choice of $U$, then we can say that
\begin{equation}
\label{eqgpe}
Pr_{X \in \mathcal{G}^{n \times n}}(|m^np(s)-|\mathsf{Per}(X)|^2| \leq \gamma(n) n! ) \geq (1-\frac{1}{\alpha})p,
\end{equation}
where $\mathcal{G}^{n \times n}$ is the ensemble of $n \times n$ Gaussian matrices.

We will now argue that the existence of a $\mathsf{poly}(n)$-time classical algorithm $C$ that can compute any mitigated probability is an unlikely complexity theoretic conequence. 
Indeed, if $C$ existed, then eqn. \ref{eqgpe} directly  implies that $C$ can be used to solve the  $|GPE_{\pm}|^2$ problem with some restrictions on the parameters \cite{aaronson_computational_2011}.  More precisely, computing the mitigated probabilities allows us to estimate $|\mathsf{Per}(X)|^2$ to additive error $  \gamma(n) n ! < n !$ with probability $(1-\frac{1}{\alpha})p \approx 1 - \frac{1}{\alpha}$ over $X \in \mathcal{G}^{n \times n}$. 
The $|GPE_{\pm}|^2$ is about classically estimating $|\mathsf{Per}(X)^2|$ to within additive error $\epsilon n !$ and probability $1-\delta$ over $X \in \mathcal{G}^{n \times n}$ for any $\epsilon, \delta$, and is strongly believed to be $\sharp$P-hard.  
Furthermore, in \cite{oszmaniec_classical_2018} it was argued that $|GPE_{\pm}|^2$ with some restrictions on the parameters similar to those in eqn. \ref{eqgpe} should also be hard to do classically. This means that there likely is no such classical algorithm $C$.

In conclusion, despite recycling mitigation being a biased error mitigation technique, we have presented numerical evidence to argue that the bias errors are small enough to render the mitigated probabilities hard to compute classically.

\section{Evidence that zero noise extrapolation (ZNE) techniques present no advantage over postselection }
\label{sec:ZNE}

In this section, we provide strong evidence that techniques based on ZNE when applied to mitigating photon loss in DVLOQC in general provide no advantage over postselection.

Suppose we are interested in computing a specific marginal probability $p(n_1...n_l|n)$ of observing $n_i$ photons in mode $i \in \{1,\dots,l\}$, with $l \leq m$, and $\sum_{i=1,\dots,l}n_i=c$, with $c \leq n$.  
The notation $|n$ indicates that we are computing the \emph{ideal} marginal probability, when no photon is lost.
Let $p(n_1 \dots n_l \cap j)$ be the probability of observing the output $(n_1,\dots,n_l)$ and detecting $j$ photons in all $m$ modes, with $j \in \{c,\dots,n\}$. When postselecting on no photons being lost, we are computing $$p(n_1 \dots n_l \cap n)=(1-\eta)^np(n_1 \dots n_l |n).$$ 
However, if we compute $p(n_1 \dots n_l)$ without caring  about whether no photon is lost, 
we end up computing 
\begin{equation}
\label{eq1}
    p_{\eta}(n_1...n_l)=\sum_{i=0, \dots, n-c}(1-\eta)^{n-i}\eta^{i}p(n_1 \dots n_l | n-i).
\end{equation}
Extrapolation techniques consist of estimating $p_{\eta}(n_1 \dots n_l)$ for  different values $\{\eta_i\}$ of loss, then deducing from these an estimate of $p(n_1 \dots n_l |n)$. One example of how this can be done is the Richardson extrapolation technique, at the heart of the zero noise extrapolation (ZNE) approach \cite{endo_practical_2018}. Interestingly, it is possible to derive a condition indicating that these techniques offer no advantage over postselection in terms of estimating $p(n_1 \dots n_l |n)$.
Let $\eta=\mathsf{min}_i \eta_i$, and suppose we collect $O(\frac{1}{\epsilon^2_{max}})$ samples from our device, where $0<\epsilon_{max} \leq 1$, then $O\big(\frac{(1-\eta)^n}{\epsilon^2_{max}}\big)$ of these will be samples where no loss has occurred. Using these postselected samples one can compute, with high confidence, an estimate $\tilde{p}(n_1 \dots n_l |n)$ of $p(n_1 \dots n_l |n)$ such that
$|\tilde{p}(n_1 \dots n_l |n)-p(n_1 \dots n_l |n)| \leq \frac{\epsilon_{max}}{\sqrt{(1-\eta)^n}}$, by Hoeffding's inequality.
For  Richardson extrapolation,  in Appendix \ref{appnogo} we use $O\big(n\frac{1}{\epsilon^2_{max}}\big)$ samples, and  compute an estimate $\tilde{p}_{extrap}(n_1 \dots n_l |n)$,  with  $ E_{extrap}:=|\tilde{p}_{extrap}(n_1 \dots n_l |n)-p(n_1 \dots n_l |n)| $ the error incurred. We then explicitly compute  $\mathsf{M}(E_{extrap})$, an upper bound on $E_{extrap}$, in terms of $\epsilon_{max}$ and the $\eta_i$'s, and show the following
\begin{theorem}
\label{thnogozne}
    For all $n \geq n_0$, with $n_0$ a positive integer, $\mathsf{M}(E_{extrap}) \geq \frac{\epsilon_{max}}{\sqrt{(1-\eta)^n}}$.
\end{theorem}
Thm. \ref{thnogozne}, proven in Appendix \ref{appnogo},  is strong evidence that techniques based on Richardson extrapolation offer no advantage over postselection. Our main technical contribution in proving Thm. \ref{thnogozne} is to link determining the error of ZNE methods to computing the norm of the inverse of a Vandermonde matrix whose entries are determined by $\eta$, we then use existing results to upper bound this norm \cite{gautschi_inverses_1978}.  

To strengthen our theoretical  analysis, we also numerically compute the  number of violations of the inequality $E_{extrap} \geq \frac{\epsilon_{max}}{\sqrt{(1-\eta_0)^n}}$ and plot these in Figure \ref{figuniferr}. 
We took $\epsilon_{max}=0.01$, $\eta_0=0.01$, $\eta_{n-c}=0.95$, and $\eta_i$ for $1<i<n-c$ equally spaced. We varied the value of $n-c$
between 3 and 14, with $c=\mathsf{ceil}(n/3)$
and $\mathsf{ceil}(.)$ is the ceiling function. For each value of $n-c$ we performed 3000 runs, where at each run we took  $n-c+1$ values of $\{\epsilon_i\}$
chosen uniformly randomly from $[-\epsilon_{max}, \epsilon_{max}]$. As can be observed in Figure \ref{figuniferr}, the number of violations approaches zero with increasing $n$, confirming that extrapolation performs worse than post-selection after some value of $n$. A similar behaviour is observed for different values of $\epsilon_{max}, \eta_0$ and $\eta_{n-c}$.

\begin{figure}[]
\includegraphics[scale=0.37]{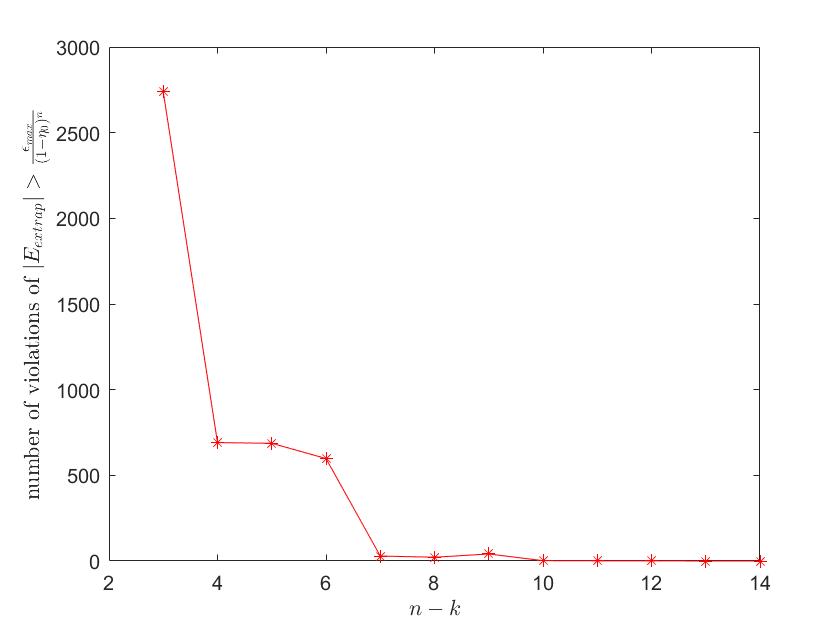}
\centering
\caption{Number of violations of Thm. \ref{thnogozne} inequality plotted versus $n-c$ .}
\label{figuniferr}
\end{figure}
\bigskip

\section{Discussion}
\label{sec:discussion}

In summary, we have presented a family of techniques, collectively referred to as \textit{recycling mitigation}, for mitigating the effects of photon loss on the outputs of linear optical quantum circuits in the discrete variable setting. 
We provided analytical and numerical evidence that these techniques outperform postselection - currently the standard method for mitigating loss in linear optical circuits. 

\newpage

There are many possible directions of future work. 
For instance, carrying out of a rigorous analysis of the set-up where recycling mitigation is applied to linear optical computation with adaptive measurements and feedforward, such as in \cite{chabaud2021quantum}. 
This would be interesting for a number of reasons, not least of which is that in this computational set-up, with a sufficient number of adaptive measurement steps, a quantum advantage for probability estimation is feasible \cite{chabaud2021quantum}.  As an estimate, if $n_{ad}$ is the number of photons consumed in the $r$ feed-forward steps in the absence of loss, and $\eta$ is the loss per mode, the overall probability that the $r$ feed-forward steps are executed correctly is $(1-\eta)^{n_{ad}}$. 
This introduces an unavoidable additional exponential sampling overhead to recycling mitigation. 
Nevertheless, it would be interesting benchmark the performance of recycling mitigation against only postselection using a set-up of this kind, especially in the possible regime where there is not sufficient feed-forward for universality, but where classical simulation methods are also inefficient \cite{chabaud2021quantum}.

\bigskip 
It would also be interesting to investigate how the method might be combined with early fault-tolerance schemes \cite{katabarwa2024early} for optical quantum computing.
Specifically, in investigating the resource trade-off between error mitigation and correction needed to achieve a certain quality of computational output.

\bigskip 
Another interesting question, motivated by finding ways to eliminate the bias error in our introduced mitigation techniques is whether one can \emph{exactly} represent an $n$--photon probability as a sum of $n-k$--photon probabilities. This can be done for the case of minors of unitary matrices with real and positive entries, such as those used to solve graph problems in DVLOQC \cite{mezher_solving_2023}. As an example, consider a Laplace expansion of the permanent of an $n \times n$ minor $U_{\mathbf{s}, \mathbf{t}}:=(u_{ij})_{i,j \in \{1, \dots, n\}}$ of a linear optical unitary $U$, this reads 
\begin{equation*}
    \mathsf{Per}(U_{\mathsf{s}, \mathsf{t}})=u_{11}\mathsf{Per}\begin{pmatrix}
        u_{22}& \dots & u_{2n} \\
        .&.&. \\
        .&.&.\\
        .&.&.\\
        u_{n2}& \dots & u_{nn}
    \end{pmatrix} + \dots +u_{1n}\mathsf{Per}\begin{pmatrix}
        u_{21}& \dots & u_{2,n-1} \\
        .&.&. \\
        .&.&.\\
        .&.&.\\
        u_{n1}& \dots & u_{n,n-1}
    \end{pmatrix}.
\end{equation*}
If $U_{\mathsf{s}, \mathsf{t}}$ is composed of positive real entries, 
then $\mathsf{Per}(U_{\mathsf{s}, \mathsf{t}})$ is proportional to the square root of a probability $\sqrt{p(s|t)}$ of a certain output of linear optical circuit $U$ with $n$ input photons \cite{aaronson_computational_2011}. In turn, each of the permanents of $n-1 \times n-1$ matrices on the right hand side of the above equation is proportional to the square root of a probability of a certain output of a linear optical circuit with $n-1$ input photons. 
Thus, in the case of minors with positive real entries, there is an exact way to represent an $n$--photon probability 
as a sum of $n-1$ photon  probabilities. In this case as well, the process can be generalised, by successive Laplace expansions, to expressing $n$--photon probabilities exactly as a sum of $n-k$--photon probabilities.

This idea of exactly representing $n$--photon probabilities as sums of $n-k$--photon probabilities is interesting since, in the presence of photon loss, $n-k$--photon experiments in general produce more converged statistics than $n$--photon experiments, for a comparable number of runs of the experiments in both cases. Furthermore, since the $n$--photon probabilities are exactly expressible as a sum of $n-k$--photon probabilities, one can use the $n-k$--photon experiments to compute the $n$--photon probabilities and \emph{indefinitely} (because of lower statistical error) outperform postselection. It is an interesting question to determine whether similar results hold for broader classes of matrices.

\bigskip 
Finally, it could be interesting to extend our techniques to the case where we allow collisions in the output. The construction of the recycled probabilities presented generalises straightforwardly to the collision regime, the main technical challenge is therefore to adapt our proofs to this regime.

\subsection*{Code availability}

Open-source code to perform recycling mitigation may be found at the following Github location:

\url{https://github.com/Quandela/Perceval/tree/main/perceval/error_mitigation}

\vspace{1em}

\subsection*{Acknowledgments}

The authors would like to thank Hugo Thomas for suggesting the approach taken to proving Lemma \ref{normalise_miti}. Furthermore, the authors would like to thank Enguerrand Monard, Pierre-Emmanuel Emeriau, Shane Mansfield, Stephen Wein, Alexia Salavrakos, Emilio Annoni, Robert Booth, and Elham Kashefi for helpful discussions. The authors would also like to thank Raúl García-Patrón for providing helpful feedback on an earlier version of this manuscript. This work has been co-funded by the European Commission as part of the EIC accelerator program under the grant agreement 190188855 for SEPOQC project, by the Horizon-CL4 program under the grant agreement 101135288 for EPIQUE project, and by the OECQ Project financed by the French State as part of France 2030.

\vspace{1em}

\bibliography{main.bib}

\begin{thebibliography}{10}

\bibitem{knill_scheme_2001}
E.~Knill, R.~Laflamme, and G.~J. Milburn.
\newblock ``A scheme for efficient quantum computation with linear optics''.
\newblock \href{https://dx.doi.org/10.1038/35051009}{Nature {\bf 409}, 46--52}~(2001).

\bibitem{raussendorf_topological_2007}
R.~Raussendorf, J.~Harrington, and K.~Goyal.
\newblock ``Topological fault-tolerance in cluster state quantum computation''.
\newblock \href{https://dx.doi.org/10.1088/1367-2630/9/6/199}{New Journal of Physics {\bf 9}, 199}~(2007).

\bibitem{de_gliniasty_spin-optical_2023}
Gr{\'{e}}goire de~Gliniasty, Paul Hilaire, Pierre-Emmanuel Emeriau, Stephen~C. Wein, Alexia Salavrakos, and Shane Mansfield.
\newblock ``A {S}pin-{O}ptical {Q}uantum {C}omputing {A}rchitecture''.
\newblock \href{https://dx.doi.org/10.22331/q-2024-07-24-1423}{{Quantum} {\bf 8}, 1423}~(2024).

\bibitem{bartolucci_fusion-based_2023}
Sara Bartolucci, Patrick Birchall, Hector Bombín, Hugo Cable, Chris Dawson, Mercedes Gimeno-Segovia, Eric Johnston, Konrad Kieling, Naomi Nickerson, Mihir Pant, Fernando Pastawski, Terry Rudolph, and Chris Sparrow.
\newblock ``Fusion-based quantum computation''.
\newblock \href{https://dx.doi.org/10.1038/s41467-023-36493-1}{Nature Communications {\bf 14}, 912}~(2023).

\bibitem{raussendorf_one-way_2001}
Robert Raussendorf and Hans~J. Briegel.
\newblock ``A {One}-{Way} {Quantum} {Computer}''.
\newblock \href{https://dx.doi.org/10.1103/PhysRevLett.86.5188}{Physical Review Letters {\bf 86}, 5188--5191}~(2001).

\bibitem{aaronson_computational_2011}
Scott Aaronson and Alex Arkhipov.
\newblock ``The computational complexity of linear optics''.
\newblock In Proceedings of the forty-third annual {ACM} symposium on {Theory} of computing.
\newblock \href{https://dx.doi.org/10.1145/1993636.1993682}{Pages 333--342}.
\newblock {STOC} '11New York, NY, USA~(2011). Association for Computing Machinery.

\bibitem{senellart_high-performance_2017}
Pascale Senellart, Glenn Solomon, and Andrew White.
\newblock ``High-performance semiconductor quantum-dot single-photon sources''.
\newblock \href{https://dx.doi.org/10.1038/nnano.2017.218}{Nature Nanotechnology {\bf 12}, 1026--1039}~(2017).

\bibitem{reck_experimental_1994}
Michael Reck, Anton Zeilinger, Herbert~J. Bernstein, and Philip Bertani.
\newblock ``Experimental realization of any discrete unitary operator''.
\newblock \href{https://dx.doi.org/10.1103/PhysRevLett.73.58}{Physical Review Letters {\bf 73}, 58--61}~(1994).

\bibitem{hadfield_single-photon_2009}
Robert~H. Hadfield.
\newblock ``Single-photon detectors for optical quantum information applications''.
\newblock \href{https://dx.doi.org/10.1038/nphoton.2009.230}{Nature Photonics {\bf 3}, 696--705}~(2009).

\bibitem{wang_toward_2018}
Hui Wang, Wei Li, Xiao Jiang, Y.-M. He, Y.-H. Li, X.~Ding, M.-C. Chen, J.~Qin, C.-Z. Peng, C.~Schneider, M.~Kamp, W.-J. Zhang, H.~Li, L.-X. You, Z.~Wang, J.P. Dowling, S.~Höfling, Chao-Yang Lu, and Jian-Wei Pan.
\newblock ``Toward {Scalable} {Boson} {Sampling} with {Photon} {Loss}''.
\newblock \href{https://dx.doi.org/10.1103/PhysRevLett.120.230502}{Physical Review Letters {\bf 120}, 230502}~(2018).

\bibitem{brod_photonic_2019}
Daniel~J. Brod, Ernesto~F. Galvão, Andrea Crespi, Roberto Osellame, Nicolò Spagnolo, and Fabio Sciarrino.
\newblock ``Photonic implementation of boson sampling: a review''.
\newblock \href{https://dx.doi.org/10.1117/1.AP.1.3.034001}{Advanced Photonics {\bf 1}, 034001}~(2019).

\bibitem{renema_classical_2019}
Jelmer Renema, Valery Shchesnovich, and Raul Garcia-Patron.
\newblock ``Classical simulability of noisy boson sampling''~(2019).
\newblock  \href{http://arxiv.org/abs/1809.01953}{arXiv:1809.01953}.

\bibitem{temme_error_2017}
Kristan Temme, Sergey Bravyi, and Jay~M. Gambetta.
\newblock ``Error {Mitigation} for {Short}-{Depth} {Quantum} {Circuits}''.
\newblock \href{https://dx.doi.org/10.1103/PhysRevLett.119.180509}{Physical Review Letters {\bf 119}, 180509}~(2017).

\bibitem{endo_practical_2018}
Suguru Endo, Simon~C. Benjamin, and Ying Li.
\newblock ``Practical {Quantum} {Error} {Mitigation} for {Near}-{Future} {Applications}''.
\newblock \href{https://dx.doi.org/10.1103/PhysRevX.8.031027}{Physical Review X {\bf 8}, 031027}~(2018).

\bibitem{giurgica-tiron_digital_2020}
Tudor Giurgica-Tiron, Yousef Hindy, Ryan LaRose, Andrea Mari, and William~J. Zeng.
\newblock ``Digital zero noise extrapolation for quantum error mitigation''.
\newblock In 2020 {IEEE} {International} {Conference} on {Quantum} {Computing} and {Engineering} ({QCE}).
\newblock \href{https://dx.doi.org/10.1109/QCE49297.2020.00045}{Pages 306--316}.
\newblock ~(2020).

\bibitem{he_zero-noise_2020}
Andre He, Benjamin Nachman, Wibe~A. de~Jong, and Christian~W. Bauer.
\newblock ``Zero-noise extrapolation for quantum-gate error mitigation with identity insertions''.
\newblock \href{https://dx.doi.org/10.1103/PhysRevA.102.012426}{Physical Review A {\bf 102}, 012426}~(2020).

\bibitem{strikis_learning-based_2021}
Armands Strikis, Dayue Qin, Yanzhu Chen, Simon~C. Benjamin, and Ying Li.
\newblock ``Learning-{Based} {Quantum} {Error} {Mitigation}''.
\newblock \href{https://dx.doi.org/10.1103/PRXQuantum.2.040330}{PRX Quantum {\bf 2}, 040330}~(2021).

\bibitem{mari_extending_2021}
Andrea Mari, Nathan Shammah, and William~J. Zeng.
\newblock ``Extending quantum probabilistic error cancellation by noise scaling''.
\newblock \href{https://dx.doi.org/10.1103/PhysRevA.104.052607}{Physical Review A {\bf 104}, 052607}~(2021).

\bibitem{van_den_berg_probabilistic_2023}
Ewout van~den Berg, Zlatko~K. Minev, Abhinav Kandala, and Kristan Temme.
\newblock ``Probabilistic error cancellation with sparse {Pauli}–{Lindblad} models on noisy quantum processors''.
\newblock \href{https://dx.doi.org/10.1038/s41567-023-02042-2}{Nature Physics {\bf 19}, 1116--1121}~(2023).

\bibitem{bonet-monroig_low-cost_2018}
X.~Bonet-Monroig, R.~Sagastizabal, M.~Singh, and T.~E. O'Brien.
\newblock ``Low-cost error mitigation by symmetry verification''.
\newblock \href{https://dx.doi.org/10.1103/PhysRevA.98.062339}{Physical Review A {\bf 98}, 062339}~(2018).

\bibitem{sagastizabal_experimental_2019-1}
R.~Sagastizabal, X.~Bonet-Monroig, M.~Singh, M.~A. Rol, C.~C. Bultink, X.~Fu, C.~H. Price, V.~P. Ostroukh, N.~Muthusubramanian, A.~Bruno, M.~Beekman, N.~Haider, T.~E. O'Brien, and L.~DiCarlo.
\newblock ``Experimental error mitigation via symmetry verification in a variational quantum eigensolver''.
\newblock \href{https://dx.doi.org/10.1103/PhysRevA.100.010302}{Physical Review A {\bf 100}, 010302}~(2019).

\bibitem{obrien_error_2021}
Thomas~E. O’Brien, Stefano Polla, Nicholas~C. Rubin, William~J. Huggins, Sam McArdle, Sergio Boixo, Jarrod~R. McClean, and Ryan Babbush.
\newblock ``Error {Mitigation} via {Verified} {Phase} {Estimation}''.
\newblock \href{https://dx.doi.org/10.1103/PRXQuantum.2.020317}{PRX Quantum {\bf 2}, 020317}~(2021).

\bibitem{mezher_mitigating_2022}
Rawad Mezher, James Mills, and Elham Kashefi.
\newblock ``Mitigating errors by quantum verification and postselection''.
\newblock \href{https://dx.doi.org/10.1103/PhysRevA.105.052608}{Physical Review A {\bf 105}, 052608}~(2022).

\bibitem{koczor_exponential_2021}
Bálint Koczor.
\newblock ``Exponential {Error} {Suppression} for {Near}-{Term} {Quantum} {Devices}''.
\newblock \href{https://dx.doi.org/10.1103/PhysRevX.11.031057}{Physical Review X {\bf 11}, 031057}~(2021).

\bibitem{koczor_dominant_2021}
Bálint Koczor.
\newblock ``The dominant eigenvector of a noisy quantum state''.
\newblock \href{https://dx.doi.org/10.1088/1367-2630/ac37ae}{New Journal of Physics {\bf 23}, 123047}~(2021).

\bibitem{huggins_virtual_2021}
William~J. Huggins, Sam McArdle, Thomas~E. O’Brien, Joonho Lee, Nicholas~C. Rubin, Sergio Boixo, K.~Birgitta Whaley, Ryan Babbush, and Jarrod~R. McClean.
\newblock ``Virtual {Distillation} for {Quantum} {Error} {Mitigation}''.
\newblock \href{https://dx.doi.org/10.1103/PhysRevX.11.041036}{Physical Review X {\bf 11}, 041036}~(2021).

\bibitem{yamamoto_error-mitigated_2022}
Kaoru Yamamoto, Suguru Endo, Hideaki Hakoshima, Yuichiro Matsuzaki, and Yuuki Tokunaga.
\newblock ``Error-{Mitigated} {Quantum} {Metrology} via {Virtual} {Purification}''.
\newblock \href{https://dx.doi.org/10.1103/PhysRevLett.129.250503}{Physical Review Letters {\bf 129}, 250503}~(2022).

\bibitem{mcclean_decoding_2020}
Jarrod~R. McClean, Zhang Jiang, Nicholas~C. Rubin, Ryan Babbush, and Hartmut Neven.
\newblock ``Decoding quantum errors with subspace expansions''.
\newblock \href{https://dx.doi.org/10.1038/s41467-020-14341-w}{Nature Communications {\bf 11}, 636}~(2020).

\bibitem{yoshioka_generalized_2022}
Nobuyuki Yoshioka, Hideaki Hakoshima, Yuichiro Matsuzaki, Yuuki Tokunaga, Yasunari Suzuki, and Suguru Endo.
\newblock ``Generalized {Quantum} {Subspace} {Expansion}''.
\newblock \href{https://dx.doi.org/10.1103/PhysRevLett.129.020502}{Physical Review Letters {\bf 129}, 020502}~(2022).

\bibitem{yang_dual-gse_2023}
Bo~Yang, Nobuyuki Yoshioka, Hiroyuki Harada, Shigeo Hakkaku, Yuuki Tokunaga, Hideaki Hakoshima, Kaoru Yamamoto, and Suguru Endo.
\newblock ``Resource-efficient generalized quantum subspace expansion''.
\newblock \href{https://dx.doi.org/10.1103/PhysRevApplied.23.054021}{Phys. Rev. Appl. {\bf 23}, 054021}~(2025).

\bibitem{ohkura_leveraging_2023}
Yasuhiro Ohkura, Suguru Endo, Takahiko Satoh, Rodney Van~Meter, and Nobuyuki Yoshioka.
\newblock ``Leveraging hardware-control imperfections for error mitigation via generalized quantum subspace''~(2023).
\newblock  \href{http://arxiv.org/abs/2303.07660}{arXiv:2303.07660}.

\bibitem{maciejewski_mitigation_2020}
Filip~B. Maciejewski, Zoltán Zimborás, and Michał Oszmaniec.
\newblock ``Mitigation of readout noise in near-term quantum devices by classical post-processing based on detector tomography''.
\newblock \href{https://dx.doi.org/10.22331/q-2020-04-24-257}{Quantum {\bf 4}, 257}~(2020).

\bibitem{arrasmith_development_2023}
Andrew Arrasmith, Andrew Patterson, Alice Boughton, and Marco Paini.
\newblock ``Development and demonstration of an efficient readout error mitigation technique for use in nisq algorithms''~(2023).
\newblock  \href{http://arxiv.org/abs/2303.17741}{arXiv:2303.17741}.

\bibitem{mills_simplifying_2023}
James Mills, Debasis Sadhukhan, and Elham Kashefi.
\newblock ``Simplifying errors by symmetry and randomisation''~(2023).
\newblock  \href{http://arxiv.org/abs/2303.02712}{arXiv:2303.02712}.

\bibitem{bravyi_mitigating_2021}
Sergey Bravyi, Sarah Sheldon, Abhinav Kandala, David~C. Mckay, and Jay~M. Gambetta.
\newblock ``Mitigating measurement errors in multiqubit experiments''.
\newblock \href{https://dx.doi.org/10.1103/PhysRevA.103.042605}{Physical Review A {\bf 103}, 042605}~(2021).

\bibitem{endo_hybrid_2021}
Suguru Endo, Zhenyu Cai, Simon~C. Benjamin, and Xiao Yuan.
\newblock ``Hybrid {Quantum}-{Classical} {Algorithms} and {Quantum} {Error} {Mitigation}''.
\newblock \href{https://dx.doi.org/10.7566/JPSJ.90.032001}{Journal of the Physical Society of Japan {\bf 90}, 032001}~(2021).

\bibitem{li_efficient_2017}
Ying Li and Simon~C. Benjamin.
\newblock ``Efficient {Variational} {Quantum} {Simulator} {Incorporating} {Active} {Error} {Minimization}''.
\newblock \href{https://dx.doi.org/10.1103/PhysRevX.7.021050}{Physical Review X {\bf 7}, 021050}~(2017).

\bibitem{su_error_2021}
Daiqin Su, Robert Israel, Kunal Sharma, Haoyu Qi, Ish Dhand, and Kamil Brádler.
\newblock ``Error mitigation on a near-term quantum photonic device''.
\newblock \href{https://dx.doi.org/10.22331/q-2021-05-04-452}{Quantum {\bf 5}, 452}~(2021).

\bibitem{gautschi_inverses_1978}
Walter Gautschi.
\newblock ``On inverses of {Vandermonde} and confluent {Vandermonde} matrices {III}''.
\newblock \href{https://dx.doi.org/10.1007/BF01432880}{Numerische Mathematik {\bf 29}, 445--450}~(1978).

\bibitem{de_palma_limitations_2023}
Giacomo De~Palma, Milad Marvian, Cambyse Rouzé, and Daniel~Stilck França.
\newblock ``Limitations of {Variational} {Quantum} {Algorithms}: {A} {Quantum} {Optimal} {Transport} {Approach}''.
\newblock \href{https://dx.doi.org/10.1103/PRXQuantum.4.010309}{PRX Quantum {\bf 4}, 010309}~(2023).

\bibitem{takagi_universal_2023}
Ryuji Takagi, Hiroyasu Tajima, and Mile Gu.
\newblock ``Universal {Sampling} {Lower} {Bounds} for {Quantum} {Error} {Mitigation}''.
\newblock \href{https://dx.doi.org/10.1103/PhysRevLett.131.210602}{Physical Review Letters {\bf 131}, 210602}~(2023).

\bibitem{quek_exponentially_2023}
Yihui Quek, Daniel Stilck~Fran{\c{c}}a, Sumeet Khatri, Johannes~Jakob Meyer, and Jens Eisert.
\newblock ``Exponentially tighter bounds on limitations of quantum error mitigation''.
\newblock \href{https://dx.doi.org/https://doi.org/10.1038/s41567-024-02536-7}{Nature Physics {\bf 20}, 1648--1658}~(2024).

\bibitem{tsubouchi_universal_2023}
Kento Tsubouchi, Takahiro Sagawa, and Nobuyuki Yoshioka.
\newblock ``Universal {Cost} {Bound} of {Quantum} {Error} {Mitigation} {Based} on {Quantum} {Estimation} {Theory}''.
\newblock \href{https://dx.doi.org/10.1103/PhysRevLett.131.210601}{Physical Review Letters {\bf 131}, 210601}~(2023).

\bibitem{zimboras2025myths}
Zolt{\'a}n Zimbor{\'a}s, B{\'a}lint Koczor, Zo{\"e} Holmes, Elsi-Mari Borrelli, Andr{\'a}s Gily{\'e}n, Hsin-Yuan Huang, Zhenyu Cai, Antonio Ac{\'\i}n, Leandro Aolita, Leonardo Banchi, et~al.
\newblock ``Myths around quantum computation before full fault tolerance: What no-go theorems rule out and what they don't''~(2025).
\newblock  \href{http://arxiv.org/abs/2501.05694}{arXiv:2501.05694}.

\bibitem{kim2023evidence}
Youngseok Kim, Andrew Eddins, Sajant Anand, Ken~Xuan Wei, Ewout Van Den~Berg, Sami Rosenblatt, Hasan Nayfeh, Yantao Wu, Michael Zaletel, Kristan Temme, et~al.
\newblock ``Evidence for the utility of quantum computing before fault tolerance''.
\newblock \href{https://dx.doi.org/10.1038/s41586-023-06096-3}{Nature {\bf 618}, 500--505}~(2023).

\bibitem{tindall2024efficient}
Joseph Tindall, Matthew Fishman, E~Miles Stoudenmire, and Dries Sels.
\newblock ``Efficient tensor network simulation of ibm’s eagle kicked ising experiment''.
\newblock \href{https://dx.doi.org/10.1103/PRXQuantum.5.010308}{PRX Quantum {\bf 5}, 010308}~(2024).

\bibitem{salavrakos_error-mitigated_2024}
Alexia Salavrakos, Tigran Sedrakyan, James Mills, Shane Mansfield, and Rawad Mezher.
\newblock ``Error-mitigated photonic quantum circuit born machine''.
\newblock \href{https://dx.doi.org/https://doi.org/10.1103/PhysRevA.111.L030401}{Physical Review A {\bf 111}, L030401}~(2025).

\bibitem{benedetti_parameterized_2019}
Marcello Benedetti, Erika Lloyd, Stefan Sack, and Mattia Fiorentini.
\newblock ``Parameterized quantum circuits as machine learning models''.
\newblock \href{https://dx.doi.org/10.1088/2058-9565/ab4eb5}{Quantum Science and Technology {\bf 4}, 043001}~(2019).

\bibitem{maring_general-purpose_2023}
Nicolas Maring, Andreas Fyrillas, Mathias Pont, Edouard Ivanov, Petr Stepanov, Nico Margaria, William Hease, Anton Pishchagin, Aristide Lema{\^\i}tre, Isabelle Sagnes, et~al.
\newblock ``A versatile single-photon-based quantum computing platform''.
\newblock \href{https://dx.doi.org/https://doi.org/10.1038/s41566-024-01403-4}{Nature Photonics {\bf 18}, 603--609}~(2024).

\bibitem{lee_error-mitigated_2022}
Donghwa Lee, Jinil Lee, Seongjin Hong, Hyang-Tag Lim, Young-Wook Cho, Sang-Wook Han, Hyundong Shin, Junaid~ur Rehman, and Yong-Su Kim.
\newblock ``Error-mitigated photonic variational quantum eigensolver using a single-photon ququart''.
\newblock \href{https://dx.doi.org/10.1364/OPTICA.441163}{Optica {\bf 9}, 88--95}~(2022).

\bibitem{heurtel_perceval_2023}
Nicolas Heurtel, Andreas Fyrillas, Grégoire~de Gliniasty, Raphaël~Le Bihan, Sébastien Malherbe, Marceau Pailhas, Eric Bertasi, Boris Bourdoncle, Pierre-Emmanuel Emeriau, Rawad Mezher, Luka Music, Nadia Belabas, Benoît Valiron, Pascale Senellart, Shane Mansfield, and Jean Senellart.
\newblock ``Perceval: {A} {Software} {Platform} for {Discrete} {Variable} {Photonic} {Quantum} {Computing}''.
\newblock \href{https://dx.doi.org/10.22331/q-2023-02-21-931}{Quantum {\bf 7}, 931}~(2023).

\bibitem{gan_fock_2022}
Beng~Yee Gan, Daniel Leykam, and Dimitris~G. Angelakis.
\newblock ``Fock {State}-enhanced {Expressivity} of {Quantum} {Machine} {Learning} {Models}''.
\newblock \href{https://dx.doi.org/10.1140/epjqt/s40507-022-00135-0}{EPJ Quantum Technology {\bf 9}, 16}~(2022).

\bibitem{mezher_solving_2023}
Rawad Mezher, Ana~Filipa Carvalho, and Shane Mansfield.
\newblock ``Solving graph problems with single photons and linear optics''.
\newblock \href{https://dx.doi.org/10.1103/PhysRevA.108.032405}{Physical Review A {\bf 108}, 032405}~(2023).

\bibitem{taylor_quantum_2024}
Adam Taylor, Gabriele Bressanini, Hyukjoon Kwon, and MS~Kim.
\newblock ``Quantum error cancellation in photonic systems: Undoing photon losses''.
\newblock \href{https://dx.doi.org/https://doi.org/10.1103/PhysRevA.110.022622}{Physical Review A {\bf 110}, 022622}~(2024).

\bibitem{nezami_permanent_2021}
Sepehr Nezami.
\newblock ``Permanent of random matrices from representation theory: moments, numerics, concentration, and comments on hardness of boson-sampling''~(2021).
\newblock  \href{http://arxiv.org/abs/2104.06423}{arXiv:2104.06423}.

\bibitem{hoeffding_collected_1994}
Wassily Hoeffding.
\newblock ``The collected works of wassily hoeffding''.
\newblock \href{https://dx.doi.org/10.1007/978-1-4612-0865-5}{Springer}. ~(1994).

\bibitem{singh_proof--work_2023}
Deepesh Singh, Gopikrishnan Muraleedharan, Boxiang Fu, Chen-Mou Cheng, Nicolas~Roussy Newton, Peter Rohde, and Gavin~Keith Brennen.
\newblock ``Proof-of-work consensus by quantum sampling''.
\newblock \href{https://dx.doi.org/10.1088/2058-9565/adae2b}{Quantum Science and Technology}~(2023).

\bibitem{oszmaniec_classical_2018}
Michał Oszmaniec and Daniel~J. Brod.
\newblock ``Classical simulation of photonic linear optics with lost particles''.
\newblock \href{https://dx.doi.org/10.1088/1367-2630/aadfa8}{New Journal of Physics {\bf 20}, 092002}~(2018).

\bibitem{stace_error_2010}
Thomas~M. Stace and Sean~D. Barrett.
\newblock ``Error correction and degeneracy in surface codes suffering loss''.
\newblock \href{https://dx.doi.org/10.1103/PhysRevA.81.022317}{Physical Review A {\bf 81}, 022317}~(2010).

\bibitem{garcia-patron_simulating_2019}
Raúl García-Patrón, Jelmer~J. Renema, and Valery Shchesnovich.
\newblock ``Simulating boson sampling in lossy architectures''.
\newblock \href{https://dx.doi.org/10.22331/q-2019-08-05-169}{Quantum {\bf 3}, 169}~(2019).

\bibitem{aaronson_bosonsampling_2016}
Scott Aaronson and Daniel~J. Brod.
\newblock ``{BosonSampling} with lost photons''.
\newblock \href{https://dx.doi.org/10.1103/PhysRevA.93.012335}{Physical Review A {\bf 93}, 012335}~(2016).

\bibitem{berkowitz_stability_2018}
Ross Berkowitz and Pat Devlin.
\newblock ``A stability result using the matrix norm to bound the permanent''.
\newblock \href{https://dx.doi.org/10.1007/s11856-018-1655-7}{Israel Journal of Mathematics {\bf 224}, 437--454}~(2018).

\bibitem{lin_probability_2011}
Zhengyan Lin and Zhidong Bai.
\newblock ``Probability {Inequalities}''.
\newblock \href{https://dx.doi.org/10.1007/978-3-642-05261-3}{Springer}. Berlin, Heidelberg~(2011).

\bibitem{bhatia_better_2000}
Rajendra Bhatia and Chandler Davis.
\newblock ``A {Better} {Bound} on the {Variance}''.
\newblock \href{https://dx.doi.org/10.1080/00029890.2000.12005203}{The American Mathematical Monthly {\bf 107}, 353--357}~(2000).

\bibitem{levenberg_method_1944}
Kenneth Levenberg.
\newblock ``A method for the solution of certain non-linear problems in least squares''.
\newblock \href{https://dx.doi.org/10.1090/qam/10666}{Quarterly of Applied Mathematics {\bf 2}, 164--168}~(1944).

\bibitem{marquardt_algorithm_1963}
Donald~W. Marquardt.
\newblock ``An {Algorithm} for {Least}-{Squares} {Estimation} of {Nonlinear} {Parameters}''.
\newblock \href{https://dx.doi.org/10.1137/0111030}{Journal of the Society for Industrial and Applied Mathematics {\bf 11}, 431--441}~(1963).

\bibitem{aaronson_bosonsampling_2014}
Scott Aaronson and Alex Arkhipov.
\newblock ``Bosonsampling is far from uniform''.
\newblock \href{https://dx.doi.org/10.26421/QIC14.15-16-7}{Quantum Information \& Computation {\bf 14}, 1383--1423}~(2014).

\bibitem{chabaud2021quantum}
Ulysse Chabaud, Damian Markham, and Adel Sohbi.
\newblock ``Quantum machine learning with adaptive linear optics''.
\newblock \href{https://dx.doi.org/10.22331/q-2021-07-05-496}{Quantum {\bf 5}, 496}~(2021).

\bibitem{katabarwa2024early}
Amara Katabarwa, Katerina Gratsea, Athena Caesura, and Peter~D Johnson.
\newblock ``Early fault-tolerant quantum computing''.
\newblock \href{https://dx.doi.org/10.1103/PRXQuantum.5.020101}{PRX Quantum {\bf 5}, 020101}~(2024).

\bibitem{ryser_combinatorial_1963}
Herbert~John Ryser.
\newblock ``Combinatorial {Mathematics}''.
\newblock \href{https://dx.doi.org/10.5948/UPO9781614440147}{American Mathematical Soc.} ~(1963).

\bibitem{jensen_sur_1906}
J.~L. W.~V. Jensen.
\newblock ``Sur les fonctions convexes et les inégalités entre les valeurs moyennes''.
\newblock \href{https://dx.doi.org/10.1007/BF02418571}{Acta Mathematica {\bf 30}, 175--193}~(1906).

\bibitem{aaronson_generalizing_2012}
Scott Aaronson and Travis Hance.
\newblock ``Generalizing and derandomizing gurvits's approximation algorithm for the permanent''~(2012).
\newblock  \href{http://arxiv.org/abs/1212.0025}{arXiv:1212.0025}.

\bibitem{janson_central_2021}
Svante Janson.
\newblock ``A central limit theorem for m-dependent variables''~(2021).
\newblock  \href{http://arxiv.org/abs/2108.12263}{arXiv:2108.12263}.

\bibitem{villalonga_efficient_2022}
Benjamin Villalonga, Murphy~Yuezhen Niu, Li~Li, Hartmut Neven, John~C. Platt, Vadim~N. Smelyanskiy, and Sergio Boixo.
\newblock ``Efficient approximation of experimental gaussian boson sampling''~(2022).
\newblock  \href{http://arxiv.org/abs/2109.11525}{arXiv:2109.11525}.

\end{thebibliography}
\newpage
\onecolumn
\addcontentsline{toc}{section}{Appendix} 
\part{Appendix} 
\parttoc % Insert the appendix TOC

\newpage
\section{Results on useful regimes of photon loss}
\label{applemm1}
\subsection{Proof of Lemma \ref{ryser_lemma}}

We now prove Lemma \ref{ryser_lemma} from the main text:

\vspace{1em}

\textit{Let $k=n-r$, there is a classical algorithm running in time $O(2^{r-1}r({n \choose n-r})^2)$  which exactly computes $\sum_{s_i^{n-k} \in \mathcal{L}(s^n_l)}p(s^{n-k}_i)$.}

\begin{proof}
    Each permanent $\mathsf{Per}(X_i)$ in $p(s^{n-k}_i)$ can be computed exactly via Ryser's algorithm \cite {ryser_combinatorial_1963} in time $O(2^{r-1}r)$. Thus, each $p(s^{n-k}_i)$ can be computed in $O({n \choose n-r}2^{r-1}r)$ time.  Further, we need $O({n \choose n-r}^22^{r-1}r)$ time to compute all the  $p(s^{n-k}_i)$'s, and therefore to compute $\sum_{s_i^{n-k} \in \mathcal{L}(s^n_l)}p(s^{n-k}_i)$, by an iterative approach where we first set a variable $\mathsf{sum}=0$, then iteratively add each computed $p(s^{n-k}_i)$ to $\mathsf{sum}$ as soon as we compute it.
\end{proof}

\section{Statistical error bounds for probability estimators}

We now give upper bounds for the statistical error of the estimators of the ideal probabilities from postselection, and for the estimators of the recycled probabilities. These are used in later proofs.

\subsection{Proof of Lemma \ref{recycled_stat_error_bound}}

\begin{lemma} \label{recycled_stat_error_bound}
With probability at least $1-2e^{-\alpha^2}$ for $\alpha>0$, the statistical error of each recycled probability, $\epsilon_{\text{hoeff},\tilde{p}_R^k(s^{n}_l)}$, is upper bounded
\begin{equation*}
|\epsilon_{\text{hoeff},\tilde{p}_R^k(s^{n}_l)}| \leq \alpha \frac{1}{{m-n+k \choose k}}\sqrt{\frac{{m \choose n}}{{n \choose k}(1-\eta)^{n-k}\eta^kN_{tot}}}.
\end{equation*}
\end{lemma}

\begin{proof}

The probability of losing $k$ out of $n$ photons is $$\text{Pr}(k)={n \choose k} \eta^k (1-\eta)^{n-k},$$ and so the number of samples from $n-k$--photon outputs is
$N_{tot,k}\approx{n \choose k} N_{tot} \eta^k (1-\eta)^{n-k}$, for $k \in \{0, \dots, n\}.$ 
There are ${m \choose n}$ recycled probabilities, and to guarantee the independence required for Hoeffding's inequality we use $N_{est,k}\approx \frac{{n \choose k} \eta^k(1-\eta)^{n-k}N_{tot}}{{m \choose n}}$ samples to estimate each $p_R^k(s^{n}_l)$. 
For each sample $w$ ranging from 1 to $N_{est,k}$ assign the value 1 to a random variable $X_w \in \{0,1\}$ if sample $w \in \mathcal{L}(s^n_l)$, and assign the value 0 to $X_w$ otherwise. The estimator $\tilde{p}_R^k(s^{n}_i)$ of $p_R^k(s^{n}_i)$ is then 
$$\tilde{p}_R^k(s^{n}_l)=\frac{\sum_w X_w}{{m-n+k \choose k}N_{est,k}},$$ 
where the ${m-n+k \choose k}^{-1}$ term is the normalisation factor detailed in eqn. \ref{recycling_normalisation_factor}. Finally, Hoeffding's inequality \cite{hoeffding_collected_1994} gives with confidence at least $1-2e^{-\alpha^2}$ that
\begin{equation*}
|\epsilon_{\text{hoeff},\tilde{p}_R^k(s^{n}_l)}| \leq \alpha \frac{1}{{m-n+k \choose k}}\sqrt{\frac{{m \choose n}}{{n \choose k}(1-\eta)^{n-k}\eta^kN_{tot}}}.
\end{equation*}

\end{proof}

Note that, the above proof holds when using the collected samples to estimate \emph{all} the probabilities, which is why we divided the $k$ photon samples by ${m \choose n}$. If we wish to estimate \emph{only one} output probability, then we need not divide by ${m \choose n}$ and we can use all the $N_{tot,k}$ samples.   Hoeffding's inequality gives us 
\begin{equation}
\label{eqstaterrth1}
|\epsilon_{\text{hoeff},\tilde{p}_R^k(s^{n}_l)}| \leq \alpha \frac{1}{{m-n+k \choose k}}\sqrt{\frac{1}{{n \choose k}(1-\eta)^{n-k}\eta^kN_{tot}}},
\end{equation}
with confidence at least $1-2e^{-\alpha^2}$.

\section{Interference deviation bounds}
\label{appintbounds}

\subsection{Proof of Lemma \ref{expectation_haar}}

We now prove Lemma \ref{expectation_haar} from the main text:

\vspace{1em}

\begin{proof}
Let  
$N'_k:=\big({m-n+k \choose k}-1\big){n \choose k}$. For $k>0$,  $I_{s^n_l,k}$ is a sum of $N'_k:=N_k-{n \choose k}$ terms of the form $p(s^n_j)\frac{1}{N_k}$. We can thus rewrite
\begin{equation*}
  I_{s^n_l,k}=  \sum_{s_i^{n-k} \in \mathcal{L}(s^n_l)}\sum_{ s^n_j \in \mathcal{G}(s_i^{n-k}), j \neq l}p(s^n_j)\frac{1}{N'_k}.
\end{equation*}
From \cite{aaronson_bosonsampling_2014}, when $m\gg n^2$, and $U$ is Haar random, which corresponds to boson sampling unitaries, each $p(s^n_j) \approx \frac{|\mathsf{Per}(X_i)|^2}{m^n}$, where $X_i$ is an $n \times n$ matrix of i.i.d. Gaussian random variables whose entries chosen from the complex normal distribution $\mathcal{N}_{\mathbb{C}}(0,1)$ of mean 0 and variance 1. $\mathsf{Per}(.)$ denotes the matrix permanent. Furthermore, from \cite{aaronson_bosonsampling_2014, nezami_permanent_2021}, we know that
\begin{equation*}
\mathrm{E}_{X \in \mathcal{G}_{n \times n}}(|\mathsf{Per}(X)|^2)=n!,
\end{equation*}
where $\mathrm{E}_{X \in \mathcal{G}_{n \times n}}(.)$
is the expectation value over the set of $n \times n$
Gaussian matrices with entries from $\mathcal{N}_{\mathbb{C}}(0,1)$. We can reexpress 
$$I_{s^n_l,k}= \frac{1}{N_k'}\sum_i\frac{|\mathsf{Per}(X_i)|^2}{m^n} .$$ Now, from the linearity of the expected value 

$$\mathrm{E}_{X \in \mathcal{G}_{n \times n}}(I_{s^n_l,k})=\mathrm{E}_{X \in \mathcal{G}_{n \times n}}\bigg( \frac{1}{N_k'}\sum_i\frac{|\mathsf{Per}(X_i)|^2}{m^n}\bigg)=\frac{1}{m^n N'_k}\sum_{i}\mathrm{E}_{X_i \in \mathcal{G}_{n \times n}}(|\mathsf{Per}(X_i)|^2)=\frac{n!}{m^n},$$
where the rightmost part follows from the fact that $\mathrm{E}_{X_i \in \mathcal{G}_{n \times n}}(|\mathsf{Per}(X_i)|^2)=n!$,  for all $i$.

Now $p_{unif} = {m \choose n}^{-1}$, and
\begin{equation*}
\begin{split}
{m \choose n}^{-1} &= \frac{n!(m-n)!}{m!} \\
&= \frac{n!}{(m)(m-1)\ldots(m-n+1)} \\
&= \frac{n!}{(m^n +f(m,n))} \\
\end{split}
\end{equation*}
with $f(m,n) \in O(m^{n-1}).$
This leads to 
\begin{equation*}
\begin{split}
 \frac{n!}{m^n} = {m \choose n}^{-1}\frac{(m^n + f(m,n)))}{m^n},
\end{split}
\end{equation*}
which then becomes
\begin{equation*}
\begin{split}
\frac{n!}{m^n} = p_{unif}\big(1 + g(m) \big).
\end{split}
\end{equation*}
with $g(m) \in O(m^{-1})$,
which was the result to be proved.
\end{proof}

\subsection{Proof of Thm.  \ref{Haarbound}}

We now prove Lemma \ref{Haarbound} from the main text:

\vspace{1em}

\textit{The deviation of interference terms around $p_{unif}$ for Haar random matrices is bounded}
\begin{equation*}
Pr\Big(\big|I_{s^n_l,k}-p_{unif}\big| \geq \epsilon_{bias,s^n_l}\Big)\leq \frac{np^2_{unif}}{\epsilon^2_{bias,s^n_l}},
\end{equation*}
\textit{where $\epsilon_{bias,s^n_l}$ is a positive real number.}

\begin{proof}
Let $X:=\sum_{s^{n-k}_i \in \mathcal{L}(s^n_l)}\sum_{s^n_j \in \mathcal{G}(s_i^{n-k}),j \neq l} p(s^n_j)$, $M:={n \choose k}({m-n+k \choose k}-1)$. For simplicity, relabel $p(s^n_j):=p_i$, so that $X$ can be rewritten as $X=\sum_{i=1, \dots, M}p_i$ is a sum of, possibly dependent, random variables $p_i$ and $\mathsf{E}_U(p_i):=\mathsf{E}_U(p)=\frac{n!}{m^n} \approx p_{unif}$, where $\mathsf{E}_U(p_i)$ is the expectation value over the Haar measure of the $m$-mode unitary group $\mathsf{U}(m)$.  From \cite{aaronson_computational_2011}, we know that this expectation value is given by $\frac{n!}{m^n}$ in the no-collision regime.
Let $\sigma^2=\mathsf{Var}(X)=\mathsf{E}_U(X^2)-(\mathsf{E}_U(X))^2$. Expanding out, we get that $\sigma^2=M\mathsf{E}_U(p^2)-M(\mathsf{E}_U(p))^2$+$2\sum_i\sum_{j>i}\mathsf{Cov}(p_i,p_j)$, where $\mathsf{Cov}(p_i,p_j)=\mathsf{E}_U(p_ip_j)-\mathsf{E}_U(p_i)\mathsf{E}_U(p_j)$ , and $\mathsf{E}_U(p^2)=\mathsf{E}_U(p_i^2)=\frac{(n+1)!n!}{m^{2n}} \approx (n+1) p^2_{unif}$, for all $i$, and this is from \cite{aaronson_computational_2011}.
 Using the Cauchy-Schwarz inequality, 
 \begin{equation*}
     |\mathsf{Cov}(p_i,p_j)| \leq \sqrt{\mathsf{Var}(p_i)\mathsf{Var}(p_j)} \leq \sqrt{(\mathsf{Var}(p))^2} \leq \mathsf{Var}(p) \leq \mathsf{E}_U(p^2)-(\mathsf{E}_U(p)) ^2
 \end{equation*}
 Thus,
 \begin{multline*}
     \sigma^2 \leq M(\mathsf{E}_U(p^2)-(\mathsf{E}_U(p))^2+ 2\sum_i \sum_{j>i} |\mathsf{Cov}(p_i,p_j)| \leq M(\mathsf{E}_U(p^2)-(\mathsf{E}_U(p))^2) +M(M-1)(\mathsf{E}_U(p^2)-(\mathsf{E}_U(p))^2) \leq \\ M^2(\mathsf{E}_U(p^2)-(\mathsf{E}_U(p))^2) \leq M^2np^2_{unif}.
 \end{multline*}
 Finally, using Chebyshev's inequality we have 
 \begin{equation*}
 Pr\Big(\big|X-\mathsf{E}_U(X)\big| \geq \epsilon_{bias,s^n_l}\Big) \leq \frac{\sigma^2}{\epsilon^2_{bias,s^n_l}},
 \end{equation*}
 and noting that $\mathsf{E}_U(X)=Mp_{unif}=\sum_i \mathsf{E}_U(p)$ from linearity of expectation value, then dividing both sides of $|X-\mathsf{E}_U(X)| \geq \epsilon_{bias,s^n_l}$ by $M$ and redefining $\epsilon_{bias,s^n_l}:=\frac{\epsilon_{bias,s^n_l}}{M}$, and finally using $\sigma^2 \leq M^2np_{unif}^2$, we get
 \begin{equation*}
    Pr\Big(\big|I_{s^n_l,k}-p_{unif}\big| \geq \epsilon_{bias,s^n_l}\Big) \leq \frac{np^2_{unif}}{\epsilon^2_{bias,s^n_l}}.
 \end{equation*}
 This completes the proof.
\end{proof}

\subsection{Proof of Lemma \ref{expectation_arbitrary}}

We now prove Lemma \ref{expectation_arbitrary} from the main text:

\vspace{1em}

\begin{equation*}
\begin{split}
\mathbf{E}_{D^k_R}\big(I_{s^n_l, k}\big) & = p_{unif},\\
\end{split}
\end{equation*}
\textit{where $\mathbf{E}_{D^k_R}(.)$ denotes the expectation value over $D^k_R$.}

\begin{proof}
Recalling that
\begin{equation*}
    p_R(s^n_l)=\frac{p(s^n_l)}{{m-n+k \choose k}}+\frac{N'_k}{N_k}I_{s^n_l, k},
\end{equation*}
and 
\begin{equation*}
    I_{s^n_l, k}=\sum_{s_i^{n-k} \in \mathcal{L}(s^n_l)}\sum_{ s^n_j \in \mathcal{G}(s_i^{n-k}), j \neq l}p(s^n_j)\frac{1}{\big({m-n+k \choose k}-1\big){n \choose k}}.
\end{equation*}

\begin{equation*}
\begin{split}
  \mathbf{E}_{D^k_R}\big(I_{s^n_l, k}\big)   &= \frac{1}{{m \choose n}}\sum_{l=1}^{{m \choose n}}\sum_{s_i^{n-k} \in \mathcal{L}(s^n_l)}\sum_{ s^n_j \in \mathcal{G}(s_i^{n-k}), j \neq l}p(s^n_j)\frac{1}{\big({m-n+k \choose k}-1\big){n \choose k}}\\
  &= \frac{1}{{m \choose n}\big({m-n+k \choose k}-1\big){n \choose k}}\sum_{l=1}^{{m \choose n}}\sum_{s_i^{n-k} \in \mathcal{L}(s^n_l)}\Bigg(\sum_{ s^n_j \in \mathcal{G}(s_i^{n-k})}p(s^n_j) - p(s^n_l)\Bigg)\\
  &= \frac{1}{{m \choose n}\big({m-n+k \choose k}-1\big){n \choose k}}\Bigg(\sum_{l=1}^{{m \choose n}}\sum_{s_i^{n-k} \in \mathcal{L}(s^n_l)}\sum_{ s^n_j \in \mathcal{G}(s_i^{n-k})}p(s^n_j) - \sum_{l=1}^{{m \choose n}}{n \choose k}p(s^n_l)\bigg)\\
  &= \frac{1}{{m \choose n}\big({m-n+k \choose k}-1\big){n \choose k}}\Bigg(\sum_{l=1}^{{m \choose n}}{m-n+k \choose k}{n \choose k}p(s^n_l) - \sum_{l=1}^{{m \choose n}}{n \choose k}p(s^n_l)\bigg)\\
  &= \frac{1}{{m \choose n}\big({m-n+k \choose k}-1\big){n \choose k}}\Bigg({m-n+k \choose k}{n \choose k} - {n \choose k}\bigg)\\
  &= p_{unif}.\\
\end{split}
\end{equation*}
Where between the third and fourth lines we have used the cardinality of the sets $|\mathcal{L}(s^n_l)|={n \choose k}$ and $|\mathcal{G}(s_i^{n-k})|={{m-n+k} \choose k}$, and between the fourth and fifth lines that $\sum_{l=1}^{{m \choose n}}p(s^n_l)=1$.
\end{proof}

\subsection{Proof of Thm.  \ref{nonHaarbound}}

 \label{proof_nonHaarbound}
We now prove Thm.  \ref{nonHaarbound} from the main text:

\vspace{1em}
\textit{The deviation of interference terms around $p_{unif}$ for an arbitrary matrix is bounded }
\begin{equation*}
\begin{split}
\text{Pr}\Big(\big|I_{s^n_l, k}-  p_{unif}\big| \geq \epsilon_{bias,s^n_l} \Big) & \leq \frac{ p_{unif}}{\epsilon_{bias,s^n_l}^2} ,\\
\end{split}
\end{equation*}
\textit{where $\epsilon_{bias,s^n_l}$ is a positive real number.}

\begin{proof}
    Let  $X:=I_{s^n_l,k}$ be a random variable defined over the uniform distribution of bit strings $s^n_l$ that returns interference terms as values, $\mu:=\mathsf{E}_{s^n_l}(X)=p_{unif}=\frac{1}{{m \choose n}}$, with $\mathsf{E}_{s^n_l}(.)$ denoting the expectation value of $X$ over the uniform distribution of $s^n_l$, it is equal to $p_{unif}$ from lemma \ref{expectation_arbitrary}. Let $\sigma^2:=\mathsf{Var}(X)$. From Chebyshev's inequality, 
    \begin{equation*}
        Pr(|X-\mu|\geq \epsilon_{bias,s^n_l}) \leq \frac{\sigma^2}{\epsilon^2_{bias,s^n_l}}
    \end{equation*}
    Let $M:=\mathsf{max}_{s^n_l}(X)$, $m:=\mathsf{min}_{s^n_l}(X)$, and note that $M \leq 1$ and $m \geq 0$. We now use the Bhatia-Davis inequality \cite{bhatia_better_2000}
    \begin{equation*}
    \sigma^2 \leq (M-\mu)(\mu-m)
    \end{equation*}
    with the upper and lower bounds on $M$ and $m$ to obtain 
    \begin{equation*}
        \sigma^2 \leq (1-p_{unif})(p_{unif}) \leq p_{unif}.
    \end{equation*}
    Replacing this in Chebyshev's inequality completes the proof.
\end{proof}

\subsection{Proofs of Lemmas \ref{Var_recycle_upperbound} and \ref{Var_interfer_upperbound}}

Both proofs will make use of Jensen's inequality \cite{jensen_sur_1906}: 
\begin{theorem}
{\normalfont \textbf{(Jensen's inequality)}} Let $f(x)$ be a convex function on $(a,b)$ and suppose $a< x_1 \leq x_2 \leq \ldots \leq x_n < b$. Then 
\begin{equation*}
\frac{f(x_1) + f(x_2) + \ldots + f(x_n)}{n} \geq f\Big(\frac{x_1 + x_2 + \ldots + x_n}{n}\Big).
\end{equation*}
Equality holds if, and only if, $x_1=x_2=\ldots=x_n$.
\end{theorem}

\vspace{1em}

We now prove Lemma \ref{Var_recycle_upperbound} from the main text:

\vspace{1em}

\textit{
The variance of the set of recycled probabilities is less than or equal to the variance of the set of ideal probabilities, that is
\begin{equation*}
\mathsf{Var}\big(\{p(s^{n}_l)\}_l\big) \geq \mathsf{Var}\big(\{p_{R}^k(s^{n}_l)\}_l\big).
\end{equation*}
}

\begin{proof}
Let $\mu = p_{unif}$ and $f(x) = (x-\mu)^2$.
Defining the variance of the ideal $n$ output photon distribution as
\begin{equation*}
\begin{split}
 \mathsf{Var}\big(\{p(s^{n}_l)\}_l\big)  &:= {m \choose n}^{-1} \sum_l (p(s^{n}_l) - \mu)^2,\\
\end{split}
\end{equation*}
and the variance of the $n-k$ output photon recycled distribution as
\begin{equation*}
\begin{split}
  \mathsf{Var}\big(\{p_{R}^k(s^{n}_l)\}_l\big)   &:= {m \choose n}^{-1}\sum_l (p_{R}^k(s^{n}_l) - \mu)^2.\\
\end{split}
\end{equation*}
These definitions are equivalent to the variance of random variables uniformly distributed over the set of bit strings $\{s^n_l\}$ that return ideal $n$ output photon probabilities in the first case, and $n-k$ output photon recycled probabilities in the second case. 
Recall that the recycled probabilities are defined by the expression
\begin{equation*}
\begin{split}
  p_R^k(s^n_l)   &:= \sum_{s_i^{n-k} \in \mathcal{L}(s^n_l)}\sum_{ s^n_j \in \mathcal{G}(s_i^{n-k})}p(s^n_j)\frac{1}{{m-n+k \choose k}{n \choose k}}.\\
\end{split}
\end{equation*}
As $f(x)$ is concave, by applying Jensen's inequality it is true that 
\begin{equation*}
\sum_{s_i^{n-k} \in \mathcal{L}(s^n_l)}\sum_{ s^n_j \in \mathcal{G}(s_i^{n-k})}\frac{f(p(s^n_j))}{{m-n+k \choose k}{n \choose k}} \geq f\bigg(\sum_{s_i^{n-k} \in \mathcal{L}(s^n_l)}\sum_{ s^n_j \in \mathcal{G}(s_i^{n-k})}p(s^n_j)\frac{1}{{m-n+k \choose k}{n \choose k}}\bigg).
\end{equation*}
From which it follows
\begin{equation*}
{m \choose n}^{-1} \sum_l \sum_{s_i^{n-k} \in \mathcal{L}(s^n_l)}\sum_{ s^n_j \in \mathcal{G}(s_i^{n-k})}\frac{f(p(s^n_j))}{{m-n+k \choose k}{n \choose k}} \geq {m \choose n}^{-1} \sum_l f\bigg(\sum_{s_i^{n-k} \in \mathcal{L}(s^n_l)}\sum_{ s^n_j \in \mathcal{G}(s_i^{n-k})}p(s^n_j)\frac{1}{{m-n+k \choose k}{n \choose k}}\bigg),
\end{equation*}
and after simplifying
\begin{equation*}
{m \choose n}^{-1} \sum_l f(p(s^n_j)) \geq {m \choose n}^{-1} \sum_l f\bigg(\sum_{s_i^{n-k} \in \mathcal{L}(s^n_l)}\sum_{ s^n_j \in \mathcal{G}(s_i^{n-k})}p(s^n_j)\frac{1}{{m-n+k \choose k}{n \choose k}}\bigg).
\end{equation*}
And this is just: $\mathsf{Var}\big(\{p(s^{n}_l)\}_l\big) \geq \mathsf{Var}\big(\{p_{R}^k(s^{n}_l)\}_l\big)$, the result to be shown.
\end{proof}

\vspace{2em}
We now prove Lemma \ref{Var_interfer_upperbound} from the main text:

\vspace{1em}

\textit{
The variance of the set of interference terms is less than or equal to the variance of the set of ideal probabilities, that is
\begin{equation*}
\mathsf{Var}\big(\{p(s^{n}_l)\}_l\big) \geq \mathsf{Var}\big(\{I_{s^n_l, k}\}_l\big).
\end{equation*}
}

\begin{proof}
Let $\mu = p_{unif}$ and $f(x) = (x-\mu)^2$. 
Defining the variance of the ideal $n$ output photon distribution as
\begin{equation*}
\begin{split}
 \mathsf{Var}\big(\{p(s^{n}_l)\}_l\big)  &:= {m \choose n}^{-1} \sum_l (p(s^{n}_l) - \mu)^2,\\
\end{split}
\end{equation*}
and the variance of the interference terms for the $n-k$ output photon recycled distribution as
\begin{equation*}
\begin{split}
  \mathsf{Var}\big(\{I_{s^n_l, k}\}_l\big)   &:= {m \choose n}^{-1}\sum_l \big(I_{s^n_l, k} - \mu\big)^2.\\
\end{split}
\end{equation*}
These definitions are equivalent to the variance of random variables uniformly distributed over the set of bit strings $\{s^n_l\}$ that return ideal $n$ output photon probabilities in the first case, and interference terms from the $n-k$ output photon recycled distribution in the second case. 
Recall that the interference terms of a recycled distribution are defined by the expression
\begin{equation*}
\begin{split}
I_{s^n_l, k}  &= \sum_{s_i^{n-k} \in \mathcal{L}(s^n_l)}\sum_{ s^n_j \in \mathcal{G}(s_i^{n-k}), j \neq l} p(s^n_j)\frac{1}{\big({m-n+k \choose k}-1\big){n \choose k}}.\\
\end{split}
\end{equation*}
As $f(x)$ is concave, by applying Jensen's inequality it is true that 
\begin{equation*}
\sum_{s_i^{n-k} \in \mathcal{L}(s^n_l)}\sum_{ s^n_j \in \mathcal{G}(s_i^{n-k}), j \neq l} \frac{f(p(s^n_j))}{\big({m-n+k \choose k}-1\big){n \choose k}} \geq f\bigg(\sum_{s_i^{n-k} \in \mathcal{L}(s^n_l)}\sum_{ s^n_j \in \mathcal{G}(s_i^{n-k}), j \neq l}p(s^n_j)\frac{1}{\big({m-n+k \choose k}-1\big){n \choose k}}\bigg).
\end{equation*}
From which it follows
\begin{equation*}
\begin{split}
{m \choose n}^{-1} \sum_l \sum_{s_i^{n-k} \in \mathcal{L}(s^n_l)}\sum_{ s^n_j \in \mathcal{G}(s_i^{n-k}), j \neq l}&\frac{f(p(s^n_j))}{\big({m-n+k \choose k}-1\big){n \choose k}} \\
&\hspace{2em}\geq {m \choose n}^{-1} \sum_l f\bigg(\sum_{s_i^{n-k} \in \mathcal{L}(s^n_l)}\sum_{ s^n_j \in \mathcal{G}(s_i^{n-k}), j \neq l}p(s^n_j)\frac{1}{\big({m-n+k \choose k}-1\big){n \choose k}}\bigg),
\end{split}
\end{equation*}
and after simplifying
\begin{equation*}
{m \choose n}^{-1}  \sum_l f(p(s^n_j)) \geq {m \choose n}^{-1} \sum_l f\bigg(\sum_{s_i^{n-k} \in \mathcal{L}(s^n_l)}\sum_{ s^n_j \in \mathcal{G}(s_i^{n-k})}p(s^n_j)\frac{1}{\big({m-n+k \choose k}-1\big){n \choose k}}\bigg).
\end{equation*}
And this is just: $\mathsf{Var}\big(\{p(s^{n}_l)\}_l\big) \geq \mathsf{Var}\big(\{I_{s^n_l, k}\}_l)\}_l\big)$, the result to be shown.
\end{proof}

\subsection{Proof of Thm.  \ref{p_upper_BD_bound}}

We now prove Thm.  \ref{p_upper_BD_bound} from the main text:

\vspace{1em}

\textit{The deviation of interference terms around $p_{unif}$ for an arbitrary matrix is bounded }
\begin{equation*}
\begin{split}
\text{Pr}\Big(\big|I_{s^n_l, k}-  p_{unif}\big| \geq \epsilon_{bias,s^n_l} \Big) & \leq \frac{ p_{unif} p_{{upper}}} {\epsilon_{bias,s^n_l}^2} + \delta\Big(1 - \frac{ p_{unif} p_{{upper}}}{\epsilon_{bias,s^n_l}^2}\Big)  ,\\
\end{split}
\end{equation*}
where \textit{$\epsilon_{bias,s^n_l}$ is a positive real number, and $p_{\text{upper}}$ is an empirically computed upper bound on the largest probability of the ideal $n$ output photon probability distribution with confidence $1-\delta$.}

\begin{proof}

    Let  $X:=I_{s^n_l,k}$ be a random variable defined over the uniform distribution of bit strings $s^n_l$ that returns interference terms as values. $\mathsf{E}_{s^n_l}(X)$ denotes the expectation value of random variable $X$ over the uniform distribution of $s^n_l$, and let $\mu_X:=\mathsf{E}_{s^n_l}(X)=p_{unif}$, the latter equality comes from lemma \ref{expectation_arbitrary}. Let $Y:=p(s^{n}_l)$ be a random variable defined over the uniform distribution of bit strings $s^n_l$ that returns $n$--photon output state probabilities as values, and $\mu_Y:=\mathsf{E}_{s^n_l}(Y)=p_{unif}$, where $\mathsf{E}_{s^n_l}(Y)$ denotes the expectation value of random variable $Y$ over the uniform distribution of $s^n_l$. Using Chebyshev's inequality we have that
    \begin{equation*}
        Pr(|X-\mu_X|\geq \epsilon_{bias,s^n_l}) \leq \frac{\mathsf{Var}(X)}{\epsilon^2_{bias,s^n_l}}
    \end{equation*}
    By lemma \ref{Var_interfer_upperbound} we have that $\mathsf{Var}(X) \leq \mathsf{Var}(Y)$. Let $M:=\mathsf{max}_{s^n_l}(Y)$, $m:=\mathsf{min}_{s^n_l}(Y)$. An empirical upper bound on the largest probability on the largest probability, $p_{{upper}}$, may be computed from sample data such that $p_{{upper}}\geq M$. A Hoeffding inequality can be used to bound the statistical error of the empirical estimator of $M$, $\tilde{M}$, so that
    $$
    |\epsilon_{hoeff,M}| \leq \sqrt{\frac{\log(\frac{2}{\delta})}{2N_{est,M}}},
    $$
    where $N_{est,M}$ is the number of samples used to compute the estimator $\tilde{M}$, and $1-\delta$ is the confidence. Defining $\epsilon_{hoeff,M}^{max}:=\sqrt{\frac{\log(\frac{2}{\delta})}{2N_{est,M}}}$ and $p_{upper} := \tilde{M} + \epsilon_{hoeff,M}^{max}$. Then $p_{upper} > M$ with confidence $1-\delta$. The smallest probability is lower bounded $m \geq 0$. We now use the Bhatia-Davis inequality \cite{bhatia_better_2000}
    \begin{equation*}
    \mathsf{Var}(Y) \leq (M-\mu_Y)(\mu_Y-m) = (M-p_{unif})(p_{unif}-m)
    \end{equation*}
    the upper and lower bounds on $M$ and $m$ then lead to  
    \begin{equation*}
        (M-p_{unif})(p_{unif}-m) \leq (p_{{upper}}-p_{unif})(p_{unif})  \leq p_{{upper}}p_{unif}.
    \end{equation*}
    The confidence that $|I_{s^n_l, k}-  p_{unif}| < \epsilon_{bias,s^n_l}$
    would be $1-\frac{ p_{unif} p_{{upper}}} {\epsilon_{bias,s^n_l}^2}$, however this must be multiplied by the independent confidence of the statement $p_{upper} > M$, which is $1-\delta$, to get an overall confidence of $1 - \frac{ p_{unif} p_{{upper}}} {\epsilon_{bias,s^n_l}^2} - \delta+ \frac{ p_{unif} p_{{upper}}\delta}{\epsilon_{bias,s^n_l}^2}$. This means the Chebyshev inequality becomes
\begin{equation*}
\begin{split}
\text{Pr}\Big(\big|I_{s^n_l, k}-  p_{unif}\big| \geq \epsilon_{bias,s^n_l} \Big) & \leq \frac{ p_{unif} p_{{upper}}} {\epsilon_{bias,s^n_l}^2} + \delta\Big(1 - \frac{ p_{unif} p_{{upper}}}{\epsilon_{bias,s^n_l}^2}\Big) ,\\
\end{split}
\end{equation*}
and the result follows.
\end{proof}

\section{Linear solving}
\label{applinearsolv}

\subsection{Bounding errors in linear solving }
\label{appperfguarlinsolv}

We first state and prove two results that we will use. These are upper bounds for the statistical error and the bias error of the mitigated value respectively, these we state as lemmas \ref{miti_bias_error} and \ref{miti_stat_error}.

\begin{lemma}
\label{miti_bias_error}
    $\epsilon_{mit, bias} \leq 3({m-n+k \choose k}-1)|\epsilon_{bias,s^n_l}|$.
\end{lemma}
\begin{proof}
Without statistical error the mitigated value may be written $p_{mit}(s^n_l)=\Big|{m-n+k \choose k}p_R(s^n_l)-{m-n+k \choose k}\frac{N'_k}{N_k} p_{unif}\Big|$. 
Let $\epsilon_{mit, bias}:=|p_{mit}(s^n_l)-p(s^n_l)|$ and $A_k:={m-n+k \choose k}p_R(s^n_l)-{m-n+k \choose k}\frac{N'_k}{N_k} p_{unif}$, so that $\epsilon_{mit, bias}=||A_k|-p(s^n_l)|$, then $p_{mit}(s^n_l)=|A_k|$, and $p(s^n_l)=A_k-{m-n+k \choose k}\frac{N'_k}{N_k}\epsilon_{bias,s^n_l}$ (see main text). 
\vspace{1em}

Consider the first possible case where $A_k <0$, this means that 
\begin{equation*}
\begin{split}
\epsilon_{mit, bias}& = | -A_k-p(s^n_l)| \\
&=|-A_k-A_k+{m-n+k \choose k}\frac{N'_k}{N_k}\epsilon_{bias,s^n_l}| \\
&\leq 2|A_k| + {m-n+k \choose k}\frac{N'_k}{N_k}|\epsilon_{bias,s^n_l}|,
\end{split}
\end{equation*}
 by triangle inequality. 
Since $p(s^n_l) \geq 0$ by definition, then ${m-n+k \choose k}\frac{N'_k}{N_k}|\epsilon_{bias,s^n_l}|=-{m-n+k \choose k}\epsilon_{bias,s^n_l} \geq -A_k$. Since $|A_k|=-A_k$, it follows that $|A_k| \leq {m-n+k \choose k}|\frac{N'_k}{N_k}\epsilon_{bias,s^n_l}|$, and therefore that $\epsilon_{mit, bias} \leq 2|A_k| + {m-n+k \choose k}\frac{N'_k}{N_k}|\epsilon_{bias,s^n_l}| \leq 3{m-n+k \choose k}\frac{N'_k}{N_k}|\epsilon_{bias,s^n_l}|$. 
\vspace{1em}

Now, consider the second possible case where $A_k \geq 0$, then $|A_k|=A_k$, and therefore $\epsilon_{mit, bias}={m-n+k \choose k}\frac{N'_k}{N_k}|\epsilon_{bias,s^n_l}| \leq 3{m-n+k \choose k}\frac{N'_k}{N_k}|\epsilon_{bias,s^n_l}|.$ This completes the proof.
\end{proof}
\vspace{1em}

\begin{lemma}
\label{miti_stat_error}
$|\tilde{p}_{mit}(s^n_l)-{p}_{mit}(s^n_l)| \leq 3 {m-n+k \choose k} |\epsilon_{\text{hoeff},\tilde{p}_R^k(s^{n}_l)}|.$
\end{lemma}
\begin{proof}
First, note that the estimator of the mitigated value may be written $\tilde{p}_{mit}=|A_k+ \epsilon_{mit,stat}|,$ where $A_k$ is defined previously, and $\epsilon_{mit,stat}={m-n+k \choose k}\tilde{p}_R(s^n_l)-{m-n+k \choose k}p_R(s^n_l), $ and $|\epsilon_{mit,stat}| \leq {m-n+k \choose k}|\epsilon_{\text{hoeff},\tilde{p}_R^k(s^{n}_l)}|.$ And therefore $|\tilde{p}_{mit}-p_{mit}|=||A_k+\epsilon_{mit,stat}|-|A_k||.$

Consider the case where $$|A_k+ \epsilon_{mit,stat}|=A_k + \epsilon_{mit,stat}.$$
If $A_k>0$, then $|A_k|=A_k$ and therefore $|\tilde{p}_{mit}-p_{mit}|=|\epsilon_{mit,stat} |\leq {m-n+k \choose k}|\epsilon_{\text{hoeff},\tilde{p}_R^k(s^{n}_l)}|$. 
If $A_k < 0$, then $A_k=-|A_k|$, and furthermore, since $A_k + \epsilon_{mit,stat} \geq 0$, then $\epsilon_{mit,stat} \geq |A_k|$. 
In this sub-case, we have that $|\tilde{p}_{mit}-p_{mit}|=|-2|A_k|-\epsilon_{mit,stat}| \leq 2|A_k| + |\epsilon_{mit,stat}| \leq 3 \epsilon_{mit,stat} \leq 3 {m-n+k \choose k} |\epsilon_{hoeff}|$.

Now, consider the case where $$|A_k+ \epsilon_{mit,stat}|=-A_k-\epsilon_{mit,stat}.$$
If $A_k<0$, then $-A_k=|A_k|$, in this sub-case we have $|\tilde{p}_{mit}-p_{mit}|=|\epsilon_{mit,stat}| \leq {m-n+k \choose k} \epsilon_{hoeff}.$ 
Now, if $A_k>0$, then $-A_k=-|A_k|$ and furthermore, since $-A_k-\epsilon_{mit,stat} \geq 0$, then $-\epsilon_{mit,stat} \geq |A_k|$. 
In this sub-case, $|\tilde{p}_{mit}-p_{mit}|=|-2|A_k|-\epsilon_{mit,stat}| \leq 3 |\epsilon_{mit,stat}| \leq 3 {m-n+k \choose k} |\epsilon_{\text{hoeff},\tilde{p}_R^k(s^{n}_l)}|.$ 

In all possible cases, $$|\tilde{p}_{mit}-p_{mit}| \leq 3 {m-n+k \choose k} |\epsilon_{\text{hoeff},\tilde{p}_R^k(s^{n}_l)}|,$$ which was the result to be proved.
\end{proof}

\vspace{1em}
% \begin{highlight}
\subsection {Proof of Thm. \ref{thpost}}
\label{proofmainres}
From Lemmas \ref{miti_bias_error} and \ref{miti_stat_error} we have that the overall error for linear solving recycling mitigation is given by 
\begin{equation}
|\tilde{p}_{mit}(s^n_l)-p_{id}(s^n_l)| \leq 3 { m-n+k \choose k} (|\epsilon_{bias,s^n_l}|+|\epsilon_{\text{hoeff},\tilde{p}_R^k(s^{n}_l)}|).
\end{equation}
Now rearranging Thm. \ref{nonHaarbound} in terms of confidence parameter $\delta_{\text{bias}} \in (0,1)$, we get that
\begin{equation} \label{prev_eq_arg}
    \text{Pr}\bigg(|I_{s^n_l,k} - p_{unif}| \leq \sqrt{\frac{p_{\text{unif}}}{\delta_{\text{bias}}}} \bigg) \geq 1 - \delta_{\text{bias}}.
\end{equation}
Now eqn. \ref{eqstaterrth1} states that
\begin{equation}
|\epsilon_{\text{hoeff},\tilde{p}_R^k(s^{n}_l)}| \leq \alpha \frac{1}{{m-n+k \choose k}}\sqrt{\frac{1}{{n \choose k}(1-\eta)^{n-k}\eta^kN_{tot}}},
\end{equation}
with confidence at least $1-2e^{-\alpha^2}$.
Stated in the same form as eqn. \ref{prev_eq_arg}, this is
\begin{equation}
\text{Pr}\bigg(|\epsilon_{\text{hoeff},\tilde{p}_R^k(s^{n}_l)}| \leq \alpha \frac{1}{{m-n+k \choose k}}\sqrt{\frac{1}{{n \choose k}(1-\eta)^{n-k}\eta^kN_{tot}}} \bigg) \geq 1-2e^{-\alpha^2}.
\end{equation}
Now applying union bound, with $|\epsilon_{bias,s^n_l}|:=|I_{s^n_l,k} - p_{unif}|$, we have that 
\begin{equation}
    \text{Pr}\bigg(|\epsilon_{bias,s^n_l}|+|\epsilon_{\text{hoeff},\tilde{p}_R^k(s^{n}_l)}| \leq  \sqrt{\frac{p_{\text{unif}}}{\delta_{\text{bias}}}} + \alpha \frac{1}{{m-n+k \choose k}}\sqrt{\frac{1}{{n \choose k}(1-\eta)^{n-k}\eta^kN_{tot}}}\bigg) \geq 1 - \delta_{\text{bias}} - 2e^{-\alpha^2}.
\end{equation}
To obtain this expression in terms of a single confidence parameter we set $\delta_{bias}=2e^{-\alpha^2}$, and the previous expression becomes
\begin{equation}
    \text{Pr}\bigg(|\epsilon_{bias,s^n_l}|+|\epsilon_{\text{hoeff},\tilde{p}_R^k(s^{n}_l)}| \leq  \frac{e^{\alpha^2/2}}{\sqrt{2}}\sqrt{p_{\text{unif}}} + \alpha \frac{1}{{m-n+k \choose k}}\sqrt{\frac{1}{{n \choose k}(1-\eta)^{n-k}\eta^kN_{tot}}}\bigg) \geq 1 - 4e^{-\alpha^2}.
\end{equation}
We then get that
\begin{equation}
|\tilde{p}_{mit}(s^n_l)-p_{id}(s^n_l)| \leq 3 \bigg[ \frac{e^{\alpha^2/2}}{\sqrt{2}}\sqrt{p_{\text{unif}}} { m-n+k \choose k} + \alpha \sqrt{\frac{1}{{n \choose k}(1-\eta)^{n-k}\eta^kN_{tot}}} \bigg],
\end{equation}
which holds with confidence $1 - 4e^{-\alpha^2}$.
% \end{highlight}

To simplify this we will now make use of the well-known inequality that, for integers $a \geq b \geq 1$,
\begin{equation}
    \Big(\frac{a}{b}\Big)^b \leq {a \choose b} \leq \Big(\frac{ea}{b}\Big)^b.
\end{equation}
Using this inequality, we have that
\begin{equation}
{ m-n+k \choose k} \leq \bigg(\frac{e(m-n+k)}{k}\bigg)^k,
\end{equation}
and
\begin{equation}
 \Big(\frac{n}{k}\Big)^k \leq { n \choose k},
\end{equation}
so that
\begin{equation}
 { n \choose k}^{-1} \leq \Big(\frac{k}{n}\Big)^k .
\end{equation}
This may be used to then derive the following:
\begin{equation}
|\tilde{p}_{mit}(s^n_l)-p_{id}(s^n_l)| \leq 3 \bigg[ \frac{e^{\alpha^2/2}}{\sqrt{2}}\sqrt{p_{\text{unif}}} \bigg(\frac{e(m-n+k)}{k}\bigg)^k + \alpha \Big(\frac{k}{n}\Big)^{k/2} \sqrt{\frac{1}{(1-\eta)^{n-k}\eta^kN_{tot}}} \bigg],
\end{equation}
which holds with confidence $1 - 4e^{-\alpha^2}$.
We now set $\delta = 4e^{-\alpha^2}$.
Then $\alpha = \sqrt{\ln(4/\delta)}$ and $\frac{e^{\alpha^2/2}}{\sqrt{2}} = \sqrt{\frac{2}{\delta}}$.
Substituting this, we then get the following bound
\begin{equation}
\big|\tilde{p}_{\mathrm{mit}}(s^n_l)-p_{\mathrm{id}}(s^n_l)\big| \leq 
3\Bigg[\sqrt{\frac{2p_{\mathrm{unif}}}{\delta}} \left(\frac{e(m-n+k)}{k}\right)^k + \sqrt{\ln(4/\delta)}\, \left(\frac{k}{n}\right)^{k/2} \sqrt{\frac{1}{(1-\eta)^{n-k}\eta^kN_{\mathrm{tot}}}} \Bigg],
\end{equation}
which holds with confidence $1-\delta$. 
Finally, since $p_{\mathrm{unif}} = \binom{m}{n}^{-1}$ and
\begin{equation}
\binom{m}{n} \ge \left(\frac{m}{n}\right)^n,
\end{equation}
we have
\begin{equation}
p_{\mathrm{unif}} \le \left(\frac{n}{m}\right)^n,
\qquad
\sqrt{p_{\mathrm{unif}}} \le \left(\frac{n}{m}\right)^{n/2}.
\end{equation}
Substituting this bound yields the expression for $f(n,m,k,\eta,N_{tot})$ in Theorem~\ref{thpost}.

\subsection{Linear solving with dependency}

We now consider the case where there is positive correlation of the interference to ideal probabilities within the recycled probabilities. 
This effect can be modelled as a linear dependence of interference terms on respective ideal probabilities. 
For a given recycled probability, a dependency term $d_k(s_l^n)$ can be used to express the interference term as a linear function of $p(s_l^n)$ and $p_{unif}$. Then
\begin{equation*}
 \big(1-d_k(s^n_l)\big)p_{unif} + d_k(s^n_l)p(s^n_l) = I_{s^n_l, k}
\end{equation*} 
defines the value $d_k(s^n_l)$, where $k\leq n-1$. 
For each interference term, $I_{s^n_l, k}$, there is a corresponding dependency term, $d_k(s^n_l)$. The dependency term, $d_k(s^n_l)$, encodes the dependence of interference term $I_{s^n_l, k}$ on ideal probability $p(s^n_l)$ for the recycled probability $p_R^k(s^n_l)$. The set of dependency terms is then denoted $\{d_k(s^n_l)\}_l$. 

The original linear solving method involves approximating the interference term as $p_{unif}$, and solving to find the mitigated probability. Now we propose first extracting an average dependency term from the distribution, which we will call the \textit{average dependency term} $d_k$, and then approximating the interference terms as $ \big(1-d_k\big)p_{unif} + d_k p(s^n_l) $ and solving as before. The motivation for this variation on the original protocol is to improve mitigation performance by capturing the enhanced ideal probability signal caused by this correlation behaviour.
In the original solving method, excluding statistical error, the recycled probability is decomposed as
\begin{equation*}
    p_R(s^n_l)=\frac{p(s^n_l)}{{m-n+k \choose k}}+\frac{N'_k}{N_k}(p_{unif}+\epsilon_{bias,s^n_l}),
\end{equation*}
where $\epsilon_{bias,s^n_l}$ is the bias error from approximating $I_{s^n_l, k}$ as $\frac{N'_k}{N_k}p_{unif}$. If the interference term $I_{s^n_l, k}$ is instead approximated as $\frac{N'_k}{N_k}\Big( \big(1-d_k\big)p_{unif} + d_k p(s^n_l) \Big)$, then the expression becomes
\begin{equation*}
\begin{split}
    p_R^k(s^n_l) & = \frac{p(s^n_l)}{{m-n+k \choose k}}+\frac{N'_k}{N_k}\big(\big((1-d_k)p_{unif}+d_kp(s^n_l)\big)+\epsilon_b\big)\\
    & = p(s^n_l)\bigg(\frac{1}{{m-n+k \choose k}} + \frac{N'_k}{N_k}d_k\bigg) + \frac{N'_k}{N_k}\big((1-d_k)p_{unif}+\epsilon_b\big),
\end{split}
\end{equation*}
where $\epsilon_b$ is the bias error in the interference terms away from the linear model. From the positivity of the interference terms $$\epsilon_b \geq -\big((1-d_k)p_{unif}+d_kp(s^n_l)\big),$$ where $k\leq n-1$. 
One factor motivating the inclusion of a dependency term is that positive correlation of the interference terms with the ideal probabilities means that rather than the signal of the ideal probability having magnitude $\frac{1}{{m-n+k \choose k}}$, its magnitude is instead $\frac{1}{{m-n+k \choose k}} + \frac{N'_k}{N_k}d_k$. And this stronger signal may be used to improve mitigation performance.  
We now show how to compute the dependency term $d_k$ from the absolute average deviation. Then compute bounds on the bias error and statistical error for the use of linear dependency in linear solving, which are then used to provide performance guarantees relative to postselection.

\subsubsection{Computing the average dependency term $d_k$}

One consequence of a general correlation of interference terms with ideal probabilities within recycled probabilities is that ${D}_k^{\text{no dep.}}\leq{D}_k$. The average dependency term $d_k$ is a weighted average of dependency terms, it can be computed from ${D}_k^{\text{no dep.}}$ and ${D}_k$. We know by definition that
\begin{equation*}
\begin{split}
{D}_0 & = \frac{1}{{|\boldsymbol{S}|}}\sum_{s^n_l \in \boldsymbol{S}}\abs{ p(s^n_l)  - p_{unif}}.\\
\end{split}
\end{equation*}
For the moment ignoring statistical error, the dependency factor can be calculated from the average distance of probabilities from the uniform probability for the $n-k$--photon recycled distribution
\begin{equation*}
\begin{split}
{D}_k & = \frac{1}{{|\boldsymbol{S}|}}\sum_{s^n_l \in \boldsymbol{S}}\abs{ \frac{{n \choose k}p(s^n_l)}{N_k} + \frac{N_k'}{N_k}\big((1-d_k(s^n_l))p_{unif}+d_k(s^n_l)p(s_l^n)\big)  - p_{unif}}\\
& = \frac{1}{{|\boldsymbol{S}|}}\sum_{s^n_l \in \boldsymbol{S}}\abs{ \frac{{n \choose k}p(s^n_l)}{N_k} + \frac{N_k'}{N_k}p_{unif} - \frac{N_k'}{N_k}d_k(s^n_l)p_{unif}+\frac{N_k'}{N_k}d_k(s^n_l) p(s^n_l)  - p_{unif}}\\
& = \frac{1}{{|\boldsymbol{S}|}}\sum_{s^n_l \in \boldsymbol{S}}\abs{ \bigg(\frac{{n \choose k}}{N_k} \bigg) \bigg(p(s^n_l) - p_{unif} \bigg) + d_k(s^n_l)\frac{N_k'}{N_k}\bigg(p(s^n_l) - p_{unif}\bigg)}\\
& = \frac{1}{{|\boldsymbol{S}|}}\sum_{s^n_l \in \boldsymbol{S}}\abs{ \bigg(\frac{{n \choose k}}{N_k} + d_k(s^n_l)\frac{N_k'}{N_k}\bigg) \bigg(p(s^n_l) - p_{unif} \bigg) }\\
\end{split}
\end{equation*}
The average dependency term, $d_k$, for the recycled distribution is defined by the expression
\begin{equation*}
\begin{split}
 \bigg(\frac{{n \choose k}}{N_k} + d_k\frac{N_k'}{N_k}\bigg)\frac{1}{{|\boldsymbol{S}|}}\sum_{s^n_l \in \boldsymbol{S}}\abs{ \bigg(p(s^n_l) - p_{unif} \bigg) } & = \frac{1}{{|\boldsymbol{S}|}}\sum_{s^n_l \in \boldsymbol{S}}\abs{ \bigg(\frac{{n \choose k}}{N_k} + d_k\big(p(s^n_l)\big)\frac{N_k'}{N_k}\bigg) \bigg(p(s^n_l) - p_{unif} \bigg) }.\\
\end{split}
\end{equation*}
This can then be used to rewrite the previous expression in terms of the dependency
\begin{equation*}\label{eq:DkD0}
\begin{split}
{D}_k & = \frac{1}{{|\boldsymbol{S}|}}\sum_{s^n_l \in \boldsymbol{S}}\abs{ \bigg(\frac{{n \choose k}}{N_k} + d_k(s^n_l)\frac{N_k'}{N_k}\bigg) \bigg(p(s^n_l) - p_{unif} \bigg) }\\
& = \bigg(\frac{{n \choose k}}{N_k} + d_k\frac{N_k'}{N_k}\bigg)\frac{1}{{|\boldsymbol{S}|}}\sum_{s^n_l \in \boldsymbol{S}}\abs{ \bigg(p(s^n_l) - p_{unif} \bigg) }\\
& = \bigg(1 + d_k\bigg({{m-n+k} \choose k}-1\bigg)\bigg) \frac{1}{{{m-n+k} \choose k}}\frac{1}{{|\boldsymbol{S}|}}\sum_{s^n_l \in \boldsymbol{S}}\abs{ \bigg(p(s^n_l) - p_{unif} \bigg) }\\
& =\bigg( \frac{1 + d_k{{m-n+k} \choose k}-d_k}{{{m-n+k} \choose k}} \bigg) {D}_0 \\
& =\bigg( \frac{1}{{{{m-n+k} \choose k}}} + d_k\frac{{{m-n+k} \choose k}-1}{{{m-n+k} \choose k}} \bigg) {D}_0. \\
\end{split}
\end{equation*}
If there is no correlation of interference terms with ideal probabilities then $d_k=0$, and then
\begin{equation*}
\begin{split}
{D}_k^{d_k=0} & = \frac{1}{{{m-n+k} \choose k}}{D}_0.\\
\end{split}
\end{equation*}
Meaning that if there is no dependence the average absolute deviation decays with $k$ proportionally to $\frac{1}{{{m-n+k} \choose k}}$. Whereas with dependency the decay is proportional to $\bigg( \frac{1 + d_k{{m-n+k} \choose k}-d_k}{{{m-n+k} \choose k}} \bigg)$.
The average dependency term may be directly computed from the previous expression as 
\begin{equation*}
\label{eqexpdk}
\begin{split}
d_k = \frac{1}{{{m-n+k} \choose k} -1}\bigg(\frac{{{m-n+k} \choose k}{D}_k}{{D}_0} - \frac{1}{{{m-n+k} \choose k}}\bigg).\\
\end{split}
\end{equation*}

\subsection{Bias error bound with average dependency term $d_k$}

We now prove Lemma \ref{dependency_bias_error}, which will be used in the next section.

\begin{lemma} \label{dependency_bias_error}
The bias error from substituting $I_{s^n_l, k}$ with $ \big(1-d_k\big)p_{unif} + d_kp(s^n_l)$ in the recycled probabilities is upper bounded
\begin{equation*}
\text{Pr}\bigg[\big|I_{s^n_l, k} - \big( (1-d_k)p_{unif} + d_k p(s^n_l) \big)\big| \geq 2\epsilon_{bias,s^n_l} \bigg] \leq \frac{2 p_{unif}}{\epsilon_{bias,s^n_l}^2}.\\
\end{equation*}
\end{lemma}

\begin{proof}

Rather than bounding $I_{s^n_l, k}$ from $p_{unif}$, we would now like to bound deviation of $I_{s^n_l, k}$ from $ \big(1-d_k\big)p_{unif} + d_kp(s^n_l) $. 
As it is required that $1 \geq d_k \geq 0$, terms of the form $\big( \big(1-d_k\big)p_{unif} + d_kp(s^n_l)\big)$ are bounded 
$$
|p(s^n_l) - \big( (1-d_k)p_{unif} + d_kp(s^n_l)\big)| \geq |\big( (1-d_k)p_{unif} + d_kp(s^n_l)\big) - p_{unif}|,
$$
or $$
|p(s^n_l) - \big( (1-d_k)p_{unif} + d_kp(s^n_l)\big)| \leq  |\big( (1-d_k)p_{unif} + d_kp(s^n_l)\big) - p_{unif}|,$$ 
with both statements true only with equidistance.
The mean of the set $\{p(s^n_l)\}_l$ is $p_{unif}$, and so also the mean of the set $\{ \big(1-d_k\big)p_{unif} + d_kp(s^n_l)\}_l$ is $p_{unif}$. If $d_k = 0$ then $\mathsf{Var}(\{ \big(1-d_k\big)p_{unif} + d_kp(s^n_l)\}_l)=0$, while if $d_k = 1$ then $\mathsf{Var}(\{ \big(1-d_k\big)p_{unif} + d_kp(s^n_l)\}_l)=\mathsf{Var}(\{p(s^n_l)\}_l)$. As $d_k\leq 1$ the variance of the set $\{ \big(1-d_k\big)p_{unif} + d_kp(s^n_l)\}_l$ is upper bounded by the variance of $\{p(s^n_l)\}_l$, so $\mathsf{Var}\big( (1-d_k)p_{unif} + d_kp(s^n_l)\big)\leq p_{unif}$, where we have used the Bhatia-Davis inequality. 
We can then bound the distance of $\big( (1-d_k)p_{unif} + d_k p(s^n_l) \big)$ from $p_{unif}$ using Chebyshev's inequality so that
\begin{equation*}
\begin{split}
\text{Pr}\bigg[\big|\big( (1-d_k)p_{unif} + d_k p(s^n_l) - p_{unif}\big)\big| \geq \epsilon_{bias,s^n_l} \bigg] & \leq \frac{  p_{unif}}{\epsilon_{bias,s^n_l}^2}  .\\
\end{split}
\end{equation*}
And in reverse form this inequality is
\begin{equation} \label{reveqn1}
\begin{split}
\text{Pr}\bigg[\big|\big( (1-d_k)p_{unif} + d_k p(s^n_l) - p_{unif}\big)\big| < \epsilon_{bias,s^n_l} \bigg] & \geq 1- \frac{  p_{unif}}{\epsilon_{bias,s^n_l}^2}  .\\
\end{split}
\end{equation}
The Thm.  \ref{nonHaarbound} bound on the bias of the interference term for arbitrary matrices is
\begin{equation*}
\begin{split}
\text{Pr}\bigg[\big|I_{s^n_l, k} - p_{unif}\big| \geq \epsilon_{bias,s^n_l} \bigg] & \leq \frac{ p_{unif}}{\epsilon_{bias,s^n_l}^2}.\\
\end{split}
\end{equation*}
Which in reverse form is
\begin{equation} \label{reveqn2}
\begin{split}
\text{Pr}\bigg[\big|I_{s^n_l, k} - p_{unif}\big| < \epsilon_{bias,s^n_l} \bigg] & \geq 1-\frac{ p_{unif}}{\epsilon_{bias,s^n_l}^2}.\\
\end{split}
\end{equation}
Now, we use triangle inequality to combine the reverse forms of the above two inequalities, so that 
\begin{equation} \label{triangle_ineq1}
\begin{split}
|I_{s^n_l, k} - \big( (1-d_k)p_{unif} + d_k p(s^n_l) )| &\leq |I_{s^n_l, k} - p_{unif}|+|-\big( (1-d_k)p_{unif} + d_k p(s^n_l)) + p_{unif}|\\
&\leq 2\epsilon_{bias,s^n_l}. \\
\end{split}
\end{equation}
Assuming independence of the inequalities in eqn. \ref{reveqn1} and eqn. \ref{reveqn2}, the the inequality in eqn. \ref{triangle_ineq1} holds with confidence $ (1-\frac{ p_{unif}}{\epsilon_{bias,s^n_l}^2})^2$. Using a Bernoulli approximation $(1-\frac{ p_{unif}}{\epsilon_{bias,s^n_l}^2})^2 \approx 1-2\frac{p_{unif}}{\epsilon_{bias,s^n_l}^2}$, and the result follows. 

\end{proof}

\subsection{Performance guarantee inequality for linear solving with dependency}
\label{appperfguarlinsolvdep}

We now derive the condition for when recycling mitigation with linear solving with dependency can outperform postselection.

\vspace{1em}
\begin{theorem} \label{linsoldep}
\textit{The condition}
\begin{equation*}
         \sqrt{\frac{{m \choose n}}{{n \choose k}(1-\eta)^{n-k}\eta^kN_{tot}}} + C\bigg({m-n+k \choose k}-1\bigg)\cdot\epsilon_{bias,s^n_l}  \leq \sqrt{\frac{{m \choose n}}{(1-\eta)^n N_{tot}}},
\end{equation*}
\textit{with $C>0$ defines a sampling regime where the sum of the worst-case statistical error and bias error of linear solving with dependency recycling mitigation is less than the worst-case statistical error of postselection.}
\end{theorem}
\begin{proof}

The estimator of the recycled probability may be rewritten in terms of the dependency term $d_k$ as
\begin{equation*}
\begin{split}
    p_R^k(s^n_l) & = \frac{p(s^n_l)}{{m-n+k \choose k}}+\frac{N'_k}{N_k}\big(\big((1-d_k -\epsilon_{\text{hoeff},d_k})p_{unif}+(d_k + \epsilon_{\text{hoeff},d_k})p(s^n_l)\big)+\epsilon_b\big) + \epsilon_{\text{hoeff},p_R^k}\\
    & = p(s^n_l)\bigg(\frac{1}{{m-n+k \choose k}} + \frac{N'_k}{N_k}(d_k + \epsilon_{\text{hoeff},d_k})\bigg) + \frac{N'_k}{N_k}\big((1-d_k-\epsilon_{\text{hoeff},d_k})p_{unif}+\epsilon_{bias,s^n_l}\big) + \epsilon_{\text{hoeff},p_R^k}.
\end{split}
\end{equation*}
Solving to find the mitigated value in the ideal case one obtains
\begin{equation*}
\begin{split}
    p_{mit,dep}(s^n_l) 
    & = \Bigg|\frac{p_R^k(s^n_l) +  \frac{N'_k}{N_k}\big((-1 + d_k )p_{unif} \big) }{\frac{1}{{m-n+k \choose k}} + \frac{N'_k}{N_k}d_k }\Bigg|.\\
\end{split}
\end{equation*}
However, the estimator of the mitigated value including error is
\begin{equation*}
\begin{split}
    \tilde{p}_{mit,dep}(s^n_l) & = \frac{p_R^k(s^n_l) -  \frac{N'_k}{N_k}\big((1-(d_k+\epsilon_{\text{hoeff},d_k}))p_{unif}+\epsilon_{bias,s^n_l}\big) + \epsilon_{\text{hoeff},p_R^k}}{\frac{1}{{m-n+k \choose k}} + \frac{N'_k}{N_k}(d_k+\epsilon_{\text{hoeff},d_k})}.
\end{split}
\end{equation*}
As, by definition, $1 \geq d_k \geq 0$, this means that $-d_k \leq \epsilon_{\text{hoeff},d_k} \leq 1 $. And the value of $\epsilon_{\text{hoeff},d_k}$ that maximally increases the error of the above quotient is then 
$\epsilon_{\text{hoeff},d_k}=-d_k$, this substitution removes the $d_k$ terms. From this point we can use Lemma \ref{miti_bias_error} and Lemma \ref{miti_stat_error}, with the one difference in the latter that the bias error $\epsilon_{bias,s^n_l}$ from Lemma \ref{dependency_bias_error} is used instead of $\epsilon_{bias,s^n_l}$. And then the total error of the mitigated probability, $\epsilon_{p_{mit,dep}(s^n_l)}$, can be upper bounded
\begin{equation*}
\begin{split}
\abs{ \epsilon_{p_{mit,dep}(s^n_l)} } & \leq 3{m-n+k \choose k} \bigg( \abs{\epsilon_{\text{hoeff},p_R^k}} + \frac{N'_k}{N_k}\abs{\epsilon_{bias,s^n_l}}\bigg) .\\ 
\end{split}
\end{equation*}
Now we use that
\begin{equation*}
|\epsilon_{\text{hoeff},p_R^k}| \leq D \frac{1}{{m-n+k \choose k}}\sqrt{\frac{{m \choose n}}{{n \choose k}(1-\eta)^{n-k}\eta^kN_{tot}}}, 
\end{equation*}
for some $D>1$ large enough so the inequality holds with high confidence. So that the overall error for linear solving with dependency is upper bounded 
\begin{equation*}
\begin{split}
\abs{ \epsilon_{p_{mit,dep}(s^n_l)} } & \leq  {m-n+k \choose k} \bigg(  D\frac{1}{{m-n+k \choose k}}\sqrt{\frac{{m \choose n}}{{n \choose k}(1-\eta)^{n-k}\eta^kN_{tot}}} + \frac{2N'_k}{N_k}\cdot \epsilon_{bias,s^n_l}\bigg) .\\ 
\end{split}
\end{equation*}
And so the condition to beat postselection is
 \begin{equation*}
        {m-n+k \choose k}\bigg(  \frac{1}{{m-n+k \choose k}}\sqrt{\frac{{m \choose n}}{{n \choose k}(1-\eta)^{n-k}\eta^kN_{tot}}} + \frac{2}{D}\frac{N'_k}{N_k}\cdot \epsilon_{bias,s^n_l} \bigg) \leq \sqrt{\frac{{m \choose n}}{(1-\eta)^n N_{tot}}},
    \end{equation*}
that is rearranged for the result, where we have the same constant $D$ for Hoeffding's inequality for postselection.
\end{proof}

\section{Extrapolation}
\label{appextrap}

\subsection{Statistical error for average absolute deviation estimator $\tilde{D}_k$}

The following lemma upper bounds the statistical error of the empirically computed absolute average deviation terms $\{\tilde{D}_k\}_{k=1}^K$, and is used in deriving the condition for linear extrapolation to outperform postselection in the next section. 
\begin{lemma} \label{abs_avg_dev_stat_error}
The statistical error of the average absolute deviation estimator $\tilde{D}_k$ is upper bounded
\begin{equation}
\abs{ \epsilon_{\text{hoeff},\tilde{D}_k} }  \in O\Bigg(\Bigg(\frac{1}{{m-n+k \choose k}^2{n \choose k}(1-\eta)^{n-k}\eta^kN_{tot}}\Bigg)^{1/4}\Bigg).
\end{equation}
\end{lemma}

\begin{proof}
 Let  $X$ be a random variable defined over the uniform distribution of bit strings $s^n_l$ that returns terms from the set $\{\tilde{p}_{R}^k(s^n_l)\}_l$ as values, $\mu:=\mathsf{E}_{s^n_l}(X)=p_{unif}=\frac{1}{{m \choose n}}$, with $\mathsf{E}_{s^n_l}(.)$ denoting the expectation value of $X$ over the uniform distribution of $s^n_l$, it is equal to $p_{unif}$. The average absolute deviation estimator, denoted $\{\tilde{D}_k\}_{k=0}^n$, is defined  
\begin{equation*}
\tilde{D}_k := {m \choose n}^{-1}\sum_{s^n_l \in \boldsymbol{S}}\abs{\tilde{p}_{R}^k(s^n_l) - p_{unif}}={D}_k + \epsilon_{\text{hoeff},\tilde{D}_k},
\end{equation*}
where  $\epsilon_{\text{hoeff},\tilde{D}_k}$ is the absolute average deviation statistical error. Let $M:=\mathsf{max}_{s^n_l}(X)$ , $m:=\mathsf{min}_{s^n_l}(X)$ , and note that $M \leq 1$ and $m \geq 0$. We can apply the Bhatia-Davis inequality \cite{bhatia_better_2000}
    \begin{equation*} 
    \mathsf{Var}(\{X\})  \leq (M-\mu)(\mu-m),
    \end{equation*}
with an upper bound of $M=1$, a lower bound of $m=0$, and the expected value $\mu=p_{unif}$, to obtain an upper bound on the variance of the random variable $X$ so that
    \begin{equation} \label{BD_for_recycled_estimators}
    \begin{split}
        \mathsf{Var}(X) &\leq (1-p_{unif})(p_{unif})\\
        &\leq p_{unif}.
        \end{split}
    \end{equation}
Now, by Jensen's inequality \cite{jensen_sur_1906} we know that $\mathsf{E}_{s^n_l}(|X-\mu|^2) \geq (\mathsf{E}_{s^n_l}|X-\mu|)^2$, therefore
    \begin{equation*}
        \mathsf{Var}(X) \geq \bigg({m \choose n}^{-1}\sum_{s^n_l \in \boldsymbol{S}}\abs{\tilde{p}_{R}^k(s^n_l) - p_{unif}}\bigg)^2,
    \end{equation*}
and as both sides of this inequality are positive and the square root operation is a monotonically increasing function for positive reals
 \begin{equation*}
        \mathsf{Var}(X)^{1/2} \geq {m \choose n}^{-1}\sum_{s^n_l \in \boldsymbol{S}}\abs{\tilde{p}_{R}^k(s^n_l) - p_{unif}}.
    \end{equation*}
And using the result from eqn. \ref{BD_for_recycled_estimators} we have
 \begin{equation*}
 \begin{split}
p_{unif}^{1/2} &\geq {m \choose n}^{-1}\sum_{s^n_l \in \boldsymbol{S}}\abs{\tilde{p}_{R}^k(s^n_l) - p_{unif}}\\
& \geq {D}_k + \epsilon_{\text{hoeff},\tilde{D}_k}.
\end{split}
\end{equation*}
Now because $\tilde{D}_k \geq 0$ 
then $ \epsilon_{\text{hoeff},\tilde{D}_k}\geq - D_k \geq - p_{unif}^{1/2}$. 
Also $p_{unif}^{1/2} \geq D_k$ and $p_{unif}^{1/2} \geq D_k + \epsilon_{\text{hoeff},\tilde{D}_k}$ mean that $2(p_{unif}^{1/2}  - {D}_k)\geq \epsilon_{\text{hoeff},\tilde{D}_k}$. 
And as $ {D}_k \geq 0$ this means $2p_{unif}^{1/2} \geq \epsilon_{\text{hoeff},\tilde{D}_k}$. 
Then we have that $2p_{unif}^{1/2} \geq \epsilon_{\text{hoeff},\tilde{D}_k}\geq - p_{unif}^{1/2}$ which gives 
\begin{equation} \label{first_bound_inequality_D_k}
|\epsilon_{\text{hoeff},\tilde{D}_k}| \leq 2p_{unif}^{1/2}.
\end{equation}
The absolute error for $\tilde{D}_k$ can also be upper bounded in terms of the statistical error of the recycled probabilities
\begin{equation} \label{prev_abs_error_bound}
\begin{split}
|\epsilon_{\text{hoeff},\tilde{D}_k}| &=  |\tilde{D}_k - {D}_k| \\
 &=  \bigg|{m \choose n}^{-1}\bigg(\sum_{s^n_l \in \boldsymbol{S}}\abs{p_{R}^k(s^n_l) + \epsilon_{\text{hoeff},\tilde{p}_R^k(s^{n}_l)}  - p_{unif}}  -  \sum_{s^n_l \in \boldsymbol{S}}\abs{p_{R}^k(s^n_l) - p_{unif}}\bigg)\bigg| 
 \\
 & \leq {m \choose n}^{-1} \bigg(\sum_{s^n_l \in \boldsymbol{S}} \bigg|\abs{p_{R}^k(s^n_l) + \epsilon_{\text{hoeff},\tilde{p}_R^k(s^{n}_l)}  - p_{unif}}-\abs{p_{R}^k(s^n_l) - p_{unif}}\bigg| \bigg) \\ &\leq {m \choose n}^{-1}\sum_{s^n_l \in \boldsymbol{S}}\big| \epsilon_{\text{hoeff},\tilde{p}_R^k(s^{n}_l)}  \big|.
 \end{split}
\end{equation}

Where the inequalities follow by using a triangle, then a reverse triangle inequality. 
All terms in the inequalities from eqns. \ref{first_bound_inequality_D_k} and \ref{prev_abs_error_bound} are positive reals, and so the inequalities may be combined to give
\begin{equation*}
\begin{split}
|\epsilon_{\text{hoeff},\tilde{D}_k}|^2 &\leq  2p_{unif}^{1/2} {m \choose n}^{-1}\sum_{s^n_l \in \boldsymbol{S}}\big| \epsilon_{\text{hoeff},\tilde{p}_R^k(s^{n}_l)}  \big|.\\
\end{split}
\end{equation*}
We now define $\mathcal{E}_{\text{hoeff},\tilde{p}_R^k}$ as the magnitude of the upper bound for the statistical error $|\epsilon_{\text{hoeff},\tilde{p}_R^k(s^{n}_l)}|$ which holds with high confidence from Hoeffdings inequality.
\begin{equation*}
\mathcal{E}_{\text{hoeff},\tilde{p}_R^k} \in O\bigg(\frac{1}{{m-n+k \choose k}}\sqrt{\frac{{m \choose n}}{{n \choose k}(1-\eta)^{n-k}\eta^kN_{tot}}}\bigg).
\end{equation*}
As  $|\epsilon_{\text{hoeff},\tilde{p}_R^k(s^{n}_l)}| \leq \mathcal{E}_{\text{hoeff},\tilde{p}_R^k}$  we can now write
\begin{equation*}
\begin{split}
|\epsilon_{\text{hoeff},\tilde{D}_k}|^2 &\leq 2p_{unif}^{1/2} \mathcal{E}_{\text{hoeff},\tilde{p}_R^k} ,\\
\end{split}
\end{equation*}
and, again using that the square root operation is a monotonically increasing function for positive reals, it follows that
\begin{equation*}
\begin{split}
|\epsilon_{\text{hoeff},\tilde{D}_k}| &\leq (2 \mathcal{E}_{\text{hoeff},\tilde{p}_R^k})^{1/2} p_{unif}^{1/4} \\
&\in O\big(\mathcal{E}_{\text{hoeff},\tilde{p}_R^k}^{1/2} p_{unif}^{1/4}\big) \\
&\in O\bigg(\bigg(\frac{1}{{m-n+k \choose k}}\bigg)^{1/2}\bigg(\frac{{m \choose n}}{{n \choose k}(1-\eta)^{n-k}\eta^kN_{tot}}\bigg)^{1/4}\bigg(\frac{1}{{m \choose n}}\bigg)^{1/4} \bigg)\\
&\in O\Bigg(\Bigg(\frac{1}{{m-n+k \choose k}^2{n \choose k}(1-\eta)^{n-k}\eta^kN_{tot}}\Bigg)^{1/4}\Bigg).
\end{split}
\end{equation*}
Which was the result to be proved.
\end{proof}

\subsection{Performance guarantee inequality for extrapolation with linear least squares \label{extrapolation_perf_g_app}}

Before giving a proof of Thm.  \ref{extraplin} we first state and prove two results that we will use. These are an upper bound for the statistical error for the gradient parameter. And an upper bound on the bias error introduce by using the average gradient parameter rather than the ideal gradient parameter. These we state as Lemmas \ref{gradient_stat_error} and \ref{gradient_bias_error}.

\begin{lemma} \label{gradient_stat_error}
The average gradient parameter statistical error $\epsilon_{hoeff,g}$ is such that 
\begin{equation}
\big|\epsilon_{\text{hoeff},g}\big| \in O\Bigg( \bigg(\frac{1}{(1-\eta)^{n}N_{tot}}\bigg)^{1/4}\Bigg).\\
\end{equation}
\end{lemma}

\begin{proof}
The bias error is $|g_{ideal - g_{avg}}|$.
First note that using the linear least squares method the average gradient parameter estimator may be written
\begin{equation*}
\tilde{g}_{avg} = \frac{3}{2n_d+1} \tilde{D}_0  - \frac{6}{n_d(n_d+1)(2n_d+1)} \sum_{i=1}^{n_d} \tilde{D}_ix_i,
\end{equation*}
and the average gradient parameter without statistical error 
\begin{equation*}
{g}_{avg} =  \frac{3}{2n_d+1} {D}_0  - \frac{6}{n_d(n_d+1)(2n_d+1)} \sum_{i=1}^{n_d} {D}_ix_i.
\end{equation*}

The statistical error for the gradient term is then
\begin{equation*}
\begin{split}
 \big|\epsilon_{\text{hoeff},g}\big| &= |\tilde{g}_{avg} -  {g}_{avg}| \\
 &= \big|\frac{3}{2n_d+1} \tilde{D}_0  - \frac{6}{n_d(n_d+1)(2n_d+1)} \sum_{i=1}^{n_d} \tilde{D}_ix_i -  \frac{3}{2n_d+1} {D}_0  + \frac{6}{n_d(n_d+1)(2n_d+1)} \sum_{i=1}^{n_d} {D}_ix_i \big|\\
 &= \big|\frac{3}{2n_d+1} ({D}_0 + \epsilon_{\text{hoeff},\tilde{D}_0})  - \frac{6}{n_d(n_d+1)(2n_d+1)} \sum_{i=1}^{n_d} ({D}_i+\epsilon_{\text{hoeff},\tilde{D}_i})x_i -  \frac{3}{2n_d+1} {D}_0  \\
 &\hspace{3em}+ \frac{6}{n_d(n_d+1)(2n_d+1)} \sum_{i=1}^{n_d} {D}_ix_i \big|\\
 &= \big|   \frac{3}{2n_d+1}  \epsilon_{\text{hoeff},\tilde{D}_0} + \frac{6}{n_d(n_d+1)(2n_d+1)}\sum_{i=1}^{n_d} \epsilon_{\text{hoeff},\tilde{D}_i} x_i\big|\\
&\leq    \frac{3}{2n_d+1}  \big| \epsilon_{\text{hoeff},\tilde{D}_0}\big| + \frac{6}{n_d(n_d+1)(2n_d+1)}\sum_{i=1}^{n_d}\big|\epsilon_{\text{hoeff},\tilde{D}_i} \big|x_i\\
&\leq    \frac{3}{2n_d+1}  A \bigg(\frac{1}{(1-\eta)^{n}N_{tot}}\bigg)^{1/4} + B\frac{6}{n_d(n_d+1)(2n_d+1)}\sum_{i=1}^{n_d} \bigg(\frac{1}{(1-\eta)^{n}N_{tot}}\bigg)^{1/4}x_i\\
& \in  O\Bigg( \bigg(\frac{1}{(1-\eta)^{n}N_{tot}}\bigg)^{1/4}\Bigg),\\
\end{split}
\end{equation*}
for $A,B >0$. Where a triangle inequality was used for the fourth line, and that $\big|\epsilon_{\text{hoeff},\tilde{D}_i} =B \big
(\frac{1}{(1-\eta)^{n}N_{tot}}\big)^{1/4}$ for $i\in \{0,\ldots,n_d\}$ from lemma \ref{abs_avg_dev_stat_error} was used to get the sixth line. 

\end{proof}

\begin{lemma} \label{gradient_bias_error}
The bias error $\epsilon_{\text{bias},g}$ from substituting the ideal gradient parameter, $g_{ideal}$, for the average gradient parameter, $g_{avg}$, is upper bounded
\begin{equation}
\big|\epsilon_{\text{bias},g}\big| \in O\bigg({m \choose n}^{-1/3}\bigg),
\end{equation}
 with confidence $1-O(m^{-n/3})$.
\end{lemma}

\begin{proof}
The absolute bias error from the use of the average gradient rather than the ideal gradient in the second iteration of linear least squares is $\epsilon_{\text{bias},g} = |g_{ideal} - g_{avg}|.$
The ideal $\alpha_{ideal,s^n_l}$ term that will result in the ideal probability being produced by the least squares method may be written
\begin{equation} \label{idealvalue}
\begin{split}
\alpha_{ideal,s^n_l} & = \frac{n_d+1}{2}g_{{ideal}} + \frac{1}{n_d}\sum_{i=1}^{n_d}( {p}^k_R(s^n_l) -p_{unif}),\\
\end{split}
\end{equation}
and because we are considering the ideal case it is true that $p(s^n_l) = \alpha_{ideal,s^n_l} + p_{unif}$. From $\alpha_{ideal,s^n_l} \leq p(s^n_l)$, and that
\begin{equation}
\begin{split}
\alpha_{ideal,s^n_l} & \geq \frac{n_d+1}{2}g_{{ideal}} + \frac{1}{n_d}\sum_{i=1}^{n_d}( -p_{unif}),\\
\end{split}
\end{equation}
it follows that $g_{ideal} \leq \alpha_{ideal,s^n_l} + p_{unif} \leq p(s^n_l) + p_{unif} $.
Now using the Bhatia-Davis inequality in a similar way as for the proof of Theorem \ref{nonHaarbound} in Section \ref{proof_nonHaarbound}, we obtain the statistical inequality
\begin{equation*}
\begin{split}
Pr\Big(\big|{p}(s^n_l)  - {m \choose n}^{-1}\sum_{s^n_l \in \boldsymbol{S}} {p}(s^n_l) \big|\leq \epsilon_{bias,s^n_l}\Big) & \geq 1-\frac{p_{unif}}{\epsilon^2_{bias,s^n_l}}.
\end{split}
\end{equation*}
This bound on the distance of ${p}(s^n_l)$ from $p_{unif}$ can be used to upper bound $g_{ideal}$, as $g_{ideal} \leq p(s^n_l) + p_{unif} $ so then
\begin{equation*}
\begin{split}
Pr\Big(g_{ideal} \leq  2p_{unif} +\epsilon_{bias,s^n_l}\Big) & \geq 1-\frac{p_{unif}}{\epsilon^2_{bias,s^n_l}}.
\end{split}
\end{equation*}
Setting $\epsilon_{bias,s^n_l}={m \choose n}^{-1/3}$ in the previous statistical inequality, it becomes
\begin{equation*}
\begin{split}
Pr\Big(g_{ideal} \leq  2p_{unif} +{m \choose n}^{-1/3}\Big) & \geq 1-{m \choose n}^{-1/3}.
\end{split}
\end{equation*}
Or in other words $g_{ideal} \in O({m \choose n}^{-1/3})$ with exponentially high confidence.
Now using Jensen's inequality,
\begin{equation*}
\begin{split}
{D}_0 & = \frac{1}{{|\boldsymbol{S}|}}\sum_{s^n_l \in \boldsymbol{S}}\abs{ p(s^n_l)  - p_{unif}} \\
&\leq \bigg(\frac{1}{{|\boldsymbol{S}|}}\sum_{s^n_l \in \boldsymbol{S}}\abs{ p(s^n_l)  - p_{unif}}^2\bigg)^{1/2} \\
&\leq ((1 - p_{unif})p_{unif})^{1/2},\\
\end{split}
\end{equation*}
where the final line follows from applying the Bhatia-Davis inequality.
From eqn.  \ref{average_gradient_expression} we have that ${g}_{\text{avg}} \leq D_0$, and so by the previous inequality it follows that ${g}_{\text{avg}} \leq {m \choose n}^{-1/2}$.
And so with exponentially high confidence $g_{ideal} \in O({m \choose n}^{-1/3})$, and ${g}_{\text{avg}} \leq {m \choose n}^{-1/2}$. 

The upper bound on the gradient bias error is then
\begin{equation*}
\big|g_{ideal}  - {g}_{\text{avg}}\big| \in O\bigg({m \choose n}^{-1/3}\bigg),
\end{equation*}
with confidence $1-O(m^{-n/3})$.
\end{proof}

We now derive the condition for when recycling mitigation with linear extrapolation can outperform postselection.

\vspace{1em}
\begin{theorem} \label{extraplin}
\textit{The condition:}
\begin{equation*}
\begin{split}
A\frac{n_d+1}{2}\bigg( \bigg(\frac{1}{(1-\eta)^{n}N_{tot}}\bigg)^{1/4} +  {m \choose n}^{-1/3} \bigg) + B \frac{n}{m-n+1}\sqrt{\frac{{m \choose n}}{n(1-\eta)^{n-1}\eta N_{tot}}} \leq C\sqrt{\frac{{m \choose n}}{(1-\eta)^n N_{tot}}}, 
\end{split}
\end{equation*}
\textit{with $A,B,C > 0$, defines a sampling regime where the sum of the worst-case statistical error and bias error of linear extrapolation recycling mitigation using the least squares method is less than the worst-case statistical error of postselection.}
\end{theorem}

\begin{proof}

The linear extrapolation method consists of two iterations of linear least squares, 
with $n_d \in \{n_d \in \mathbb{Z}^+ | n_d<n\}$ data points used in both iterations. 
In the first iteration of least squares the data set $\{k,\tilde{D}_k\}_{k=1}^{n_d}$ is used to compute the gradient parameter $g_{\text{avg}}$. 
Where $\tilde{D}_k$ is the estimated absolute average deviation for the $n-k$--photon recycled distribution from statistics where $k$ photons were lost. 
In the second iteration of least squares the data set $\{k,\tilde{p}^k_R(s^n_l)-p_{unif}\}_{k=1}^{n_d}$ is used to fit a linear model function that depends on $g_{\text{avg}}$, in order to generate the mitigated output for each $s^n_l$. 

In the first iteration of least squares the linear model function used is
\begin{equation*}
f(x_i,g_{\text{avg}}) = -\tilde{g}_{\text{avg}}x_i + \tilde{D}_0.
\end{equation*}
For which the residuals are of the form
\begin{equation*}
\begin{split}
r_i & = y_i - ( -\tilde{g}_{\text{avg}}x_i + \tilde{D}_0 ), \\
\end{split}
\end{equation*}
The sum of the squared residuals is a function of $\alpha_{s_n}$, and may be written
\begin{equation*}
\begin{split}
S(\tilde{g}_{\text{avg}}) & = \sum_{i=1}^{n_d} r_i ^2 \\
& = \sum_{i=1}^{n_d} (y_i - ( -\tilde{g}_{\text{avg}}x_i + \tilde{D}_0 )) ^2 \\
\end{split}
\end{equation*}
The optimal solution may be found by taking the derivative of S with respect to $\tilde{g}_{\text{avg}}$
\begin{equation*} 
\begin{split}
0 = \frac{\text{d}S}{\text{d}\tilde{g}_{\text{avg}}} & = 2 \sum_{i=1}^{n_d} (y_i - ( -\tilde{g}_{\text{avg}}x_i + \tilde{D}_0 ))(x_i), \\
\end{split}
\end{equation*}
and then solving
\begin{equation*} 
\begin{split}
 \sum_{i=1}^{n_d} \tilde{g}_{\text{avg}}x_i^2   = \sum_{i=1}^{n_d} \tilde{D}_0 x_i - \sum_{i=1}^{n_d} y_ix_i.\\
\end{split}
\end{equation*}
Simplifying and substituting $\tilde{D}_i$ terms, this gives the optimal average gradient parameter as
\begin{equation} \label{average_gradient_expression}
\begin{split}
 \tilde{g}_{\text{avg}}   = \frac{3}{2n_d+2} \tilde{D}_0  - \frac{6}{n_d(n_d+1)(2n_d+1)} \sum_{i=1}^{n_d} \tilde{D}_ix_i.\\
\end{split}
\end{equation}
Using lemma \ref{gradient_stat_error}, the statistical error of the average gradient parameter estimator, $\epsilon_{\text{hoeff},g}$, can be upper bounded
\begin{equation*}
\begin{split}
\big|\epsilon_{\text{hoeff},g}\big| &\leq |\tilde{g}_{\text{avg}} - {g}_{\text{avg}}| \\
&\in O\Bigg( \bigg(\frac{1}{(1-\eta)^{n}N_{tot}}\bigg)^{1/4}\Bigg).\\
\end{split}
\end{equation*}

The average gradient parameter estimator, $\tilde{g}_{\text{avg}}$, is used in a second iteration of least squares to generate the mitigated outputs. The data set used to generate the mitigated output for each $s^n_l$ is defined as $\{x_i,y_i\}_{i=1}^{n_d}:=\{k,\tilde{p}^k_R(s^n_l)-p_{unif}\}_{k=1}^{n_d}$. 
Each output string is assigned an individual model function of the form 
\begin{equation*}
f_{s^n_l}(x_i,\alpha_{s^n_l})= \text{sgn}(-y_1)\tilde{g}_{\text{avg}}x_i + \alpha_{s^n_l}.
\end{equation*}
This function models the decay of the terms $\{\tilde{p}^k_R(s^n_l)-p_{unif}\}_{k=1}^{n_d}$ from the data set towards zero with increasing $k$.
If $p^k_R(s^n_l)<p_{unif}$ then the gradient term (namely $\text{sgn}(-y_1)$) will be positive, and if $p^k_R(s^n_l)>p_{unif}$ the gradient will be negative.
In the following analysis we will assume the sign, and therefore the gradient, is negative ($ \text{sgn}(-y_1)=-1$) so that the model function is
\begin{equation*}
f_{s^n_l}(x_i,\alpha_{s^n_l})= -\tilde{g}_{\text{avg}}x_i + \alpha_{s^n_l},
\end{equation*}
noting that the analysis is the same for the case of positive sign. 
The optimal prefactor variable $\alpha_{s^n_l}$ to fit the model function to the data set is computed using the residuals 
\begin{equation*}
\begin{split}
r_i & = y_i - ( -\tilde{g}_{\text{avg}}x_i + \alpha_{s^n_l}), \\
\end{split}
\end{equation*}
The sum of the squared residuals is a function of $\alpha_{s^n_l}$, and may be written
\begin{equation*}
\begin{split}
S(\alpha_{s^n_l}) & = \sum_{i=1}^{n_d} r_i ^2\\
& = \sum_{i=1}^{n_d} (y_i - ( -\tilde{g}_{\text{avg}}x_i + \alpha_{s^n_l})) ^2 \\
\end{split}
\end{equation*}
The value of $\alpha_{s^n_l}$ that minimises $S$ may be found by taking the derivative of S with respect to $\alpha_{s^n_l}$, and solving 
\begin{equation*}
\begin{split}
0 = \frac{\text{d}S}{\text{d}\alpha_{s^n_l}} & =  -2\sum_{i=1}^{n_d} (y_i - ( -g_{\text{avg}}x_i + \alpha_{s^n_l}))
\end{split}
\end{equation*}
Which gives the optimal value as 
\begin{equation} \label{optimalvalue}
\begin{split}
\alpha_{s^n_l} & = \frac{n_d+1}{2}\tilde{g}_{\text{avg}} + \frac{1}{n_d}\sum_{i=1}^{n_d} y_i .\\
\end{split}
\end{equation}
And substituting in the data values this is then
\begin{equation}
\begin{split}
\alpha_{s^n_l} & = \frac{n_d+1}{2}\tilde{g}_{\text{avg}} + \frac{1}{n_d}\sum_{k=1}^{n_d} (\tilde{p}^k_R(s^n_l)-p_{unif}) .\\
\end{split}
\end{equation}
Now $\tilde{g}_{\text{avg}}$ can be decomposed into an ideal gradient parameter, $g_{ideal}$, the statistical error for the estimator, $\epsilon_{\text{hoeff},g}$, and the bias error from substituting the ideal gradient parameter for the average gradient parameter, $\epsilon_{\text{bias},g}$. 
The ideal gradient parameter, $g_{ideal}$, defines the gradient that results in the function outputting the ideal probability at $x_i=0$ and $\alpha_{s^n_l} = p(s^n_l)$. 
And the recycled probability estimators $\tilde{p}^k_R(s^n_l)$ can be decomposed into recycled probabilities, ${p}^k_R(s^n_l)$, and their associated statistical errors, $\epsilon_{\text{hoeff},\tilde{p}_R^k(s^{n}_l)}$. So that
\begin{equation} 
\begin{split}
\alpha_{s^n_l} 
& = \frac{n_d+1}{2}\big(g_{{ideal}} + \epsilon_{\text{hoeff},g} + \epsilon_{\text{bias},g}\big) + \frac{1}{n_d}\sum_{i=1}^{n_d}({p}^k_R(s^n_l) + \epsilon_{\text{hoeff},\tilde{p}_R^k(s^{n}_l)}-p_{unif}) .\\
\end{split}
\end{equation}
While the ideal prefactor variable can be stated as
\begin{equation} \label{idealvalue}
\begin{split}
\alpha_{ideal,s^n_l} & = \frac{n_d+1}{2}g_{{ideal}} + \frac{1}{n_d}\sum_{i=1}^{n_d}({p}^k_R(s^n_l) -p_{unif}).\\
\end{split}
\end{equation}
And the mitigated output is $p_{mit}(s^n_l) = p_{unif} + \alpha_{s^n_l}$.
The magnitude of the error for the mitigated output is then
\begin{equation} 
\begin{split}
| p_{mit}(s^n_l) - p(s^n_l)| &= |p_{unif}+\alpha_{s^n_l} - (p_{unif}+\alpha_{ideal,s^n_l})|\\
 & = |\alpha_{s^n_l} - \alpha_{ideal,s^n_l}|\\
& \leq \frac{n_d+1}{2}( |\epsilon_{\text{hoeff},g}| + |\epsilon_{\text{bias},g}|) + |\epsilon_{\text{hoeff},\tilde{p}_R^1(s^{n}_l)}|\\
& \leq A \frac{n_d+1}{2}\bigg( \bigg(\frac{1}{(1-\eta)^{n}N_{tot}}\bigg)^{1/4} +   {m \choose n}^{-1/3}\bigg) + B\frac{n}{m-n+1}\sqrt{\frac{{m \choose n}}{n(1-\eta)^{n-1}\eta N_{tot}}}.\\
\end{split}
\end{equation}
Where to get the second line a triangle inequality and that in the high loss regime $$|\epsilon_{\text{hoeff},\tilde{p}_R^1(s^{n}_l)}| \geq \frac{1}{n_d}\sum_{i=1}^{n_d} |\epsilon_{\text{hoeff},\tilde{p}_R^k(s^{n}_l)}|$$ were used, and lemmas \ref{gradient_stat_error}, \ref{gradient_bias_error} and \ref{recycled_stat_error_bound}, along with error term $\epsilon_{bias}$ that is bounded according to the inequality $|\epsilon_{bias}| \in  O({m \choose n}^{-1/3})$ derived in Appendix \ref{gradient_bias_error}, were used to get the fourth line.
And so the condition to beat postselection is
\begin{equation*}
\begin{split}
& A\frac{n_d+1}{2}\bigg( \bigg(\frac{1}{(1-\eta)^{n}N_{tot}}\bigg)^{1/4} +   {m \choose n}^{-1/3}\bigg) + B\frac{n}{m-n+1}\sqrt{\frac{{m \choose n}}{n(1-\eta)^{n-1}\eta N_{tot}}} \leq C \sqrt{\frac{{m \choose n}}{(1-\eta)^n N_{tot}}}.
\end{split}
\end{equation*}

\end{proof}

We now provide a short proof of the following theorem from the main results section:

\textbf{Theorem \ref{thexpextrap}.}
\textit{
Assume that Conjectures \ref{conjbias} and \ref{conjstatextr} are true. Furthermore, assume that we are in the case where $|p_{mit}(s)-p_{id}(s)|=\epsilon_{stat}(m,n,k,\eta,N_{tot})+ \kappa(n) \frac{1}{{m \choose n}}$, with $\kappa(n)<1$, and $ |p_{post}(s)-p_{id}(s)|= \sqrt{\frac{1}{(1-\eta)^nN_{tot}}}$, where $p_{mit}(s)$ is the output probability of exponential extrapolation recycling mitigation. 
Then exponential extrapolation outperforms postselection for up to sample number $N_{tot} \ \in O\Big(\frac{{m \choose n}^2}{\kappa(n)^2 (1-\eta)^n}\Big)$ and up to additive error $\epsilon \in O\Big(\frac{\kappa(n)}{{m \choose n}}\Big)$.
}

\begin{proof}
Conjecture \ref{conjstatextr} being true implies that  $$\sqrt{\frac{1}{(1-\eta)^nN_{tot}}}-\epsilon_{stat}(m,n,k,\eta,N_{tot}) \geq c\sqrt{\frac{1}{(1-\eta)^nN_{tot}}},$$ for some $0<c<1$ for sufficiently large $n$. 
Thus, the condition that exponential extrapolation outperforms postselection, given by the following inequality holding $|p_{mit}-p_{id}(s)| \leq |p_{post}(s)-p_{id}(s)|$, is satisfied if $\kappa(n) \frac{1}{{m \choose n}} \leq c\sqrt{\frac{1}{(1-\eta)^nN_{tot}}}.$ 
The  bound on $N_{tot}$ in the Theorem immediately follows. 
Furthermore, the  bound on $\epsilon$ is directly obtained by replacing  $(1-\eta)^nN_{tot} \in O\Big(\frac{{m \choose n}^2}{\kappa(n)^2}\Big)$ in $\frac{1}{\sqrt{(1-\eta)^nN_{tot}}}$, then using Stirling's approximation.
    \end{proof}

\section{Deterministic upper bound}
\label{appbreaking}
\subsection{Proof of Thm.  \ref{exp_barrier_thm}}

We now prove Thm.  \ref{exp_barrier_thm} from the main text:

\vspace{1em}

\textit{For the  class of unitary matrices $U$ with submatrices $A$ such that $p_{max}:=\mathsf{max}_{s^n_l}(p(s^n_l))=|\mathsf{Per}(A)|^2$, and where these matrices $A$ satisfy $\frac{h^A_{\infty}}{\|A\|_2}\ll 1$, the bias error $\epsilon_{bias,s^n_l}$ is bounded}
\begin{equation*}
\epsilon_{bias,s^n_l} \in O(e^{-0.00002n}).
\end{equation*}

\begin{proof}
Let $p_{max}=\mathsf{max}_{s^n_l}(p(s^n_l))$, where $\max_{s^n_l}(.)$  denotes the maximum over the $n$--photon probability distribution.
It is immediate to see that the interference term satisfies
$I_{s^n_l,k} \leq p_{max}$. 
Furthermore,  $p_{max}=|\mathsf{Per}(A)|^2$, where $A:=(a_{ij})_{i,j \in \{1, \dots, n\}}$ is a submatrix of the linear optical unitary $U$ \cite{aaronson_bosonsampling_2014}. 
Note that $\|A\|_2 \leq \|U\|_2 \leq 1$, where $\|.\|_2$ denotes the spectral norm \cite{aaronson_generalizing_2012}.
Let $h^A_{\infty}:=\frac{1}{n} \sum_{i=1, \dots, n} \|\mathbf{A}_i\|_{\infty}$, where $\mathbf{A}_i$ is the $i$th row of $A$, and $\|\mathbf{A_i}\|_{\infty}:=\mathsf{max}_j(|a_{ij}|)$. 
Let $\|A\|_2 \leq T$. Using Thm.  2 in \cite{berkowitz_stability_2018}, we have that
\begin{equation*}
    |\mathsf{Per}(A)| \leq 2T^ne^{-0.00001(1-\frac{h^A_{\infty}}{\|A\|_2})^2n}.
\end{equation*}
Since $\|A\|_{2} \leq 1$, we immediately have
\begin{equation*}
|\mathsf{Per}(A)| \leq 2e^{-0.00001(1-\frac{h^A_{\infty}}{\|A\|_2})^2n}.
\end{equation*}
For classes of matrices where $\frac{h^A_{\infty}}{\|A\|_2}\ll 1$, we have that
\begin{equation*}
|\mathsf{Per}(A)| \leq 2e^{-0.00001n},
\end{equation*}
and therefore that
\begin{equation*}
p_{max} \leq  4e^{-0.00002n}.
\end{equation*}
Now, we have that
$ -p_{unif} \leq I_{s_l^n,k}-p_{unif} \leq p_{max}-p_{unif}$, and therefore
the bias error is given by 
$$\epsilon_{1}=|I_{s_l^n,k}-p_{unif} | \leq \mathsf{max}\{|p_{max}-p_{unif}|,p_{unif}\}.$$
In the case where $|p_{max}-p_{unif}| < p_{unif}$, then we have an exponentially small bias error $\epsilon_{bias,s^n_l} \leq p_{unif} \leq \frac{1}{{m \choose n}} \leq e^{-Cnlogn}$, for some $C>0$, 
when $m \in \Omega (n^{5.1})$ which is the boson sampling regime. 
Now, when $|p_{max}-p_{unif}| > p_{unif}$, we can use the bound for $p_{max}$ we obtained, and show that
$$\epsilon_{bias,s^n_l} \leq |p_{max}-p_{unif}| \leq  |4e^{-0.00002n}-e^{-Clog(n)}|,$$
and $|4e^{-0.00002n}-e^{-Clog(n)}| \in O(e^{-0.00002n})$.
\end{proof}

\section{ Richardson extrapolation methods for photon loss mitigation}
\label{appnogo}

In this section, we prove Thm.  \ref{thnogozne} as well as provide similar evidence (that the methods present no advantage over postselection)  for different methods of performing extrapolation with increasing rates of photon loss.

\subsection*{First method of extrapolation at various noise rates}
Let $m$ be the number of modes of a linear optical circuit which can implement any $m \times m$ unitary. Into this circuit we input $n$ photons in the first $n$ modes. The notation $|n_1,\dots,n_m\rangle$ denotes a state with $n_i$ photons in the $i^{th}$ mode, where $i \in \{1,\dots,m\}$. $\eta \in ]0,1[$ is the probability to lose a photon in any given mode, and is the same for all modes. 
We want to compute a specific marginal probability $p(n_1...n_l|n)$ 
with $l \leq m$, and $\sum_{i=1,\dots,l}n_i=c$, with $c \leq n$. The $|n$ indicates that we are computing the \emph{ideal} marginal probability, when no photon is lost.
Let $p(n_1 \dots n_l \cap j)$ be the probability of observing the output $(n_1,\dots,n_l)$ and detecting $j$ photons in all $m$ modes, with $j \in \{c,\dots,n\}$. When postselecting on no photons being lost, we are computing $$p(n_1 \dots n_l \cap n)=(1-\eta)^np(n_1 \dots n_l |n).$$ 
However, if we compute $p(n_1 \dots n_l)$ without caring  about whether no photon is lost, 
we end up computing 
\begin{equation}
\label{eq1}
    p_{\eta}(n_1...n_l)=\sum_{i=0, \dots, n-c}(1-\eta)^{n-i}\eta^{i}p(n_1 \dots n_l | n-i).
\end{equation}
Extrapolation techniques consist of estimating $p_{\eta}(n_1 \dots n_l)$ for different values of $\eta$, then deducing from these an estimate of $p(n_1 \dots n_l |n)$. One example of how this can be done is the Richardson extrapolation technique, at the heart of the zero noise extrapolation (ZNE) approach. Rather interestingly, we will show that these techniques offer no advantage over post-selection in terms of estimating $p(n_1 \dots n_l |n)$. 
Let $\alpha_i:=p(n_1 \dots n_l|n-i)$, we can then write
\begin{equation}
\label{eq2}
p_{\eta}(n_1 \dots n_l):=\sum_{i=0, \dots, n-c}(1-\eta)^{n-i}\eta^{i}\alpha_i.
\end{equation}
A natural way to estimate $\alpha_0$ through extrapolation would be to compute an estimate $\tilde{p}_{\eta}(n_1 \dots n_l)$ of $p_{\eta}(n_1 \dots n_l)$ for $n-c+1$ values $\eta_i$ of $\eta$, with $i \in \{0, \dots, n-c\}$, and (by convention) $\eta_{i+1} > \eta_i$. We will deal with additive error estimates, that is, $\tilde{p}_{\eta}(n_1 \dots n_l)=p_{\eta}(n_1 \dots n_l)+ \epsilon$, with $|\epsilon| \leq \epsilon_{max}$, and $\epsilon_{max} \in [0,1]$ is the additive error estimate. Note that an $\epsilon_{max}$ additive error estimate  $\tilde{p}_{\eta}(n_1 \dots n_l)$ can be obtained with high probability from $O(\frac{1}{\epsilon_{max}^2})$ runs of the boson sampling device, by Hoeffding's inequality \cite{hoeffding_collected_1994}. 

Performing the above mentioned extrapolation strategy, we obtain the following set of equations, written in matrix form, to be solved for obtaining an estimate of $\alpha_0$

\begin{multline}
    \label{eq3}
    \begin{pmatrix}
(1-\eta_0)^n  & \eta_0(1-\eta_0)^{n-1} & \dots & \eta_0^{n-c}(1-\eta_0)^k\\
. & . & \dots & . \\
. & . & \dots & . \\
(1-\eta_{n-c})^n  & \eta_{n-c}(1-\eta_{n-c})^{n-1} & \dots & \eta_{n-c}^{n-c}(1-\eta_{n-c})^c
\end{pmatrix} \begin{pmatrix}
    \tilde{\alpha}_0 \\
    .\\
    .\\
    \tilde{\alpha}_{n-c}
\end{pmatrix}
= \begin{pmatrix}
        p_{\eta_0}(n_1 \dots n_l) \\
        .
        \\
        . \\
        p_{\eta_{n-c}}(n_1 \dots n_l)
    \end{pmatrix}
    + \begin{pmatrix}
        \epsilon_0
     \\
     . \\
     . \\
     \epsilon_{n-c}
    \end{pmatrix},
\end{multline}
with $|\epsilon_{i}| \leq \epsilon_{max}$, and $\tilde{\alpha}_i$ is an estimate of $\alpha_{i}$ obtained from using the estimates $\tilde{p}_{\eta_i}(n_1 \dots n_l)=p_{\eta_i}(n_1 \dots n_l)+ \epsilon_i$
Let, 
$$L:=\begin{pmatrix}
(1-\eta_0)^n  & \eta_0(1-\eta_0)^{n-1} & \dots & \eta_0^{n-c}(1-\eta_0)^c\\
. & . & \dots & . \\
. & . & \dots & . \\
(1-\eta_{n-c})^n  & \eta_{n-c}(1-\eta_{n-c})^{n-1} & \dots & \eta_{n-c}^{n-c}(1-\eta_{n-k})^c
\end{pmatrix}.$$
We can rewrite $L$ as
\begin{equation}
    \label{eqL}
    L=DW,
\end{equation}
with 
\begin{equation}
    \label{eqD}
    D=\begin{pmatrix}
        (1-\eta_0)^n & 0 & 0 & \dots 0 \\
        0 & (1-\eta_1)^n& 0 & \dots 0 \\
        . & . & . & . \\
        . & . & . & . \\
        0 & 0 & 0 & \dots (1-\eta_{n-c})^n
    \end{pmatrix}
\end{equation}
a diagonal matrix, and
\begin{equation}
    \label{eqW}
    W=\begin{pmatrix}
        1 & \frac{\eta_0}{1-\eta_0} & \dots & (\frac{\eta_0}{1-\eta_0})^{n-c} \\
        1 & \frac{\eta_1}{1-\eta_1} & \dots & (\frac{\eta_1}{1-\eta_1})^{n-c} \\
        . & . & . & . \\
        . & . & . & . \\
        1 & \frac{\eta_{n-c}}{1-\eta_{n-c}} & \dots & (\frac{\eta_{n-c}}{1-\eta_{n-c}})^{n-c} \\
    \end{pmatrix} 
\end{equation}
a Vandermonde matrix \cite{gautschi_inverses_1978}. 
Multiplying both sides of eqn.  \ref{eq3} by $L^{-1}$ and invoking standard matrix multiplication rules, we obtain
\begin{equation}
    \label{eqextrap}
    \tilde{\alpha}_0=\sum_{i=0, n-c}L^{-1}_{1i}p_{\eta_i}(n_1 \dots n_l)+\sum_{i=0, \dots, n-c}L^{-1}_{1i}\epsilon_{i},
\end{equation}
 where $L^{-1}_{1i}$ is the element of the first row and $i$th column of $L^{-1}$. Note that, by construction of our method, $\eta_i \neq \eta_{i+1}$, then $D$ is invertible, and so is $W$ \cite{gautschi_inverses_1978}, thus $L^{-1}$ always exists. Furthermore, $\alpha_0=\sum_{i=0, n-k}L^{-1}_{1i}p_{\eta_i}(n_1 \dots n_l)$, since exactly computing the probabilities $p_{\eta_i}(n_1 \dots n_l)$ will lead to an exact computation of $\alpha_0$. Therefore, the error associated to our extrapolation technique is given by
 \begin{equation}
 \label{eqerrextr}
     E_{extrap}:=|\sum_{i=0, \dots, n-k}L^{-1}_{1i}\epsilon_{i}|.
 \end{equation}
Note that the overall sample complexity of the extrapolation protocol is $O(\frac{n-c+1}{\epsilon^2_{max}})$. 
For post-selection, an $\epsilon_{max}$ additive error estimate 
 $\tilde{p}_{\eta}(n_1 \dots n_l \cap n)$ requires $O(\frac{1}{\epsilon^2_{max}})$ samples, and can be performed for $\eta=\eta_0$, that is without artificially increasing loss.  From eqn.  \ref{eq1}, we see that post-selection induces an error of 
 \begin{equation}
     \label{eqerrpostsel}
     E_{post}:= |\frac{\epsilon}{\sqrt{(1-\eta_0)^n}}| \leq \frac{\epsilon_{max}}{\sqrt{(1-\eta_0)^n}},
 \end{equation}
with $|\epsilon| \leq \epsilon_{max}.$

In the remainder of this section, we will compute an upper bound for $E_{extrap}$, and show that this upper bound is greater than the corresponding upper bound for $E_{post}$ shown in eqn.  \ref{eqerrpostsel}. This gives strong evidence that the error induced by extrapolation is higher than that of postselection for a comparable number of samples, and therefore that extrapolation offers no advantage over postselection. 
Although the upper bound argument we show gives strong evidence that extrapolation techniques are not advantageous when compared to postselection, we will provide further evidence that this is the case. In particular, for a random distribution of errors $\{\epsilon_i\}$ with $|\epsilon_i| \leq \epsilon_{max}$, we numerically show that the condition 
\begin{equation}
    E_{extrap} > \frac{\epsilon_{max}}{(1-\eta_0)^n}
\end{equation}
is never violated after some value of $n$, confirming our analytical results.
We start by the analytical upper bound argument. By a triangle inequality,
\begin{equation*}
E_{extrap}\leq ||L^{-1}||_{\infty}\epsilon_{max},
\end{equation*}
with 
$||L^{-1}||_{\infty}:=\mathsf{max}_i\sum_j| L^{-1}_{ij}|$.
Also,
\begin{equation*}
||L^{-1}||_{\infty} \leq ||D^{-1}||_{\infty}||W^{-1}||_{\infty}.
\end{equation*}
From the definition of $||.||_{\infty}$, we can directly show
$$||D^{-1}||_{\infty}=\frac{1}{(1-\eta_{n-c})^n}.$$
For $||W^{-1}||_{\infty}$, we can use an upper bound on the norm of Vandermonde matrices shown in \cite{gautschi_inverses_1978}
$$||W^{-1}||_{\infty} \leq \mathsf{max}_i \prod_{ j \neq i = 0, \dots, n-c}\frac{1+\frac{\eta_i}{1-\eta_i}}{|\frac{\eta_i}{1-\eta_i}-\frac{\eta_j}{1-\eta_j}|}.$$
Thus,
$$||L^{-1}||_{\infty} \leq\frac{1}{(1-\eta_{n-c})^n}  \mathsf{max}_i \prod_{ j \neq i = 0, \dots, n-c}\frac{1+\frac{\eta_i}{1-\eta_i}}{|\frac{\eta_i}{1-\eta_i}-\frac{\eta_j}{1-\eta_j}|}.$$
and
$$E_{extrap} \leq \epsilon_{max}\frac{1}{(1-\eta_{n-c})^n}  \mathsf{max}_i \prod_{ j \neq i = 0, \dots, n-c}\frac{1+\frac{\eta_i}{1-\eta_i}}{|\frac{\eta_i}{1-\eta_i}-\frac{\eta_j}{1-\eta_j}|}.$$
We will now show that the upper bound on $E_{extrap}$ is greater than that of $E_{post}$. This is quantified through the following theorem .
\begin{theorem}
\label{th1}
  For $\eta_0, \dots, \eta_{n-c}$ with $\eta_{i+1} > \eta_{i}$ and $\eta_{0} \geq 0$, $\eta_{n-c} <1$ the following holds
  \begin{equation}
  \label{eqrtp}
      \frac{1}{(1-\eta_{n-c})^n}\mathsf{max}_i \prod_{ j \neq i = 0, \dots, n-c}\frac{1+\frac{\eta_i}{1-\eta_i}}{|\frac{\eta_i}{1-\eta_i}-\frac{\eta_j}{1-\eta_j}|} \geq \frac{1}{(1-\eta_0)^n}.
  \end{equation}
\end{theorem}
\begin{proof}
    First, note that $\frac{\eta}{1-\eta}$ is a monotonically increasing function of $\eta$, thus $\frac{\eta_{i}}{1-\eta_{i}} < \frac{\eta_{i+1}}{1-\eta_{i+1}}$. This allows us to lower bound $\mathsf{max}_i \prod_{ j \neq i = 0, \dots, n-k}\frac{1+\frac{\eta_i}{1-\eta_i}}{|\frac{\eta_i}{1-\eta_i}-\frac{\eta_j}{1-\eta_j}|}$ as
    \begin{equation*}
        \mathsf{max}_i \prod_{ j \neq i = 0, \dots, n-c}\frac{1+\frac{\eta_i}{1-\eta_i}}{|\frac{\eta_i}{1-\eta_i}-\frac{\eta_j}{1-\eta_j}|} \geq \frac{(1+\frac{\eta_0}{1-\eta_0})^{n-c}}{(\frac{\eta_{n-c}}{1-\eta_{n-c}}-\frac{\eta_{0}}{1-\eta_0})^{n-c}}.
    \end{equation*}
    Our strategy is to show that the following holds
    \begin{equation}
    \label{eqrtp2}
       \frac{1}{(1-\eta_{n-c})^n} \frac{(1+\frac{\eta_0}{1-\eta_0})^{n-c}}{(\frac{\eta_{n-c}}{1-\eta_{n-c}}-\frac{\eta_0}{1-\eta_0})^{n-c}} \geq \frac{1}{(1-\eta_0)^n}
    \end{equation}
    eqn.  \ref{eqrtp2} being true implies that eqn.  \ref{eqrtp} is also true, and thus is sufficient for proving Thm.  \ref{th1}.
eqn.  \ref{eqrtp2} can be rewritten as 
    \begin{equation}
    \label{eqrtp3}
    \frac{(1-\eta_0)^n}{(1-\eta_{n-c})^n} \frac{(1-\eta_0)^{n-c}}{(\frac{\eta_{n-c}}{1-\eta_{n-c}}-\frac{\eta_0}{1-\eta_0})^{n-c}} \geq 1.
    \end{equation}
    rewriting the left hand side of the above eqn. 
    $$\frac{(1-\eta_0)^n}{(1-\eta_{n-c})^n} \frac{(1+\frac{\eta_0}{1-\eta_0})^{n-c}}{(\frac{\eta_{n-c}}{1-\eta_{n-c}}-\frac{\eta_0}{1-\eta_0})^{n-c}}=\frac{(1-\eta_0)^c}{(1-\eta_{n-c})^c}\big(\frac{(1-\eta_0)(1+\frac{\eta_0}{1-\eta_0})}{(1-\eta_{n-c})(\frac{\eta_{n-c}}{1-\eta_{n-c}}-\frac{\eta_0}{1-\eta_0})}\big)^{n-c},$$
    and observing that 
    $$\frac{(1-\eta_0)(1+\frac{\eta_0}{1-\eta_0})}{(1-\eta_{n-c})(\frac{\eta_{n-c}}{1-\eta_{n-c}}-\frac{\eta_0}{1-\eta_0})}=\frac{1-\eta_0}{\eta_{n-c}-\eta_0},$$
    and plugging this into eqn.  (\ref{eqrtp3}) we obtain
    \begin{equation}
    \label{eqrtp4}
    \frac{(1-\eta_0)^c}{(1-\eta_{n-c})^c}(\frac{1-\eta_0}{\eta_{n-c}-\eta_0})^{n-c} \geq 1.
    \end{equation}
Now,
$1-\eta_0 \geq 1-\eta_{n-c}$ and $1-\eta_0 \geq \eta_{n-c}-\eta_0$ since $\eta_0 < \eta_{n-c}<1$. This implies that eqn.  \ref{eqrtp4} is true, and thus eqn. s \ref{eqrtp3} and \ref{eqrtp2} hold, and therefore  Thm.  \ref{th1} is proved.
    \end{proof}

\subsection*{Second method of extrapolation at various noise rates}
Another possible extrapolation technique can be performed by considering eqn.  (\ref{eq1}) where the $(1-\eta_i)^{n-i}$  terms are expanded  in order to obtain
\begin{equation}
\label{eqsecexpansion}
    p_{\eta}(n_1 \dots n_l)=\sum_{i=0, \dots,n}\beta_i\eta^i.
\end{equation}
Where $\beta_0=\alpha_0$, and $\beta_i$s are linear combinations of the $\alpha_i$s defined previously. The extrapolation procedure proceeds in a similar manner to that described above, but now we compute $p_{\eta}(n_1 \dots n_l)$ for $n+1$ values of loss $\eta_n>\eta_{n-1} \dots > \eta_0$ to solve for the coefficients $\{\beta_i \}$, with the matrix $L$ in this case  given by
\begin{equation}
    \label{eqWprt2}
    L=\begin{pmatrix}
        1 & \eta_0 & \dots & \eta^n_0 \\
        1 & \eta_1 & \dots & \eta^n_1 \\
        . & . & . & . \\
        . & . & . & . \\
        1 & \eta_{n} & \dots & \eta^n_{n} \\
    \end{pmatrix} .
\end{equation}
$L$ is a Vandermonde matrix, and we can directly use the result of \cite{gautschi_inverses_1978} to upper bound $||L^{-1}||_{\infty}$ as
\begin{equation*}
    ||L^{-1}||_{\infty} \leq \mathsf{max}_i \prod_{ j \neq i = 0, \dots, n}\frac{1+\eta_i}{|\eta_i-\eta_j|}.
\end{equation*}
Therefore, 
$$E_{extrap} \leq \epsilon_{max}\mathsf{max}_i \prod_{ j \neq i = 0, \dots, n}\frac{1+\eta_i}{|\eta_i-\eta_j|}$$
We will now prove that the upper bound on $E_{extrap}$ is larger than that of $E_{post}$ for this method of extrapolation, giving strong evidence that this technique offers no advantage over post-selection. This amounts to proving the following.
\begin{theorem}
    \label{th2}
    For $\eta_0, \dots, \eta_{n}$ with $\eta_{i+1} > \eta_{i}$ and $\eta_{0} \geq 0$, $\eta_{n} <1$ the following holds
  \begin{equation}
      \mathsf{max}_i \prod_{ j \neq i = 0, \dots, n}\frac{1+\eta_i}{|\eta_i-\eta_j|} \geq \frac{1}{(1-\eta_0)^n}.
  \end{equation}
\end{theorem}
\begin{proof}
    We note that $$\mathsf{max}_i \prod_{ j \neq i = 0, \dots, n}\frac{1+\eta_i}{|\eta_i-\eta_j|} \geq \frac{(1+\eta_0)^n}{(\eta_n-\eta_0)^n}, $$
and that  $\frac{(1+\eta_0)^n}{(\eta_n-\eta_0)^n} \geq \frac{1}{(1-\eta_0)^n}$ since $\frac{1}{\eta_n-\eta_0} > \frac{1}{1-\eta_0}$ and $1+\eta_0 > 1$. This completes the proof.  
\end{proof}

For this technique as well we numerically compute the number of violations of $E_{extrap}> \frac{\epsilon_{max}}{(1-\eta_0)^n}$ and plot these in Figure \ref{figerr2}. We took $\epsilon_{max}=0.01$, $\eta_0=0.01$, $\eta_{n}=0.95$, and $\eta_i$ for $0<i<n$ equally spaced. We varied the value of $n$
between 3 and 16.  For each value of $n$ we performed 3000 runs, where at each run we took  $n+1$ values of $\{\epsilon_i\}$
chosen uniformly randomly from $[-\epsilon_{max}, \epsilon_{max}]$. As can be observed in Figure \ref{figerr2}, the number of violations approaches zero with increasing $n$, confirming that this extrapolation performs worse than post-selection after some value of $n$. A similar behaviour is observed for different values of $\epsilon_{max}, \eta_0$ and $\eta_{n}$.
\begin{figure}[h]
\includegraphics[scale=0.3]{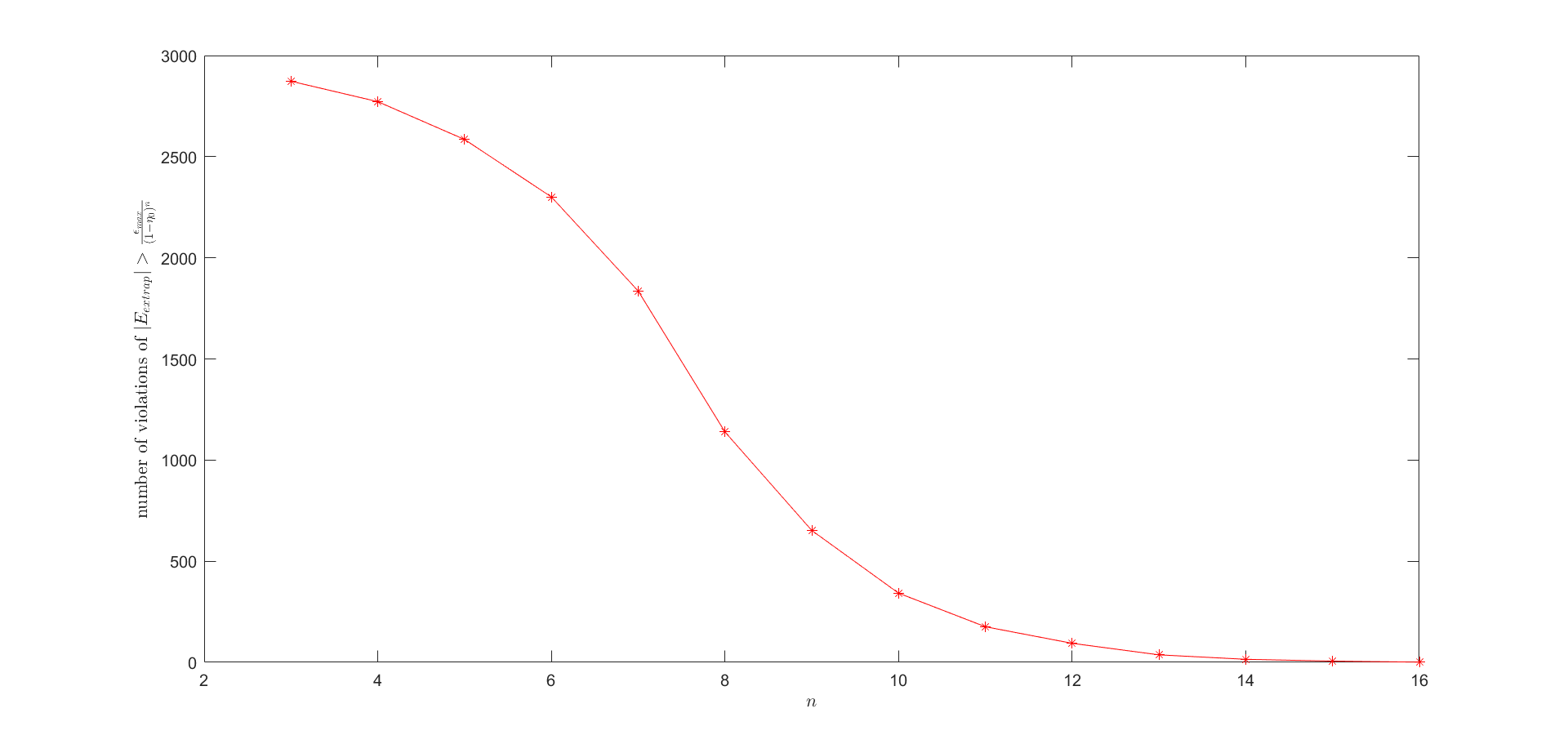}
\centering
\caption{Number of violations of $|E_{extrap}| \geq \frac{\epsilon_{max}}{(1-\eta_0)^n}$ plotted versus $n$ (see main text).}
\label{figerr2}
\end{figure}

\bigskip

\subsection*{Proof of Thm.  \ref{thnogozne}}

We now prove Thm.  \ref{thnogozne} from the main text:

\vspace{1em}

\textit{For all $n \geq n_0$, with $n_0$ a positive integer, $\mathsf{M}(E_{extrap}) \geq \frac{\epsilon_{max}}{\sqrt{(1-\eta)^n}}$.}

A third extrapolation technique which can be used is the Richardson extrapolation technique, at the heart of the zero noise extrapolation method for mitigating qubit errors, which has gained widespread use and is at the heart of ZNE \cite{endo_hybrid_2021}. This method uses the expansion of eqn.  (\ref{eqsecexpansion}) with $\eta_i=c_i\eta$, for $i \in \{0, \dots, n\}$ with $\eta \in [0,1]$, $c_0=1$, and $c_i$
are positive reals satisfying $c_{i+1}>c_i$. The Richardson extrapolation method consists of computing an estimate $\tilde{p}(n_1 \dots n_l |n)$ of $p(n_1 \dots n_l |n)$ as follows
\begin{equation*}
    \tilde{p}(n_1 \dots n_l|n)=\sum_{i=0 \dots n} \gamma_i \tilde{p}_{\eta_i}(n_1 \dots n_l),
\end{equation*}
where, $\gamma_i:=(-1)^n\prod_{j \neq i}\frac{c_j}{c_i-c_j}$. It can be shown that \cite{temme_error_2017} $\sum_{i=0, \dots,n} \gamma_i=1$ and $\sum_{j=0, \dots, n} \gamma_ic^j_i=0$ for $j=1, \dots, n$. Furthermore, it can also be shown that $p(n_1 \dots n_l|n)=\sum_{i=0, \dots n}\gamma_ip_{\eta_i}(n_1 \dots n_l)$. Thus, the error associated with this technique is given by
\begin{equation*}
    E_{extrap}:=|\sum_{j=0, \dots n}\gamma_i\epsilon_i|.
\end{equation*}
In the remainder of this paragraph, we will show that $\gamma_{i-1}=L^{-1}_{1i}$ for $i \in \{1, \dots, n+1 \}$ where $L$ is the Vandermonde matrix of eqn.  (\ref{eqWprt2}) with $\eta_i=c_i\eta$. This means that the results of the previous section follow through, and therefore that Richardson extrapolation offers no advantage over post-selection.

Let $\mathbf{B}:= \begin{pmatrix}
    \beta_0 \\
    \beta_1 \\
    . \\
    . \\
    .\\
    \beta_n
\end{pmatrix},$
and $\mathbf{P}=\begin{pmatrix}
    p_{\eta_0}(n_1 \dots n_l) \\
    p_{\eta_1} (n_1 \dots n_l) \\
    . \\
    . \\
    .\\
    p_{\eta_n}(n_1 \dots n_l)
\end{pmatrix}.$ 
In the case of a perfect computation of the probabilities $p_{\eta_i}(n_1 \dots n_l)$ (i.e $\epsilon_i=0$), the extrapolation technique based on the matrix $L$ of eqn.  (\ref{eqWprt2}) 
amounts to computing $\beta_0$
from the following system of eqn. s 
\begin{equation*}
    \mathbf{B}=L^{-1}\mathbf{P}.
\end{equation*}
Invoking standard multiplication rules for the above equation, one obtains
\begin{equation*}
    p(n_1 \dots n_l)=\beta_{0}=\sum_{i=1, \dots n+1}L^{-1}_{1i}p_{\eta_{i-1}}(n_1 \dots n_l).
\end{equation*}
Plugging the expansion $p_{\eta_j}(n_1 \dots n_l)=\sum_{i= 0 , \dots n} \beta_i \eta^j_i$ and $\eta_i=c_i\eta $ into the above equation, one obtains
\begin{equation*}
    \beta_0=\sum_{i=1, \dots n+1}L^{-1}_{1i}\beta_0+ \sum_{j=1 \dots, n} \sum_{i=1 \dots n+1}L^{-1}_{1i}c^j_{i-1}\beta_j\eta^j.
\end{equation*}
By a direct identification of the left hand side of the above equation with its right hand side, we find that
$\sum_{i=1 \dots n+1}L^{-1}_{1i}=1$, and $\sum_{i=1, \dots, n+1}L^{-1}_{1i}c^j_{i-1}=0$ for $j \in \{1 \dots n\}$. These sets of equations are exactly those defining the coefficients $\gamma_i$ (as seen previously) and allowing to uniquely determine them, thus we can make the identification $L^{-1}_{1i}=\gamma_{i-1}$ for $i \in \{1, \dots, n\}$ and our result is demonstrated.

\section{Lyapunov bound}
\label{applyapunov}
\subsection{Prospects of improving bounds on the interference terms}
\label{sec:prospects}

A natural question is whether one can tighten the  bounds, beyond what is guaranteed from Thm.  \ref{Haarbound}, on the error term $\epsilon_{bias,s^n_l}$ incurred by replacing the interference term $I_{s^n_l,k}$ with its average over the Haar measure.

In this section, we explore one attempt to tighten the bound on $\epsilon_{bias,s^n_l}$. We will work in the no-collision regime where $m\gg n^2$, so that we can approximate output probabilities as moduli squared of permanents of i.i.d. Gaussian matrices, with appropriate rescaling.
In this regime, the interference term $I_{s^n_l,k}$ can be thought of as a sum of $N'_k$ random variables 
\begin{equation}
\label{eqtighten1}
X_i:=\frac{|\mathsf{Per}(G_i)|^2}{m^n},
\end{equation}
 for $i \in \{1, \dots, N'_k\}$, where $G_i$ is an i.i.d. Gaussian $n \times n$ matrix with entries chosen independently from $\mathcal{N}_\mathbb{C}(0,1)$. Note that for  some $i$ and $j$, it is possible that $G_i$ or $G_j$ share some rows in common, or are even equal. This corresponds to the fact that the interference term is in general a sum of probabilities of output bit strings sharing some overlap (meaning their associated permanents have common rows \cite{aaronson_computational_2011}). This means that the random variables $X_i$ need not all be independent. Nevertheless,  we will look at a sum of independent random variables $X_i$ having the form of eqn.  (\ref{eqtighten1}). Our reason for working with an independent sum is that it simplifies the analysis, while  giving an intuition about what to expect in  the more general case of possibly dependent $X_i$'s. We will further comment on this point below.
 
 In the remainder of this section, we will show that  a $\mathsf{poly}(n)$ sized sum of independent $X_i$'s distributed according to eqn.  (\ref{eqtighten1}) \emph{does not} verify a sufficient condition, the Lyapunov condition  \cite{janson_central_2021}, for the central limit theorem (CLT) to hold as $n \to \infty$. Our proof relies on a conjecture of \cite{nezami_permanent_2021} on the expectation value $\mathbf{E}(X^{t}_i)$ over the set of Gaussian matrices $G_i$, where $t \in \mathbb{N}$ and $t>2$.   We also provide numerical evidence that the distribution of this sum is indeed not a normal distribution. Our result shows that it is non-trivial to improve the bound on $\epsilon_{bias,s^n_l}$ by trying to link $I_{s^n_l,k}$ to some known probability distribution, which was our initial motivation for trying to prove a CLT convergence result. 

 Let $Y_i$, $i \in \{1, \dots, N\}$ be real, independent and identically distributed, random variables satisfying $\mathbf{E}(Y_i)=0$, where $\mathbf{E}(.)$  denotes expectation value over the distribution of the $Y_i$'s. Let 
 \begin{equation}
     \label{eqtighten2}
     S_N:=\sum_{i=1, \dots, N}Y_i,
 \end{equation}
 and
 \begin{equation}
 \label{eqtighten3}
     \sigma_N:=\sqrt{\mathbf{E}(S^2_N)}.
 \end{equation}
 Furthermore, for any $r > 2$ and $r \in \mathbb{N}$ let
 \begin{equation}
     \label{eqtighten4}
     \mathbf{E}(|Y|^r):=\mathbf{E}(|Y_i|^r),
 \end{equation}
 for all $i \in \{1, \dots, N\}$.
 The Lyapunov condition can be stated in this case as \cite{janson_central_2021}
 \begin{theorem} {\normalfont (Lyapunov condition)}
 \label{thlyapunov}
 If for some fixed $r>2$, \begin{equation}
 \label{eqlyapunov}
     N\frac{\mathbf{E}(|Y|^r)}{\sigma_N^r} \to 0,
 \end{equation} then as $N \to \infty$
 \begin{equation}
     \frac{S_N}{\sigma_N} \to^{d} \mathcal{N}(0,1).
 \end{equation}
 Where $\to^{d}$ denotes convergence  of the distribution of $\frac{S_N}{\sigma_N}$.
 \end{theorem}

 For dependent identically distributed random variables, which correspond to the probabilities constituting the interference term, the Lyapunov condition becomes stricter to verify. In particular, for a specific type of dependence, $M(n)$-dependent random variables \cite{janson_central_2021},  the numerator in eqn.  (\ref{eqlyapunov}) is multiplied by $M(n)^{r-1}$,  where $M(n)>1$ and $M(n) \in \mathbb{N}$ is an integer whose value can depend on $n$ \cite{janson_central_2021}. One would expect, as mentioned earlier, that the non-convergence results established here for the case of independent random variables hold as well for the dependent case, although we do not formally prove this.

 A final ingredient we will use is the following conjecture appearing in \cite{nezami_permanent_2021} (Section 4.8, Conjecture 4.12), and whose truth is supported by  numerical simulations performed in \cite{nezami_permanent_2021}.
  \newtheorem{conjecture}{conjecture}
  \begin{conj}
  \label{conjperm}
     For $n,t > 2$ , $n,t \in \mathbb{N}$
     \begin{equation}
     \label{conjnezami}
     \mathbf{E}_{G \in \mathcal{G}_{n \times n}}(|\mathsf{Per}(G)|^{2t}) \in \Theta\bigg(\frac{(n!)^{2t}(t!)^{2n}}{(nt)!}\bigg).
     \end{equation}
  \end{conj}
 Where $\mathrm{E}_{G \in \mathcal{G}_{n \times n}}(.)$
is the expectation value over the set of $n \times n$
Gaussian matrices with entries from $\mathcal{N}_{\mathbb{C}}(0,1)$. 
For simplicity, we will henceforth denote $\mathbf{E}_{G \in \mathcal{G}_{n \times n}}(.)$ as $\mathbf{E}(.)$. Let 
\begin{equation}
\label{eqdefy}
Y_i:=X_i-\frac{n!}{m^n},
\end{equation}
where $X_i$ are as defined in eqn.  (\ref{eqtighten1}), and are independent, identically distributed random variables, $i \in \{1, \dots, N\}$. It is immediate to observe that $\mathbf{E}(Y_i)=0$. We  show the following.

\begin{theorem}
\label{thlyapunoviidgauss}
For all $r>2$, $n \gg 1$, and for the independent, identically distributed random variables $Y_i$ defined in eqn.  (\ref{eqdefy}), we have that
\begin{equation}
     N\frac{\mathbf{E}(|Y|^r)}{\sigma_N^r} \geq A\frac{\beta(r)^{n}}{N^{\frac{r}{2}-1}},
\end{equation}
where $\beta(r)>1$ is a positive real number dependent on $r$, and $A>0$ a constant.
\end{theorem}

When $N=\mathsf{poly}(n)$, Thm.  \ref{thlyapunoviidgauss} shows that Lyapunov condition is not satisfied. Although this condition is sufficient, but not necessary, for the CLT to hold we provide numerical evidence that $N=n^2$, $N=n^3$, and $N=n^4$ sized sums of i.i.d. Gaussian matrices do not converge to a normal distribution for $n \in \{2,3,4,5\}$. We plot our results  for $N=n^3$ and $n \in \{2,3,4,5\}$ in Figures \ref{fig:lyapunovfigs} (a)-(d).

It is interesting to note that in Thm.  \ref{thlyapunoviidgauss}, when $N\in O(\mathsf{exp}(n))$, the lower bound on the Lyapunov condition can converge to 0 as $n \to \infty$. Marginal probabilities corresponding to a large number of lost photons are sums of, possibly exponential, numbers of probabilities having the form of eqn.  (\ref{eqtighten1}). Our result provides  evidence that these marginals are asymptotically normally distributed, and therefore efficient to sample from. Although low order as well as high loss marginals of boson sampling are known to be easy to compute and sample from \cite{oszmaniec_classical_2018, villalonga_efficient_2022,aaronson_bosonsampling_2016}, our result might provide a new perspective on simulating lossy boson sampling marginals by linking these to normally distributed random variables.

\begin{figure}[]
\centering
\subfigure[]{
\includegraphics[width=.48\textwidth, height=5.05cm]{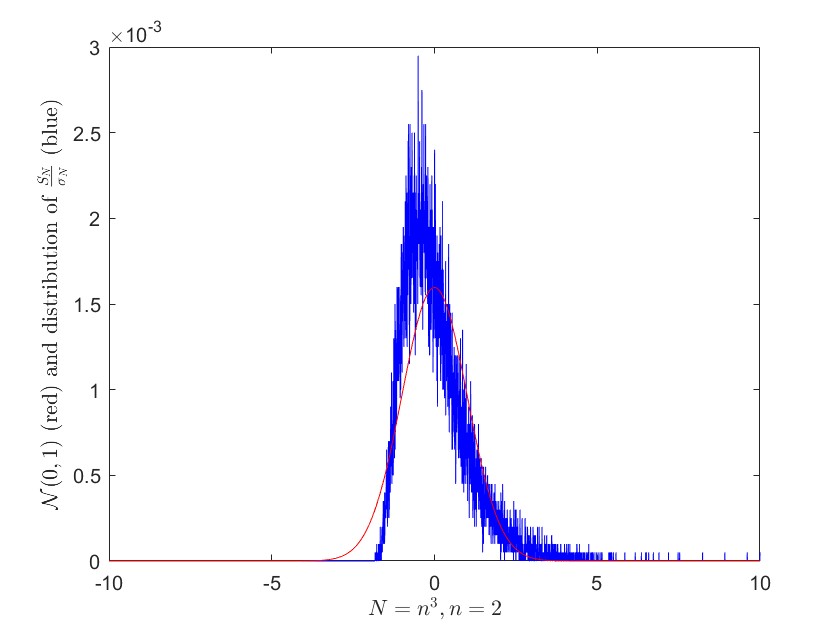}
}
\subfigure[]{
\includegraphics[width=.48\textwidth, height=4.95cm]{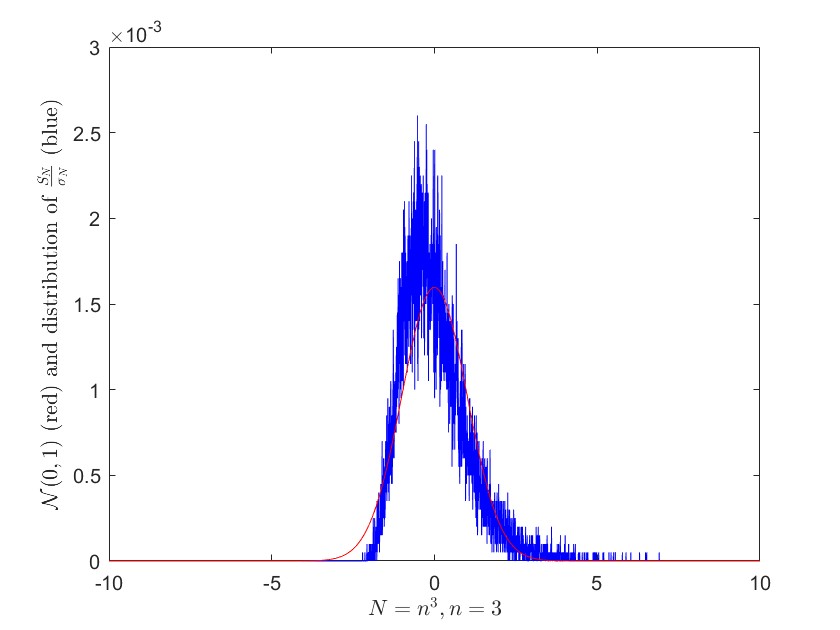}
}
\subfigure[]{
\includegraphics[width=.48\textwidth, height=5.05cm]{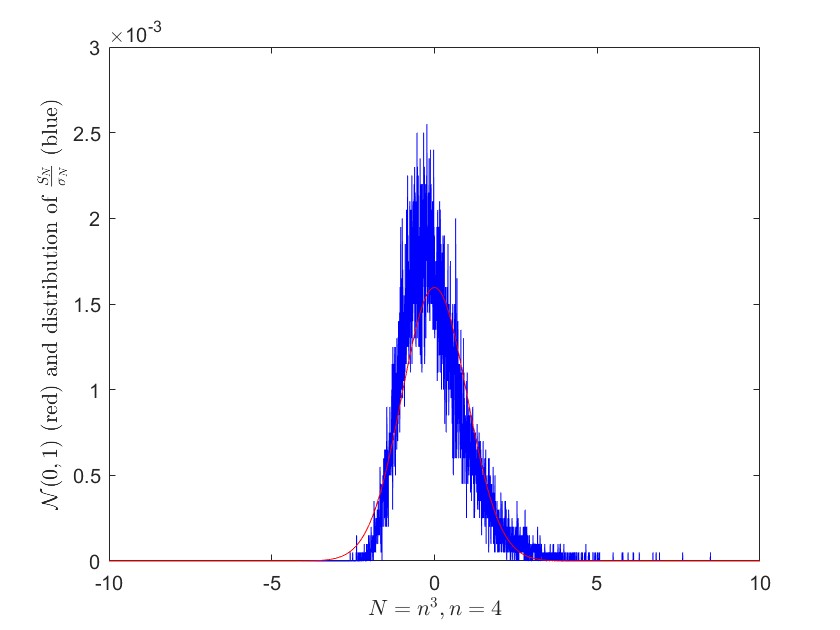}
}
\subfigure[]{
\includegraphics[width=.48\textwidth, height=5.05cm]{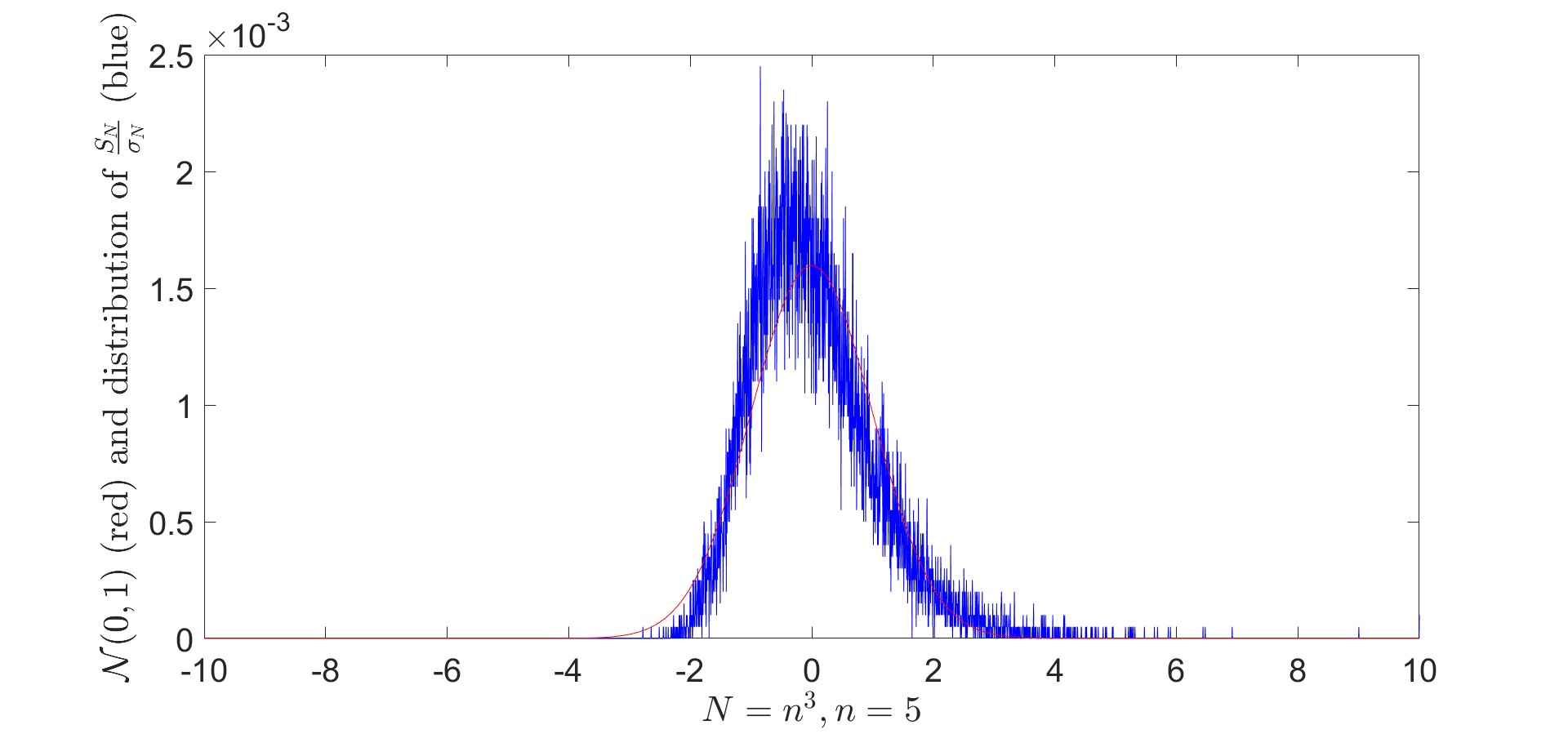}
}
\caption{(a)-(b) \textit{Comparison of distributions of $\frac{S_N}{\sigma_N}$ and the normal $\mathcal{N}(0,1)$ distribution.} (a) Distribution of $\frac{S_N}{\sigma_N}$ compared to the normal $\mathcal{N}(0,1)$ distribution. $n=2$, $N=n^3$, and 20000 samples of i.i.d. Gaussian matrices were used to construct the distribution of  $\frac{S_N}{\sigma_N}$. (b) Distribution of $\frac{S_N}{\sigma_N}$ compared to the normal $\mathcal{N}(0,1)$ distribution. $n=3$, $N=n^3$, and 20000 samples of i.i.d. Gaussian matrices were used to construct the distribution of  $\frac{S_N}{\sigma_N}$. (c) Distribution of $\frac{S_N}{\sigma_N}$ compared to the normal $\mathcal{N}(0,1)$ distribution. $n=4$, $N=n^3$, and 20000 samples of i.i.d. Gaussian matrices were used to construct the distribution of  $\frac{S_N}{\sigma_N}$. (d) Distribution of $\frac{S_N}{\sigma_N}$ compared to the normal $\mathcal{N}(0,1)$ distribution. $n=5$, $N=n^3$, and 20000 samples of i.i.d. Gaussian matrices were used to construct the distribution of  $\frac{S_N}{\sigma_N}$.}
\label{fig:lyapunovfigs}

\end{figure}
\subsection{Proof of Thm.  \ref{thlyapunoviidgauss}}

We now prove Thm.  \ref{thlyapunoviidgauss} .

\begin{proof}
Fix an $r>2$. We begin by evaluating $\sigma^r_N$.
    $$\sigma^2_N=(\mathbf{E}(S^2_N))=\mathbf{E}((\sum_{i=1, \dots, N}Y_i)^2)=\sum_{i=1,\dots,N}\mathbf{E}(Y_i^2)+2\sum_{i=1,\dots,N}\sum_{j>i}\mathbf{E}(Y_iY_j).$$
The $Y_i's$ are independent identically distributed, and $\mathbf{E}(Y_i)=0$ by construction, therefore $\mathbf{E}(Y_iY_j)=\mathbf{E}(Y_i)\mathbf{E}(Y_j)=0$, $\forall i \neq j$, and $\sum_{i=1,\dots,N}\mathbf{E}(Y_i^2)=N\mathbf{E}(Y^2)=N\mathbf{E}(X-\frac{n!}{m^n})^2=N\big(\mathbf{E}(X^2)-2\mathbf{E}(X)\frac{n!}{m^n}+(\frac{n!}{m^n})^2\big).$
Now, $\mathbf{E}(X^2)=\frac{(n!)^2(n+1)}{m^{2n}}$ and $\mathbf{E}(X)=\frac{n!}{m^n}$  \cite{nezami_permanent_2021}. Thus
$$\sigma^2_N=N(\frac{n(n!)^2}{m^{2n}}),$$ and then
$$\sigma^r_N=(N(\frac{n(n!)^2}{m^{2n}}))^{\frac{r}{2}}.$$
We will now compute a lower bound $\mathbf{E}(|Y|^r)$. First, we will use a triangle inequality,
\begin{equation*}
    \mathbf{E}(|Y|^r) \geq |\mathbf{E}(Y^r)|.
\end{equation*}
We now focus on computing $|\mathbf{E}(Y^r)|$. Let $\mu:=\frac{n!}{m^n}$
$$\mathbf{E}(Y^r)=\mathbf{E}(X-\mu)^r=\sum_{i=0, \dots, r}{r \choose i}\mathbf{E}(X^i)\mu^{r-i}(-1)^{r-i}=\sum_{i=0, \dots, r}{r \choose i}\frac{\mathbf{E}(|\mathsf{Per}(G)|^{2i})\mu^{r-i}(-1)^{r-i}}{m^{ni}}.$$
Plugging conjecture \ref{conjperm} into $\mathbf{E}(Y^r)$, we obtain
$$\mathbf{E}(Y^r)=\sum_{i=0, \dots, r}{r \choose i}\frac{A \frac{(n!)^{2i}(i!)^{2n}}{(ni)!}(\frac{n!}{m^n})^{r-i}(-1)^{r-i}}{m^{ni}}=\frac{1}{m^{nr}}\sum_{i=0,\dots,r}(-1)^{r-i} B_i\frac{{r \choose i}(n!)^{r+i}(i!)^{2n}}{(ni)!}, $$
for some constants  $A,B_i>0$ for $i \in \{0, \dots,r\}$. Since $n\gg 1$, we can use a Stirling approximation for $i>0$
$$\frac{(n!)^{r+i}}{(ni)!} \approx \frac{(\frac{n}{e})^{nr+ni}(\sqrt{2\pi n})^{r+i}}{(\frac{ni}{e})^{ni} \sqrt{2 \pi ni}}=\bigg(\frac{1}{i}\bigg)^{ni}\bigg(\frac{n}{e}\bigg)^{nr} \frac{(2 \pi n)^{\frac{i+r}{2}}}{\sqrt{2 \pi ni}}.$$
Plugging this into $\mathbf{E}(Y^r)$, we get
$$\mathbf{E}(Y^r)=\frac{(\frac{n}{e})^{nr}}{m^{nr}}\bigg((-1)^rA(2\pi n)^{\frac{r}{2}}+\sum_{i=1,\dots,r}(-1)^{r-i}B_i{r \choose i}(i!)^{2n}\frac{1}{i^{ni}}\frac{(2 \pi n)^{\frac{i+r}{2}}}{\sqrt{2 \pi ni}}\bigg).$$
For any $i \leq r$, $\frac{(2 \pi n)^{\frac{i+r}{2}}}{\sqrt{2 \pi ni}} \in O(n^r)$, we can simplify the above to
$$\mathbf{E}(Y^r)=A\frac{n^r(\frac{n}{e})^{nr}}{m^{nr}}\bigg(B(-1)^r+\sum_{i=1,\dots,r}(-1)^{r-i}C_i{r \choose i}(i!)^{2n}\frac{1}{i^{ni}}\bigg),$$
for some constants $A,B,C_i>0$.
Now, we will look at 
$|B(-1)^r+\sum_{i=1,\dots,r}(-1)^{r-i}C_i{r \choose i}(i!)^{2n}\frac{1}{i^{ni}}))|= |(-1)^rB+\sum_{i=1,\dots,r}(-1)^{r-i}C_i{r \choose i}(\frac{(i!)^{\frac{2}{i}}}{i})^{ni}|$. 
$\alpha(i):=\frac{(i!)^{\frac{2}{i}}}{i}$ is a monotonically increasing function of $i$ for $i>2$, this can be directly verified using Stirling's approximation for large $i$, since $\frac{(i!)^{\frac{2}{i}}}{i} \in \frac{O(i^2)}{i} \in O(i)$. For small $i$, we verified this by numerical simulation. Furthermore, $\alpha(i) \geq 1$ for $i \geq 1$. 
Therefore $|(-1)^rB+\sum_{i=1,\dots,r}(-1)^{r-i}C_i{r \choose i}(\frac{(i!)^{\frac{2}{i}}}{i})^{ni}| \in O((\alpha(r))^{nr}) \in O(\beta(r)^n),$ where $\beta(r):=\alpha(r)^r=\frac{(r!)^2}{r^r}.$ Note that $\beta(r)>1$ for $r>2$.
Therefore 
$$|\mathbf{E}(Y^r)| \in \frac{O(n^r(\frac{n}{e})^{nr}\beta(r)^n)}{m^{nr}}.$$
Replacing $n! \approx (\frac{n}{e})^n \sqrt{2 \pi n}$ in $\sigma^r_{N}$, then performing straightforward simplifications, we obtain 
$$\frac{N|\mathbf{E}(Y^r)|}{\sigma^r_N} \in O\bigg(\frac{\beta(r)^n}{N^{\frac{r}{2}-1}}\bigg).$$
Noting, as previously mentioned,  that $\frac{N\mathbf{E}(|Y|^r)}{\sigma^r_N} \geq \frac{N|\mathbf{E}(Y^r)|}{\sigma^r_N}$ completes the proof.
\end{proof}

\end{document}